\newif\ifllncs

\ifllncs
    \documentclass[runningheads,a4paper]{llncs} 
    \newcommand{\authnote}[3]{\textcolor{#3}{[{\footnotesize {#2} {\bf -- #1}}]}}
    \newcommand{\prabhanjan}[1]{\authnote{Prabhanjan}{}{blue}}
    \newcommand{\yaoting}[1]{\authnote{Yao-Ting}{}{red}}
    \newcommand{\aditya}[1]{\authnote{Aditya}{}{brown}}
    \newcommand{\oldtext}[1]{\authnote{Old}{}{red}}
\else
    \documentclass{article}
    \usepackage{fullpage}
    
    \usepackage{palatino} 
    \usepackage{mathpazo} 

\fi

\usepackage{diagbox}
\usepackage{amsmath,amsfonts,amssymb}
\usepackage{mathtools}

\ifllncs
\else
\usepackage{amsthm} 
\fi

\usepackage{circuitikz}
\usepackage{graphicx} 
\usepackage[T1]{fontenc}
\usepackage{physics}
\usepackage{todonotes}
\usepackage{mdframed} 
\usepackage{lipsum}
\usepackage{multicol}
\usepackage{bbm}
\usepackage{longfbox}
\usepackage{enumitem}
\usepackage{float} 
\usepackage{nicefrac} 
\usepackage{dsfont}
\usepackage{subcaption} 
\usepackage{bm} 
\usepackage{comment}
\usepackage[style=alphabetic,minalphanames=3,maxalphanames=4,maxnames=99,backref=true]{biblatex}
\usepackage{hyperref}
\usepackage[nameinlink]{cleveref}
\hypersetup{colorlinks={true},linkcolor={blue},citecolor=blue}

\allowdisplaybreaks

\ifllncs

\spnewtheorem{myclaim}[theorem]{Claim}{\bfseries}{\itshape}
\Crefname{myclaim}{Claim}{Claims}

\spnewtheorem{fact}[theorem]{Fact}{\bfseries}{\itshape}
\Crefname{fact}{Fact}{Facts}

\renewenvironment{proof}[1][Proof]
{\par\noindent\textit{#1. }}
{\hfill$\square$\par}

\else
\newtheorem{theorem}{Theorem}[section]
\newtheorem{proposition}[theorem]{Proposition}
\newtheorem{definition}[theorem]{Definition}
\newtheorem{lemma}[theorem]{Lemma}

\newtheorem{myclaim}[theorem]{Claim}

\newtheorem{corollary}[theorem]{Corollary}

\newtheorem{fact}[theorem]{Fact}
\Crefname{fact}{Fact}{Facts}
\fi

\newmdenv[
  linewidth=1pt,
  roundcorner=5pt,
  linecolor=black,
  innerleftmargin=10pt,
  innerrightmargin=10pt,
  innertopmargin=5pt,
  innerbottommargin=5pt
]{protocolbox}

\newcommand{\protocol}[2]{ 
  \hfill
  \begin{protocolbox}
    #1\textnormal{:} 
    #2
  \end{protocolbox}
}

\newcommand{\myparagraph}[1]{\vspace{.5em} \noindent \textbf{#1.}\,}

\newcommand{\resp}{resp.,\ }
\newcommand{\ie}{i.e.,\ }

\newcommand{\wrt} {with respect to\ }

\newcommand{\N}{\mathbb{N}}

\newcommand{\R}{\mathbb{R}}

\newcommand{\ceil}[1]{\lceil #1 \rceil}

\renewcommand{\bra}[1]{\langle#1\rvert}

\renewcommand{\ket}[1]{\lvert#1\rangle}
\newcommand{\set}[1]{\{ #1 \}}
\newcommand{\bit}{\{0,1\}}

\newcommand{\cA}{{\mathcal A}}

\newcommand{\cE}{{\mathcal E}}
\newcommand{\cF}{{\mathcal F}}

\newcommand{\cH}{{\mathcal H}}
\newcommand{\cI}{{\mathcal I}}

\newcommand{\cO}{{\mathcal O}}
\newcommand{\cP}{{\mathcal P}}
\newcommand{\cQ}{{\mathcal Q}}
\newcommand{\cR}{{\mathcal R}}

\newcommand{\cT}{{\mathcal T}}
\newcommand{\cU}{{\mathcal U}}
\newcommand{\cV}{{\mathcal V}}

\newcommand{\cX}{{\mathcal X}}
\newcommand{\cY}{{\mathcal Y}}

\newcommand{\bfD}{\mathbf{D}}

\newcommand{\bfP}{\mathbf{P}}
\newcommand{\bfQ}{\mathbf{Q}}

\newcommand{\bfT}{\mathbf{T}}

\newcommand{\bfa}{\mathbf{a}}

\newcommand{\bfk}{\mathbf{k}}

\newcommand{\sfA}{\mathsf{A}}
\newcommand{\sfB}{\mathsf{B}}
\newcommand{\sfC}{\mathsf{C}}

\newcommand{\sfE}{\mathsf{E}}
\newcommand{\sfF}{\mathsf{F}}
\newcommand{\sfG}{\mathsf{G}}
\newcommand{\sfH}{\mathsf{H}}

\newcommand{\sfR}{\mathsf{R}}

\newcommand{\eps}{\varepsilon}
\newcommand{\veps}{\varepsilon}

\newcommand{\secp}{{\lambda}}
\newcommand{\poly}{\mathsf{poly}}

\newcommand{\Exp}{\operatorname*{\mathbb{E}}}
\newcommand{\Ex}{\Exp}

\newcommand{\negl}{\mathsf{negl}}
\newcommand{\Supp}{\operatorname{Supp}}

\newcommand{\TD}{\mathrm{TD}}
\newcommand{\SD}{\mathrm{SD}}

\newcommand{\commit}{{\mathsf{Commit}}}
\newcommand{\reveal}{{\mathsf{Reveal}}}

\newcommand{\C}{\mathbb{C}}
\newcommand{\Haar}{\mathcal{H}}
\newcommand{\hilbert}{\cH}
\newcommand{\sym}{\mathsf{Sym}}
\renewcommand{\ketbra}[2]{\ket{#1}\bra{#2}}

\newcommand{\inner}[2]{\langle #1, #2 \rangle}

\newcommand{\bad}{\mathsf{Bad}}

\newcommand{\alice}{\sfA}
\newcommand{\bob}{\sfB}
\newcommand{\challenger}{\sfC}
\newcommand{\eve}{\sfE}

\newcommand{\wt}{\widetilde}

\newcommand{\adv}{\mathsf{Adv}}

\newcommand{\Symgp}{\mathrm{Sym}}
\newcommand{\Unitary}{\mathrm{U}}

\renewcommand{\bra}[1]{\langle#1\rvert}

\renewcommand{\ket}[1]{\lvert#1\rangle}
\renewcommand{\ketbra}[2]{\ket{#1}\!\bra{#2}}
\newcommand{\projector}[1]{\ket{#1}\!\bra{#1}}

\newcommand{\secparam}{\lambda}

\newcommand{\id}{\mathrm{id}}

\newcommand{\haarstates}{\Haar}
\newcommand{\haarunitaries}{\mu}

\newcommand{\Adversary}{{\cal A}}

\newcommand{\dist}{\mathsf{dist}}
\newcommand{\num}{\mathsf{num}}

\newcommand{\reg}[1]{{\color{brown} \mathbf{#1}}}
\newcommand{\regA}{{\color{brown} \mathbf{A}}}
\newcommand{\regB}{{\color{brown} \mathbf{B}}}
\newcommand{\regC}{{\color{brown} \mathbf{C}}}

\newcommand{\regL}{{\color{brown} \mathbf{L}}}

\newcommand{\regR}{{\color{brown} \mathbf{R}}}

\newcommand{\regX}{{\color{brown} \mathbf{X}}}
\newcommand{\regY}{{\color{brown} \mathbf{Y}}}

\newcommand{\Dom}{\operatorname{Dom}}

\newcommand{\ptrans}{\Theta}

\newcommand{\newhybrid}[1]{\vspace{.5em} \noindent \textbullet~\textbf{Hybrid~#1:}}
\newcommand{\hybrid}[1]{\textbf{Hybrid~#1}}

\newcommand{\HUD}{\mathsf{HUD}}
\newcommand{\NAInvHUD}{\mathsf{NA\text{-}Inv\text{-}HUD}}
\newcommand{\NIInvHUD}{\mathsf{NI\text{-}Inv\text{-}HUD}}

\newcommand{\sen}{\mathsf{Sen}}
\newcommand{\rec}{\mathsf{Rec}}

\title{On the Limitations of Pseudorandom Unitaries\footnote{A preliminary version of this paper appears in the proceedings of \textit{the 23rd Theory of Cryptography Conference} (\textsc{TCC~2025}). This is the full version.}\\
{\em {\small Or: Cryptographic Applications of LOCC indistinguishability of identical versus independent Haar unitaries}}}

\index{Ananth, Prabhanjan}
\index{Gulati, Aditya}
\index{Lin, Yao-Ting}

\ifllncs
    \author{Prabhanjan Ananth\inst{1} \and Aditya Gulati\inst{1} \and Yao-Ting Lin\inst{1}}
    \authorrunning{P. Ananth et al.}
    \titlerunning{On the Limitations of PRUs}
    \institute{University of California, Santa Barbara, CA, USA \\
    \email{prabhanjan@cs.ucsb.edu} \;
    \email{adityagulati@ucsb.edu} \;
    \email{yao-ting\_lin@ucsb.edu}}
\else
    \author{Prabhanjan Ananth\thanks{\texttt{prabhanjan@cs.ucsb.edu}}\\ \small{UCSB} \and Aditya Gulati\thanks{\texttt{adityagulati@ucsb.edu}}\\ \small{UCSB} \and Yao-Ting Lin\thanks{\texttt{yao-ting\_lin@ucsb.edu}}\\ \small{UCSB}}
    \date{}
\fi

\begin{document}

\maketitle

\begin{abstract}
\noindent Pseudorandom unitaries (PRUs), one of the key quantum pseudorandom notions, are efficiently computable unitaries that are computationally indistinguishable from Haar random unitaries. While there is evidence to believe that PRUs are weaker than one-way functions, so far its relationship with other quantum cryptographic primitives (that are plausibly weaker than one-way functions) has not been fully established. 

\par In this work, we focus on quantum cryptographic primitives with classical communication, referred to as QCCC primitives. Our main result shows that QCCC bit commitments and QCCC key agreement, cannot be constructed from pseudorandom unitaries in a black-box manner.

\par Our core technical contribution is to show (in a variety of settings) the difficulty of distinguishing identical versus independent Haar unitaries by separable channels. Our result strictly improves upon prior works which studied similar problems in the context of learning theory [Anshu, Landau, Liu, STOC 2022] and cryptography [Ananth, Gulati, Lin, TCC 2024]. 
\end{abstract}

\ifllncs
    
\else
\newpage 
    \tableofcontents
\newpage 
\fi

\section{Introduction}

\noindent The governing principles of quantum information presents new opportunities for cryptography. The seminal work of~\cite{BB84} first opened up the possibility of utilizing quantum resources to weaken the computational assumptions needed for cryptographic tasks. Notably, in the last few years, extensive efforts have been devoted to exploit quantum resources to weaken the assumptions needed behind traditional cryptographic tasks. These include designing quantum bit commitments~\cite{BartusekCKM21a,GLSV21,MY22,AQY21,BCQ23,BEMPQY23,behera2023pseudorandomness,KT24,BJ24}, secret-key encryption~\cite{AQY21}, digital signatures~\cite{MY22}, public-key encryption~\cite{BGH+23,Col23,KMNY24,MW24} and message authentication codes~\cite{AQY21} from a variety of quantum cryptographic primitives, such as pseudorandom (function-like) states~\cite{JLS18,AQY21}, one-way state generators~\cite{MY22,}, EFI pairs~\cite{BCQ23}, one-way puzzles~\cite{KT24} and so on. The surprising result of Kretschmer~\cite{Kretschmer21} opened up the possibility that these quantum cryptographic primitives -- referred to as {\em Microcrypt} primitives -- can exist even if one-way functions don't. One particular Microcrypt primitive that is of interest to us is {\em pseudorandom unitaries} (PRUs). Interestingly, pseudorandom unitaries implies all of the above aforementioned primitives, making it one of the strongest primitives in Microcrypt. 

\myparagraph{Pseudorandom Unitaries} Formally, a pseudorandom unitary, introduced by~\cite{JLS18}, is associated with a family of  keyed polynomial time quantum circuits $\left\{G(k,\cdot)\}_{k \in \{0,1\}^{\secparam}}\right\}_{\secparam \in \mathbb{N}}$ such that $G(k,\cdot)$ implements an $n$-qubit unitary and moreover, any quantum polynomial time algorithm cannot distinguish whether it has oracle access to $G(k,\cdot)$, where $k$ is sampled at random, or an $n$-qubit Haar random unitary. Recent works~\cite{MPSY24,CDXBBH24,MH25} showed that PRUs can be based on the existence of one-way functions. PRUs, and more generally, Haar random unitaries, have applications to different areas of sciences, including randomized benchmarking~\cite{MGE12}, learning theory~\cite{ZLKQHC24}, complexity theory~\cite{CLS25}, random matrix theory~\cite{Meckes19} and black-hole physics~\cite{HP07}. However, the impact of PRUs on quantum cryptography is yet to be thoroughly investigated. While some cryptographic applications of pseudorandom unitaries have been explored recently~\cite{LQSYZ24,AGKL}, the relationship between PRUs and many core quantum cryptographic primitives remains to be much explored. We focus on a class of core primitives that we discuss next. 

\myparagraph{Quantum Cryptography with Classical Communication} Typically, quantum cryptographic primitives require the existence of quantum communication channels. This means that for quantum cryptography to be realized in practice, not only do we require quantum computers to exist but we also additionally require the infrastructure of quantum internet to be set up. A relatively more desirable scenario is when we only require classical communication channels to carry out quantum cryptographic tasks. The class of quantum cryptographic primitives that only require classical communication to be QCCC primitives. Examples in this class include quantum key agreement, quantum commitments, quantum (secret-key and public-key) encryption, quantum digital signatures and so on. Recent works~\cite{behera2023pseudorandomness,ALY24,CGG24,AGL24,barhoush2024signatures,KQT25,GMMY24} have focused on characterizing the complexity of QCCC primitives as well as proposing feasibility results. Precisely characterizing the minimal assumptions behind the existence of QCCC primitives is still an active and ongoing research direction. Concretely, we initiate the research direction of understanding the relationship between pseudorandom unitaries and QCCC primitives. 

\subsection{Our Results}
We show a variety of black-box separations between pseudorandom unitaries and QCCC primitives.

\myparagraph{Separations Between PRUs and QCCC primitives} We give a general framework to show separations between PRUs and QCCC primitives. 
\par First, we consider QCCC bit commitments. This is an interactive bit commitment scheme, where the sender and receiver are allowed to be quantum polynomial time algorithms but the communication is restricted to be classical. We summarise the prior work on QCCC bit commitments:
\begin{itemize}
    \item {\em Feasibility}: It is known that QCCC bit commitments, with statistical binding and computational hiding, can be realized from pseudorandom state generators with short output length~\cite{AGQY22,ALY24} or from a special variant of pseudorandom function-like state generators~\cite{behera2023pseudorandomness}.
    \item {\em Limitations}:~\cite{CLM23} presented evidence that non-interactive commitments cannot be constructed from one-way functions. Recently, a black-box separation between QCCC commitments and pseudorandom state generators (or even pseudorandom function-like state generators) with long output length was recently established~\cite{AGL24}.
\end{itemize}

\noindent In particular, the above black-box separation did not address the (im)possibility of basing QCCC commitments on PRUs, which imply the existence of both pseudorandom state generators and pseudorandom function-like state generators. We address this problem below.  

\begin{theorem}[Informal]
There is a black-box separation between pseudorandom unitaries and interactive quantum bit commitments with classical communication (QCCC commitments). 
\end{theorem}

\noindent Conceptually speaking, the above separation is the best one could hope for, in terms of separating QCCC commitments, since most Microcrypt primitives are implied by PRUs and moreover, QCCC commitments can be constructed from one-way functions~\cite{Naor89CRTPYO}. \\ 

\noindent Another QCCC primitive we address is QCCC key agreement. This is a type of quantum key agreement protocol where the communication is required to be classical. While it is well known that quantum key agreement protocol (with no restriction on the type of communication) is even information-theoretically possible~\cite{BB84}, precisely characterizing the minimal assumptions for achieving classical communication has been an ongoing research direction. We summarize the state of the art below:
\begin{itemize}
    \item {\em Feasibility}: unlike the case of QCCC commitments, it is currently unknown whether we can base QCCC key agreement on weaker assumptions than classical key agreement.   
    \item {\em Limitations}: Recent works~\cite{ACC+22,LLLL24,LLLL25}, presented evidence on the difficulty of constructing quantum key agreement with classical communication from one-way functions. These works fall short of showing a full fledged black-box separation. Another recent work~\cite{AGL24} showed a separation between QCCC commitments and pseudorandom function-like state generators.
\end{itemize}

\noindent We show the following. 

\begin{theorem}[Informal]
\label{thm:main:intro}
There is a black-box separation between pseudorandom unitaries and interactive quantum key agreement with classical communication (QCCC key agreements). 
\end{theorem}

\noindent Unlike~\cite{ACC+22,LLLL25}, our result does not impose restrictions on the number of rounds or the perfect completeness of the key agreements, nor does it rely on any unproven conjectures. Concretely,~\cite{ACC+22} present a separation between (interactive) QCCC key agreements and one-way functions based on a conjecture and~\cite{LLLL25} present a separation between {\em 2-round} key agreement and one-way functions (without assuming any conjecture). We instead separate (interactive) QCCC key agreements and PRUs (without assuming any conjecture). 

\myparagraph{On Strong PRUs and QCCC primitives} We also consider the notion of strong pseudorandom unitaries (strong PRUs)~\cite{MH25} and explore its relationship with QCCC primitives. A strong pseudorandom unitary is associated with a family of  keyed polynomial time quantum circuits $\left\{G(k,\cdot)\}_{k \in \{0,1\}^{\secparam}}\right\}_{\secparam \in \mathbb{N}}$ such that $G(k,\cdot)$ implements an $n$-qubit unitary and moreover, any quantum polynomial time algorithm cannot distinguish whether it has oracle access to $G(k,\cdot),G(k,\cdot)^{\dagger}$, where $k$ is sampled at random, or $(U,U^{\dagger})$, where $U$ is an $n$-qubit Haar random unitary. Clearly, strong PRUs imply the existence of PRUs. 

We rule out a class of constructions of QCCC primitives from strong PRUs. As before, we consider the case of commitments and key agreement.  
\par Consider the following class of constructions of QCCC commitments that use a strong PRU as a black box: 
\begin{itemize}
    \item At the beginning of the protocol, the sender and the receiver makes non-adaptive (or parallel) calls to the strong PRU, 
    \item After that, both the sender and receiver never query the strong PRU for the rest of the protocol.
\end{itemize}
\noindent This class of QCCC commitments (i.e., making only non-adaptive calls to strong PRUs) do not exist.\ifllncs\else\footnote{Interestingly, while the requirement on non-adaptivity might come across as quite restrictive, we do know how to achieve commitment schemes with {\em quantum communication} from PRUs making only non-adaptive queries. The high level description of the scheme is as follows:
\begin{itemize}
    \item In the \textbf{query phase}, the committer and receiver create the state $$\frac{1}{2^{\secp/2}}\sum_k U_k\ket{0^n}_{\reg{C}} \ket{k}_{\reg{R}},$$ where $U_k$ is the PRU corresponding to the key $k$ and $n>\secp$. 
    \item In the \textbf{commitment phase}, the committer sets the following state as: $$\ket{\psi_0}_{\reg{CR}}=\frac{1}{2^{\secp/2}}\sum_k U_k\ket{0^n}_{\reg{C}} \ket{k}_{\reg{R}},\qquad and \qquad\ket{\psi_1}_{\reg{CR}}=\frac{1}{2^{\secp/2}}\sum_k \ket{k}_{\reg{C}} \ket{k}_{\reg{R}}.$$ To commit to bit $b$, the committer sends register $\reg{C}$ to the receiver. 
    \item In the \textbf{reveal phase}, the committer sends $b,\reg{R}$ to the receiver. If $b=0$, the receiver does a swap test between $\reg{CR}$ and their copy of $$\frac{1}{2^{\secp/2}}\sum_k U_k\ket{0^n} \ket{k},$$ which created in the query phase. Otherwise, if $b=1$, the receiver does a swap test between $\reg{CR}$ and $$\frac{1}{2^{\secp/2}}\sum_k \ket{k}\ket{k}.$$ Accept the bit $b$ is the swap test passes. 
\end{itemize}}
\fi
Similar conclusions can be reached for key agreements as well. The proof template of the impossibility result is quite similar to the proof of~\Cref{thm:main:intro} and hence, we omit the proof details. 

\myparagraph{Central Results: LOCC Indistinguishability Frameworks for Unitaries} Central to proving our results is to establishing the so-called LOCC indistinguishability frameworks.\footnote{Here, LOCC refers to \emph{local operations and classical communication}, a well-studied notion in quantum information. It is very similar to QCCC, with the following nuance: LOCC is usually used in an information-theoretic setting, where parties may have unbounded running time and no computational assumptions are made. By contrast, QCCC is usually used for cryptographic primitives.} More broadly, LOCC indistinguishability, often referred to as data hiding in the quantum information community, and also related to distributed quantum property testing in the learning theory community, is a well-studied topic. Nonetheless, existing techniques do not directly apply to our setting. This is partly because, to the best of our knowledge, our work is the \textit{first} to study this problem in the \textit{oracle} setting. In contrast, all prior work considers the \textit{state} setting, where two parties are given quantum states. Much like the gap between constructing PRSs and PRUs, transitioning to the oracle setting introduces significant challenges. We consider two flavors of LOCC indistinguishability frameworks. 
\par We describe the first setting, referred to as \emph{Haar unitary distinguishing game} ($\mathsf{HUD}(N, \alice, \bob)$). There are two parties $\alice$ and $\bob$. $\alice$ receives oracle access to a unitary $U$ and $\bob$ receives oracle access to a unitary $V$. The unitaries $U$ and $V$ are generated as follows: sample a bit $b$ at random. 
\begin{itemize}
    \item If $b=0$: sample $U$ from the Haar distribution on $N$-dimensional unitaries. Set $V=U$. 
    \item If $b=1$: sample $U$ and $V$ {\em independently} from the Haar distribution on $N$-dimensional unitaries. 
\end{itemize}
\par Both  $\alice$ and $\bob$ can perform an arbitrary finite number of rounds of classical communication, during which they can make queries adaptively. In the end, $\alice$ outputs a bit $b'$. The output of the experiment $\mathsf{HUD}(N, \alice, \bob)$ is 1 if $b'=b$. 

\begin{theorem}[Informal]
The probability that $\mathsf{HUD}(N, \alice, \bob)$ outputs 1 is at most $\frac{1}{2} + O\left( \frac{t^2}{N} \right)$, where $t$ is the number of queries made by $\alice$ and $\bob$. 
\end{theorem}

\noindent The above theorem is formally stated in~\Cref{thm:adap_LOCC}. This improves upon prior works who study this problem in the context of showing separations~\cite{AGL24} and learning theory~\cite{ALL22,GHYZ24,AS24}. Establishing the indistinguishability of different states against LOCC adversaries has a rich literature in quantum information theory (under the umbrella of data hiding schemes)~\cite{BDF+99, DLT02, EW02, GB02, HLS05, MWW09, CLMO13, PNC14, CH14, CLMOW14, HBAB19,Har23}. 

\noindent We now describe the second setting, referred to as \emph{Non-Adaptive Haar unitary distinguishing game} ($\mathsf{NA}-\mathsf{HUD}(N, \alice, \bob)$). This game is similar to $\mathsf{HUD}(N, \alice, \bob)$ except in the following way:
\begin{itemize}
    \item $\alice$ and $\bob$ have additionally access to both $U^{\dagger}$ and $V^{\dagger}$,
    \item $\alice$ and $\bob$ can only make non-adaptive calls to $U^{\dagger}$ and $V^{\dagger}$. 
\end{itemize}
\noindent We show the following: 

\begin{theorem}[Informal]
The probability that $\mathsf{NA}-\mathsf{HUD}(N, \alice, \bob)$ outputs 1 is at most $\frac{1}{2} + O\left( \frac{t^2}{N^{1/8}} \right)$, where $t$ is the maximum number of non-adaptive queries made by $\alice$ or $\bob$.
\end{theorem}

\noindent We also consider another setting, where $\alice$ sends just one message to $\bob$. However, $\alice$ (resp., $\bob$) has access to $(U,U^{\dagger})$ (resp., $(V,V^{\dagger})$) and moreover, they can query both the oracle adaptively. A recent work~\cite{GMMY24} implicitly studies this setting. 

\paragraph{Inverse Adaptive Queries: Obstacles.} Proving LOCC indistinguishability in the presence of inverse adaptive queries remains an open problem. We explain briefly the obstacles to obtaining LOCC indistinguishability for the generalised case. First, existing techniques from state-based settings do not extend to our oracle model. The two natural approaches to attempt the generalized setting are using approximation formulas (like the twirling formula~\cite{HY24,SHH25}) and using purification (like path-recording framework~\cite{MH25}). The approximation idea relies on using a multiplicative approximation for Haar twirling. For forward-only queries, one can obtain such a result using the Weingarten calculus~\cite{HY24,SHH25}; however, no analogous closed-form or approximation is known for the mixed twirling case involving both $U$ and $U^\dagger$, making the inverse-query setting substantially harder. For the purification route, it is unclear whether purification techniques can help at all, once we purify the interaction with a common unitary, both parties become entangled with the shared purification register, making it difficult to analyze the protocol within the LOCC framework. Moreover, as far we know, path-recording only gives us an additive approximation. We believe that significantly new ideas are needed to tackle the inverse adaptive setting. 

\paragraph{Related and Concurrent Work on LOCC Indistinguishability:} Prior work by~\cite{AGL24,ALL22} studied the LOCC indistinguishability of Haar-random states. Specifically, they showed that two parties restricted to LOCC protocols cannot distinguish whether they share identical Haar states or independently sampled ones. In a concurrent and independent work~\cite{GZ25} consider an LOCC indistinguishability game for Haar swap unitaries\footnote{Haar swap unitaries are the unitary that swaps $\ket{0}$ with a Haar-random state $\ket{\psi}$.}. Specifically, they showed that two parties restricted to LOCC protocols cannot distinguish whether they share oracle access to identical Haar swap unitaries or independently sampled ones. Their techniques are similar in spirit to those used in~\cite{AGL24}. While the result in~\cite{GZ25} is incomparable to ours, we note that indistinguishability of Haar swap unitaries is arguably a weaker task than that of Haar unitaries, since swap unitaries capture the power of Haar states, whereas Haar unitaries are believed to be strictly stronger.

\section{Technical Overview}
\noindent We first discuss the main ideas behind the central results (i.e. Haar indistinguishability frameworks) in the paper before discussing their applications to proving separations. 
\par All of our Haar indistinguishability frameworks follow a similar template that we describe below: 
\begin{itemize}
    \item {\bf Step 1}: in the first step, we express the probability of the output of the experiment in both the identical (i.e., both $\alice$ and $\bob$ receive access to the same oracle) and the independent (i.e., both $\alice$ and $\bob$ receive access to the independent oracles) settings in terms of permutation operators. A permutation operator of the form $P(\pi)$, for $\pi:[t] \rightarrow [t]$ is an $nt$-qubit unitary that permutes the $t$ blocks (each of size $n$) according to the permutation $\pi$. 
    \item {\bf Step 2}: in the second step, we express the probability that $\alice$ outputs an outcome $b$ in the identical experiment in terms of {\em multiplicative factor} of the probability that $\alice$ outputs $b$ in the independent experiment.
    \item {\bf Step 3}: once we obtain the relations between the probabilities in both the experiments, we then solve for an optimization problem to complete the proof.
\end{itemize}
\noindent The crux of the proof of our Haar indistinguishability lemmas lies in Step 2. For the warmup case (\Cref{sec:warmup:state}), we can directly rely upon existing results whereas for the more general cases (\Cref{sec:hud},~\Cref{sec:nahud}), we need to make non-trivial observations using Haar twirling approximation formula~\cite{SHH25}.

\subsection{Warm-up: Haar State Distinguishing Game}
\label{sec:warmup:state}
\noindent We first consider a simpler version of the Haar unitary distinguishing game, that we refer to as Haar {\em state} distinguishing game. 
\par In this simpler version, $\alice$ and $\bob$ receive as input $\ket{\psi}^{\otimes t}$ and $\ket{\phi}^{\otimes t}$ respectively. Here, $\ket{\psi}$ is sampled from the Haar distribution. Depending on the challenge bit, $\ket{\phi}=\ket{\psi}$ or $\ket{\phi}$ is sampled independently from the Haar distribution on $n$-qubit states. Prior work~\cite{ALL22,AGL24} showed that the probability that $(\alice,\bob)$ can predict the challenge bit correctly is at most $\frac{1}{2} + O\left( \frac{t^2}{N} \right)$, where $N = 2^n$. We will give an alternate proof of this fact. 
\par Since LOCC protocols with classical outputs can be cast as separable measurements, we will assume without loss of generality that the protocol can be implemented by the following separable measurement: $\{ \Lambda^0 := \sum_{a \in \Sigma} M_a^0 \otimes N_a^0,\allowbreak \Lambda^1 := \sum_{a \in \Sigma} M_a^1 \otimes N_a^1\}$, where $\Lambda^0 + \Lambda^1$ is the identity. Here, $M_a^0,M_a^1 \succeq 0$ act on $\alice$'s registers, and $N_a^0,N_a^1 \succeq 0$ act on $\bob$'s register. According to~\cite{Harrow13church}, $\alice$ and $\bob$'s state are: 
\begin{align*}
\rho & = \frac{1}{(N + 2t - 1)^{\downarrow 2t}} \sum_{\pi \in \Symgp_{2t}} P_N(\pi) \\
\sigma & = \frac{1}{(N + t - 1)^{\downarrow t}} \sum_{\pi_A \in \Symgp_{t}} P_N(\pi_A) 
\otimes \frac{1}{(N + t - 1)^{\downarrow t}} \sum_{\pi_B \in \Symgp_{t}} P_N(\pi_B).
\end{align*}
Here, $N^{\downarrow t} = N!/(N-t)!$ denotes the falling factorial and $P_N(\pi)$ is a permutation operator that permutes the blocks according to a permutation $\pi$. 
\par We need to show the following: 
\[
\frac{\Tr(\Lambda^0 \rho) + \Tr(\Lambda^1 \sigma)}{2} 
\leq \frac{1}{2} + O\qty(\frac{t^2}{N}).
\]
From~\cite{CGY24,AKY25}, we have that: 
\begin{align*}
& \Tr\left( \sum_{\pi \in \Symgp_{2t}} \left( \Pi^A \otimes \Pi^B \right) P_N(\pi) \right) \\
& \geq \Tr \left( \sum_{\pi_A \in \Symgp_t} \Pi^A \cdot P_N(\pi_A) \right) \cdot \Tr \left( \sum_{\pi_B \in \Symgp_t} \Pi^B \cdot P_N(\pi_B) \right),
\end{align*}
for any two positive semidefinite (PSD) matrices $\Pi^A$ and $\Pi^B$. Using this fact along with the fact that $\frac{(N + 2t - 1)^{\downarrow 2t}}{(N + t - 1)^{\downarrow t}(N + t - 1)^{\downarrow t}} \leq 1 + \veps$, where $\veps = O\left(\frac{t^2}{N} \right)$, we have the following:
\begin{align*}
\left(1 + \veps \right) \cdot \Tr\left( \Lambda^0 \rho \right) 
\geq \Tr\left( \Lambda^0 \sigma \right), \\
\left(1 + \veps \right) \cdot \Tr\left( \Lambda^1 \rho \right) 
\geq \Tr\left( \Lambda^1 \sigma \right).
\end{align*}
\noindent Since $\Tr\left( \Lambda^0 \rho \right) + \Tr\left( \Lambda^1 \rho \right) = 1$ and $\Tr\left( \Lambda^0 \sigma \right) + \Tr\left( \Lambda^1 \sigma \right) = 1$, we have
\begin{align*}
& \Tr(\Lambda^0 \rho) + \Tr(\Lambda^1 \sigma)
\leq \Tr(\Lambda^0 \rho) + (1 + \veps) \cdot \Tr(\Lambda^1 \rho)
\leq 1 + \veps,
\end{align*}
which completes the proof.

\subsection{Haar Unitary Distinguishing Game} 
\label{sec:hud}

Now, we move to the oracle setting in which $\alice$ and $\bob$ each have oracle access, but not the inverse, to either the same Haar random unitary $U$ ($b = 0$) or two independent Haar random unitaries $U,V$ ($b = 1$). For a formal definition, see~\Cref{def:HUDgame}. For simplicity, let's first consider the non-adaptive setting in which $\alice$ and $\bob$ both have to make queries \emph{all at once} in the beginning. Then they lose the oracle access after they start communicating. At the end, they locally perform a two-outcome measurement to generate the output. In other words, suppose they each makes $t$ queries with the initial state being $\rho \otimes \sigma$, the resulting state is either 
\begin{align*}
& \text{$b = 0$:} \quad \Ex_{U \sim \haarunitaries(N)}\qty[ U^{\otimes t}\rho U^{\dagger, \otimes t} \otimes U^{\otimes t}\sigma U^{\dagger, \otimes t} ]
\quad \text{or} \\
& \text{$b = 1$:} \quad \Ex_{U \sim \haarunitaries(N)}\qty[ U^{\otimes t}\rho U^{\dagger, \otimes t} ] \otimes \Ex_{V \sim \haarunitaries(N)} \qty[ V^{\otimes t}\sigma V^{\dagger, \otimes t} ].
\end{align*}
Fortunately, the closed form of them is well-studied in the literature. For any density matrix $\rho$ in $(\C^N)^{\otimes k}$, it is known that
\[
\Ex_{U \sim \haarunitaries(N)}\qty[ U^{\otimes k}\rho U^{\dagger, \otimes k} ] 
= \sum_{\pi,\tau \in \Symgp_k} \mathrm{Wg}(\pi,\tau,N) \cdot \Tr(P_N(\tau)^\intercal \cdot \rho) \cdot P_N(\pi),
\]
where $\mathrm{Wg}(\pi,\tau,N)$ is the Weingarten function~\cite{Wei83,Col03,CS06}. Despite having the exact expression, the calculation of the Weingarten function is often complicated. Recent works~\cite{HY24,SHH25} derive the following approximation formula
\[
\Phi_{\sf approx}(\rho) := \frac{1}{N^k} \sum_{\pi\in\Symgp_k} \Tr( P_N(\pi)^\intercal \cdot \rho ) \cdot P_N(\pi),
\]
which dramatically simplifies our analysis. After the query phase, the joint state of $(\alice,\bob)$ can be approximated by
\ifllncs
\else
\begin{align*}
& \text{$b = 0$:} \quad \frac{1}{N^{2t}} \sum_{\pi\in\Symgp_{2t}} \Tr( P_N(\pi)^\intercal \cdot \rho \otimes \sigma ) \cdot P_N(\pi) \\
& \text{$b = 1$:} \quad \frac{1}{N^t} \sum_{\pi_A\in\Symgp_t} \Tr( P_N(\pi_A)^\intercal \cdot \rho ) \cdot P_N(\pi_A) \otimes \frac{1}{N^t} \sum_{\pi_B\in\Symgp_t} \Tr( P_N(\pi_B)^\intercal \cdot \sigma ) \cdot P_N(\pi_B).
\end{align*}
\fi
The classical communication together with the final measurement corresponds to performing a separable measurement of the form $\{\sum_{a \in \Sigma} M^0_a \otimes N^0_a$, $\sum_{a \in \Sigma} M^1_a \otimes N^1_a \}$ on their state, where $\Sigma$ is an alphabet\footnote{One can add dummy zero operators to equalize the number of summands in the two measurement elements.} and $M^0_a$, $N^0_a$, $M^1_a$, $N^1_a$ are PSD operators for all $a \in \Sigma$ such that $M^0_a,M^1_a$ (\resp $N^0_a,N^1_a$) act on $\alice$'s (\resp $\bob$'s) system. Now, the probability $p_{b'|b}$ of $(\alice,\bob)$ outputting $b' \in \bit$ conditioned on the challenge being $b$ can be approximated by\footnote{Since $\Phi_{\sf approx}$ is not trace-preserving, $q_{0|b} + q_{1|b}$ is not necessarily equal to $1$.}
\ifllncs
\else
\begin{align*}
& q_{b'|0} := \frac{1}{N^{2t}} \sum_{a \in \Sigma, \pi\in\Symgp_{2t}} \Tr( P_N(\pi)^\intercal \cdot \rho \otimes \sigma ) \cdot \Tr( P_N(\pi) \cdot M_a^{b'} \otimes N_a^{b'} ), \\
& q_{b'|1} := \quad \frac{1}{N^{2t}} \sum_{a \in \Sigma,\pi_A,\pi_B\in\Symgp_t} \Tr( P_N(\pi_A)^\intercal \cdot \rho ) \cdot \Tr( P_N(\pi_A) \cdot M_a^{b'} ) \cdot \Tr( P_N(\pi_B)^\intercal \cdot \sigma ) \cdot \Tr( P_N(\pi_B) \cdot N_a^{b'} ).
\end{align*}
\fi
According to~\cite{SHH25}, for all $b,b' \in \bit$, $p_{b'|b}$ and $q_{b'|b}$ satisfy
\[
(1 - \veps) q_{b'|b} \leq p_{b'|b} \leq (1 + \veps) q_{b'|b}
\]
for $\veps = O(t^2/N)$. Next, we derive a relation between $q_{b'|0}$ and $q_{b'|1}$. Using the fact that $\Tr(X^\intercal Y) = \Tr(XY^\intercal)$, for any $a \in \Sigma$ and $\pi \in \Symgp_{2t}$,
\begin{align*}
& \Tr( P_N(\pi)^\intercal \cdot \rho \otimes \sigma ) \cdot \Tr( P_N(\pi) \cdot M_a^{b'} \otimes N_a^{b'} ) \\
& = \Tr( P_N(\pi) \otimes P_N(\pi) \cdot M_a^{b'} \otimes N_a^{b'} \otimes \rho^\intercal \otimes \sigma^\intercal ) \\
& = \Tr( P_{N^2}(\pi) \cdot ( M_a^{b'} \otimes \rho^\intercal )  \otimes ( N_a^{b'} \otimes \sigma^\intercal ) ),
\end{align*}
where the last equality follows from re-ordering the registers. That is, we can view $P_N(\pi) \otimes P_N(\pi)$ as a permutation $P_{N^2}(\pi)$ defined over $(\C^{N^2})^{\otimes 2t}$. Moreover, $\rho^\intercal, \sigma^\intercal \succeq 0$ since taking transpose preserves positivity. Now, summing over $a \in \Sigma$ and $\pi \in \Symgp_{2t}$ (\resp $\pi_A, \pi_B \in \Symgp_t$), we can use the same idea as in the Haar state distinguishing game to obtain $q_{0|0} \geq q_{0|1}$ and $q_{1|0} \geq q_{1|1}$. Finally, using the above bounds and elementary calculation, we obtain
\[
p_{1|1} \leq (1 + \veps) q_{1|1} 
\leq (1 + \veps) q_{1|0} 
\leq \frac{1 + \veps}{1 - \veps} p_{1|0}
= \frac{1 + \veps}{1 - \veps} (1 - p_{0|0})
\leq 1 + 4\veps - p_{0|0},
\]
which implies the guessing probability $(p_{0|0} + p_{1|1})/2$ is at most $1/2 + 2\veps$. \\

\noindent For the general case where $\alice$ and $\bob$ are allowed to use ancilla and make queries during the communication, most of the work is dedicated to expressing $(\alice,\bob)$'s guessing probability in a form that fits the above template. Looking ahead, the proof crucially relies on the fact that the approximation error in~\cite{SHH25} is \emph{multiplicative} which allows us to apply post-selections to relate the guessing probability of an adaptive querier to the one of a non-adaptive querier.

\subsection{Non-Adaptive Invertible Haar Unitary Distinguishing Game}
\label{sec:nahud}
We consider the non-adaptive setting in which $\alice$ and $\bob$ have forward access to the oracle, as well as its \emph{inverse}, but must make all queries simultaneously before communication. For a formal definition, see~\Cref{def:NAInvHUDgame}. Our technique is to show that oracle access to $(U,U^\dagger)$ is indistinguishable from $(U,V)$ under a polynomial number of \emph{non-adaptive} queries. Then using a standard hybrid, we can reduce the distinguishing game with the inverse to the one without the inverse. Informally, this is a quantum analogue of the indistinguishability between $(\pi,\pi^{-1})$ and $(\pi,\pi')$ for random permutations $\pi$ and $\pi'$ under non-adaptive queries. It is not hard to see that indistinguishability holds as long as the distinguisher does not query $\pi(x)$ to the second oracle for some query $x$ to the first oracle. Since the distinguisher must decide the queries to the second oracle, denoted by $\cQ$, ahead of time, we can perturb the first oracle so that it never returns any values belonging to $\cQ$. We follow this intuition and use the path-recording framework~\cite{MH25} to formalize the proof.

\subsection{Non-Interactive Invertible Haar Unitary Distinguishing Game} 
\label{sec:overview:ni:hud}
We consider the non-interactive setting where $\alice$ can make adaptive queries to her oracle and its inverse. Next, $\alice$ sends an $m$-bit message $\tau$ to $\bob$. Then $\bob$ can make adaptive queries to his oracle and its inverse. For a formal definition, see~\Cref{def:NIInvHUDgame}. Unlike our previous strategies, the proof relies on the concentration of the Haar measure, which has been extensively used in many recent works~\cite{Kretschmer21,CM24TCC,GMMY24,GHYZ24}. Operationally, the concentration property implies the following. Suppose $\sfC$ is an oracle quantum algorithm that makes $t$ queries to a Haar random unitary $U$ and outputs a bit. Now, if we sample a Haar random unitaries $U$, the probability $\Pr[1 \gets \sfC^U]$ will be $\veps$-close to a \emph{constant} $\Ex_{V\sim\haarunitaries(N)}[\Pr[1 \gets \sfC^V]]$ with probability at least $1 - \exp(-\Omega(N\veps^2/t^2))$, which is \emph{doubly-exponentially} small in the length of $U$. We observe that this property holds even when $\sfC$ has access to $U^\dagger$, which will be leveraged in our proof. 

In the Haar distinguishing game, we can view $\bob(\tau)$ as a $t$-query algorithm for any possible message $\tau$ sent by $\alice$. Now, for any $\tau \in \bit^m$, we define the set of ``bad'' oracles $U$ such that $\Pr[1 \gets \bob^{U,U^\dagger}(\tau)]$ is $\veps$-far away from the constant $\Ex_{V\sim\haarunitaries(N)}[\Pr[1 \gets \bob^{V,V^\dagger}(\tau)]]$. The concentration property ensures that the probability of a Haar random $U$ being ``good'' relative to \emph{all} $\tau \in \bit^m$ is at least $1 - 2^m\exp(-\Omega(N\veps^2/t^2))$ using the union bound. As a result, conditioning on $\alice$ is given a ``good'' oracle $U$, the probability $\Pr[1 \gets \bob^{U,U^\dagger}(\tau)]$ (when $b = 0$) is $\veps$-close to $\Ex_{V\sim\haarunitaries(N)}[\Pr[1 \gets \bob^{V,V^\dagger}(\tau)]]$ (when $b = 1$) for \emph{every} $\tau$. Careful readers may notice that the above argument holds regardless of how many queries $\alice$ makes, \ie $\alice$ could learn the description of $U$ and send the optimal $\tau$. Finally, by picking $\veps$ properly, we prove the distinguishing advantage is $O\qty( \poly(m,t,\log N) / \sqrt{N})$.




\subsection{From LOCC Indistinguishability to Black-Box Separations}
We discuss how to translate the LOCC indistinguishability frameworks into black-box separations. We first discuss the case of separating QCCC commitments and PRUs. Similar ideas can be adapted to rule out QCCC key agreement and PRUs as well. 
\par To rule out QCCC commitments, we consider the following model: the sender $\sen$ and the receiver $\rec$ in the QCCC commitments have access to exponentially many i.i.d Haar unitaries ${\cal U} = \{U_k\}_{k \in [\ell]}$, where the unitary $U_k$ is an $k$-qubit Haar unitary and $\ell=2^{\secparam}$. Here, $\sen$ and $\rec$ run in time polynomial in $\secparam$. A recent work~\cite{ABGL24} showed that pseudorandom unitaries exist in this model. We show that QCCC commitments do not exist in this model. 
\par To show the impossibility of QCCC commitments, we crucially use the Haar indistinguishability result from~\Cref{sec:hud}. Instead of both $\sen$ and $\rec$ having oracle access to the same set of unitaries $\{U_k\}_{k \in [\ell]}$, we now have $\sen$ get access to $\{U_k\}_{k \in [\ell]}$ and $\rec$ get access to ${\cal V} = \{V_k\}_{k \in [\ell]}$, where $V_k$ is a $k$-qubit Haar unitary sampled independently from $U_k$. This should not affect the correctness or security at all, thanks to the Haar indistinguishability result. Concretely, the Haar indistinguishability result needs to be invoked polynomially (in $\secparam$) many times, for each unitary the adversary queries to. 
\par Now since both $\sen$ and $\rec$ receive independently sampled Haar random unitaries, instead of giving oracle access to ${\cal U}$ and ${\cal V}$, we can now have $\sen$ itself sample ${\cal U}$ and similarly, have $\rec$ itself sample ${\cal V}$. In other words, we were able to completely compile out the oracles to obtain a correct and secure commitment scheme, satisfying information-theoretic security, in the plain model. Since information-theoretic security is impossible in the plain model, this would prove that QCCC commitments do not exist in the plain model. 
\par Unfortunately, there is a flaw in the above argument. The flaw lies at the point where we invoke the Haar indistinguishability result. Note that it was shown that the probability that any adversary succeeds in the Haar distinguishing game is at most $O\left( \frac{t^2}{N} \right)$, where $N$ is the dimension of the unitary it has access to. This means that if $N$ is too small (i.e. $\leq t^2$) then invoking the Haar indistinguishability result is meaningless. This means that we cannot switch from $U_k$ to $V_k$, for very small $k$. Indeed, this should not be surprising at all since it is possible for a LOCC adversary, using tomography, to distinguish the two cases when both have access to $U_k$ versus when $\sen$ has access to $U_k$ and $\rec$ has access to $V_k$, where $k=O(\log(\secparam))$.  
\par To address the flaw, we revise the above argument and incorporate the following two step approach:
\begin{itemize}
    \item In the first step, we transform the QCCC commitment scheme, where both $\sen$ and $\rec$ have access to ${\cal U}=\{U_k\}_{k \in [\ell]}$ into another secure and correct commitment scheme, where both have access to $\{U_k\}_{k \in [c \cdot (\log(\secparam))]}$, for some constant $c$. That is, we rid of all the unitaries with input length $> c \cdot (\log(\secparam))$. To achieve this transformation, similar to the above argument, we invoke the Haar indistinguishability result. 
    \item In the second step, we prove the impossibility of any QCCC commitment scheme where both $\sen$ and $\rec$ have access to $\{U_k\}_{k \in [O(\log(\secparam))]}$. Informally, the impossibility holds due to the fact that both $\sen$ and $\rec$ can separately perform tomography on the unitaries they have access to, in order to recover an approximate representation of all the unitaries $\{U_k\}_{k \in [O(\log(\secparam))]}$. At the end of the tomography process, both $\sen$ and $\rec$ have access to noisy shared randomness. We then argue that realizing QCCC commitments in the noisy share randomness model is impossible. Conceptually, this is similar to the model where both $\sen$ and $\rec$ have access to the same random string and it is known that achieving information-theoretic security in this setting is impossible. 
\end{itemize}
\noindent We follow a similar template to prove the impossibility of a certain class of constructions of QCCC commitments from strong PRUs. 
\section{Preliminaries} \label{sec:prelim}

We denote the security parameter by $\secp$. We refer readers to~\cite{nielsen_chuang_2010,WatrousBook} for the basics of quantum information.

\subsection{Notation}

\myparagraph{Indexing, sets, and the symmetric group} We use the notation $[n]$ to refer to the set $\set{ 1, \ldots, n }$.  For $N, t \in \N$, we denote the falling factorial by $N^{\downarrow t} := N!/(N-t)!$. We use $\Symgp_t$ to refer to the symmetric group of $t$ elements (i.e. the group of all permutations of $t$ elements). For a tuple $\vec{x} = (x_1,\dots,x_t)$, we define its support set as $\Supp(\vec{x}) := \bigcup_{i \in [t]} \set{x_i}$, which consists of the distinct elements appearing in $\vec{x}$.

\myparagraph{Quantum information}
A register $\reg{R}$ is a named finite-dimensional Hilbert space.  If $\reg{A}$ and $\reg{B}$ are registers, then $\reg{AB}$ denotes the tensor product of the two associated Hilbert spaces. 

We define $\mathrm{L}(\hilbert_\regA, \hilbert_\regB)$ as the set of all linear operators mapping a Hilbert space $\hilbert_\regA$ to a Hilbert space $\hilbert_\regB$. Similarly, we define $\mathrm{L}(\hilbert_\regA)$ as the set of all linear operators mapping a Hilbert space $\hilbert_\regA$ to itself. We define $\Unitary(\hilbert_\regA)$ as the set of all unitary operators acting on a Hilbert space $\Unitary(\hilbert_\regA)$. For $N \in \N$, we define $\Unitary(N)$ as the set of all $N$-dimensional unitary operators. We define $\mathrm{CP}(\hilbert_\regA, \hilbert_\regB)$ as the set of all linear completely positive maps from $\mathrm{L}(\hilbert_\regA)$ to $\mathrm{L}(\hilbert_\regB)$. Similarly, we define $\mathrm{C}(\hilbert_\regA, \hilbert_\regB)$ as the set of all quantum channels from $\mathrm{L}(\hilbert_\regA)$ to $\mathrm{L}(\hilbert_\regB)$. 

We denote by $\TD(\rho, \rho') = \frac{1}{2} \norm{\rho - \rho'}_1$ the trace distance between operators $\rho$ and $\rho'$, where $\norm{X}_1 = \Tr(\sqrt{X^{\dagger} X})$ is the trace norm. For two pure states $\ket{\psi}$ and $\ket{\phi}$, which may not be normalized, we use $\TD\qty(\ket{\psi},\ket{\phi})$ as a shorthand for $\TD \qty( \projector{\psi}, \projector{\phi} )$. By $\norm{ \cdot }_\diamond$ we denote the diamond norm, which has the following operation meaning. 
\begin{fact}\cite{AKN98}
\label{fact:diamondnorm}
Let $\cE_0$ and $\cE_1$ be quantum channels. Then the maximum distinguishing advantage of running the channel once is given by $1/2 \cdot \norm{ \cE_0 - \cE_1 }_\diamond$. Furthermore, when the channel is run $q$ times, the advantage is at most $q/2 \cdot \norm{ \cE_0 - \cE_1 }_\diamond$.
\end{fact}

We write $A^\intercal$ to denote the transpose of $A$. We write $\ptrans_{\reg{Y}}(A_{\reg{XY}})$ to denote the partial transpose (\wrt $\reg{Y}$) of $A_{\reg{XY}}$.\footnote{We note that the partial transpose is basis-dependent. Throughout this work, all partial transpose operations are defined over the computational basis. That is, $\ptrans_{\reg{Y}}( \ketbra{i}{j}_{\reg{X}} \otimes \ketbra{k}{\ell}_{\reg{Y}} ) = \ketbra{i}{j}_{\reg{X}} \otimes \ketbra{\ell}{k}_{\reg{Y}}$.} We write $A \succeq 0$ to denote that $A$ is a positive semidefinite (PSD) operator. Similarly, we write $A \succeq B$ to indicate that $A - B$ is a PSD operator.

For positive integers $N, t \in \N$ and permutation $\sigma \in \Symgp_t$, we let $P_N(\sigma)_{\reg{X_1}\ldots\reg{X_t}} \in \Unitary(N^t)$ be the unitary that acts on registers $\reg{X_1}, \dots, \reg{X_t}$, where each $\reg{X_i}$ associated with an $N$-dimensional Hilbert space, by permuting the registers according to $\sigma$. That is,
\begin{equation*}
    P_N(\sigma)_{\reg{X_1}\ldots\reg{X_t}} \ket{x_1}_{\reg{X_1}} \ket{x_2}_{\reg{X_2}} \dots \ket{x_t}_{\reg{X_t}}= \ket{x_{\sigma^{-1}(1)}}_{\reg{X_1}} \ket{x_{\sigma^{-1}(2)}}_{\reg{X_2}} \dots \ket{x_{\sigma^{-1}(t)}}_{\reg{X_t}}.
\end{equation*}

\myparagraph{Unitary, LOCC, and separable channels}
\begin{definition}[Unitary channels]
Let $\hilbert_{\reg{A}}$ be a Hilbert space and let $U \in \Unitary(\hilbert_{\reg{A}})$ be a unitary operator. The \emph{unitary channel} $\cE_U \in \mathrm{C}(\hilbert_{\reg{A}},\hilbert_{\reg{A}})$ is defined to be $\cE_U(X) := UXU^\dagger$ for every $X \in \mathrm{L}(\hilbert_{\reg{A}})$.
\end{definition}

\begin{definition}[Separable maps and channels~{\cite[Definitions~6.17,6.20]{WatrousBook}}]
\label{def:sep_maps}
Let $\hilbert_{\reg{A_0}},\hilbert_{\reg{A_1}},\hilbert_{\reg{B_0}}$, and $\hilbert_{\reg{B_1}}$ be Hilbert spaces. The set $\mathrm{SepCP}(\hilbert_{\reg{A_0}}, \hilbert_{\reg{A_1}}: \hilbert_{\reg{B_0}}, \hilbert_{\reg{B_1}})$ is defined to be the set of all completely positive maps of the form $\cE \in \mathrm{CP}(\hilbert_{\reg{A_0}} \otimes \hilbert_{\reg{B_0}}, \hilbert_{\reg{A_1}} \otimes \hilbert_{\reg{B_1}})$ for which there exists an alphabet $\Sigma$ and collections of completely positive maps $\set{\Psi_a \in \mathrm{CP}(\hilbert_{\reg{A_0}}, \hilbert_{\reg{A_1}}) }_{a \in \Sigma}$ and $\set{\Phi_a \in \mathrm{CP}(\hilbert_{\reg{B_0}}, \hilbert_{\reg{B_1}}) }_{a \in \Sigma}$ such that
\[
\cE = \sum_{a \in \Sigma} \Psi_a \otimes \Phi_a.
\]
Elements of the set $\mathrm{SepCP}(\hilbert_{\reg{A_0}}, \hilbert_{\reg{A_1}}: \hilbert_{\reg{B_0}}, \hilbert_{\reg{B_1}})$ are called \emph{separable maps}. \\

\noindent Define $\mathrm{SepC}(\hilbert_{\reg{A_0}}, \hilbert_{\reg{A_1}}: \hilbert_{\reg{B_0}}, \hilbert_{\reg{B_1}}) 
:= \mathrm{SepCP}(\hilbert_{\reg{A_0}}, \hilbert_{\reg{A_1}}: \hilbert_{\reg{B_0}}, \hilbert_{\reg{B_1}}) 
\cap \mathrm{C}(\hilbert_{\reg{A_0}} \otimes \hilbert_{\reg{B_0}}, \hilbert_{\reg{A_1}} \otimes \hilbert_{\reg{B_1}})$. Elements of the set $\mathrm{SepC}(\hilbert_{\reg{A_0}}, \hilbert_{\reg{A_1}}: \hilbert_{\reg{B_0}}, \hilbert_{\reg{B_1}})$ are called \emph{separable channels}.
\end{definition}

\begin{proposition}[{\cite[Proposition~6.18]{WatrousBook}}]
\label{prop:sep_kraus}
Let $\hilbert_{\reg{A_0}},\hilbert_{\reg{A_1}},\hilbert_{\reg{B_0}}$, and $\hilbert_{\reg{B_1}}$ be Hilbert spaces and let $\cE \in \mathrm{CP}(\hilbert_{\reg{A_0}} \otimes \hilbert_{\reg{B_0}}, \hilbert_{\reg{A_1}} \otimes \hilbert_{\reg{B_1}})$ be a completely positive map. It holds that 
\[
\cE \in \mathrm{SepCP}(\hilbert_{\reg{A_0}}, \hilbert_{\reg{A_1}}: \hilbert_{\reg{B_0}}, \hilbert_{\reg{B_1}})
\]
if and only if there exists an alphabet $\Sigma$ and collections of operators
\[
\set{A_a \in \mathrm{L}(\hilbert_{\reg{A_0}},\hilbert_{\reg{A_1}})}_{a \in \Sigma}
\quad \text{and} \quad
\set{B_a \in \mathrm{L}(\hilbert_{\reg{B_0}},\hilbert_{\reg{B_1}})}_{a \in \Sigma}
\]
such that
\[
\cE(X) = \sum_{a \in \Sigma} (A_a \otimes B_a) X (A_a \otimes B_a)^\dagger
\]
for every $X \in \mathrm{L}(\hilbert_{\reg{A_0}} \otimes \hilbert_{\reg{B_0}}, \hilbert_{\reg{A_1}} \otimes \hilbert_{\reg{B_1}})$.
\end{proposition}

\begin{definition}[One-way right/left LOCC channel~{\cite[Definition~6.25]{WatrousBook}}]
Let $\hilbert_{\regA},\hilbert_{\reg{A'}},\hilbert_{\regB}$, and $\hilbert_{\reg{B'}}$ be Hilbert spaces and let $\cE \in \mathrm{C}(\hilbert_{\regA} \otimes \hilbert_{\regB}, \hilbert_{\reg{A'}} \otimes \hilbert_{\reg{B'}})$ be a channel. The channel $\cE$ is a \emph{one-way right LOCC channel} if there exists an alphabet $\Sigma$ and a collection $\set{ \Phi_a \in \mathrm{CP}(\hilbert_{\regA}, \hilbert_{\reg{A'}}) }_{a \in \Sigma}$ of completely positive maps satisfying $\sum_{a \in \Sigma} \Phi_a \in \mathrm{C}(\hilbert_{\regA}, \hilbert_{\reg{A'}})$ along with a collection $\set{ \Psi_a \in \mathrm{C}(\hilbert_{\regB}, \hilbert_{\reg{B'}}) }_{a \in \Sigma}$ of channels, such that
\[
\cE = \sum_{a \in \Sigma} \Phi_a \otimes \Psi_a.
\]

\noindent \emph{One-way left LOCC channels} are defined symmetrically be switching $\regA$ with $\regB$ and $\reg{A'}$ with $\reg{B'}$.
\end{definition}

Note that every one-way left/right LOCC channel is a separable channel according to~\Cref{def:sep_maps}.

\myparagraph{Haar measure} We denote by $\haarstates_N$ the Haar distribution over $N$-dimensional states, and $\haarunitaries_N$ the Haar measure over the unitary group $\Unitary(N)$ (\ie the unique left and right invariant measure).

\ifllncs
\else
\myparagraph{Tensor Network Diagrams} Tensor network diagrams provide a graphical representation of tensor operations. We introduce some frequently used diagrams below, where for $\reg{X}$ and $\reg{X'}$ such that $\hilbert_{\reg{X}} = \hilbert_{\reg{X'}} = d_X$, the \emph{non-normalized} maximally entangled state is denoted by $\ket{\Omega}_{\reg{XX'}} := \sum_{x \in [d_X]} \ket{x}_{\reg{X}} \ket{x}_{\reg{X'}}$. For a comprehensive introduction, see~\cite[Section~6]{Mel24}.

\begin{figure}[H]
\centering
\begin{subfigure}{0.3\textwidth} 
\centering
\resizebox{1\textwidth}{!}{%
\begin{circuitikz}
\tikzstyle{every node}=[font=\LARGE]
\draw [short] (8,7.25) .. controls (9.25,7.75) and (9.25,9.25) .. (8,9.75);
\draw  (1.75,8.5) rectangle (4.25,6);
\draw  (5.5,8.5) rectangle (8,6);
\draw (4.25,7.25) to[short] (5.5,7.25);
\node [font=\LARGE] at (3,7.25) {$A$};
\node [font=\LARGE] at (6.75,7.25) {$B$};
\draw (1.75,9.75) to[short] (8,9.75);
\draw [short] (1.75,9.75) .. controls (0.5,9.25) and (0.5,7.75) .. (1.75,7.25);
\end{circuitikz}
}%
\caption{Trace: $\Tr(AB)$}
\end{subfigure}
\hfill 
\begin{subfigure}{0.3\textwidth} 
\centering
\resizebox{1\textwidth}{!}{%
\begin{circuitikz}
\tikzstyle{every node}=[font=\LARGE]
\draw [short] (5.5,4.75) .. controls (6.75,5.25) and (6.75,6.75) .. (5.5,7.25);
\draw  (1.75,8.5) rectangle (4.25,6);
\draw (4.25,7.25) to[short] (5.5,7.25);
\node [font=\Huge] at (3,7.25) {$A$};
\draw (1.75,9.75) to[short] (5.5,9.75);
\draw [short] (1.75,9.75) .. controls (0.5,9.25) and (0.5,7.75) .. (1.75,7.25);
\draw (1.75,4.75) to[short] (5.5,4.75);
\draw [short] (-0.75,7.25) .. controls (1,6.75) and (0,5.25) .. (1.75,4.75);
\draw (-2,7.25) to[short] (-0.75,7.25);
\draw (8,7.25) to[short] (9.25,7.25);
\draw [short] (5.5,9.75) .. controls (7.25,9.25) and (6.25,7.75) .. (8,7.25);
\end{circuitikz}
}%
\caption{Transpose: $A^\intercal$}
\end{subfigure}
\hfill
\begin{subfigure}{0.3\textwidth} 
\centering
\resizebox{1\textwidth}{1.5cm}{%
\begin{circuitikz}
\tikzstyle{every node}=[font=\LARGE]
\draw  (1.25,9) rectangle (2.5,7.75);
\draw (4.5,8.5) to[short] (2.5,8.5);
\draw (-0.75,8.5) to[short] (1.25,8.5);
\draw  (1.25,7.5) rectangle (2.5,6.25);
\draw (4.5,7) to[short] (2.5,7);
\draw (-0.75,7) to[short] (1.25,7);
\node [font=\LARGE] at (1.8,8.4) {$A$};
\node [font=\LARGE] at (1.8,6.9) {$B$};
\end{circuitikz}
}%
\caption{Tensor Product: $A \otimes B$}
\end{subfigure}
\\ 
\begin{subfigure}{0.3\textwidth} 
\centering
\resizebox{1\textwidth}{!}{%
\begin{circuitikz}
\tikzstyle{every node}=[font=\LARGE]
\draw  (1.25,9) rectangle (3.75,6.5);
\node [font=\LARGE] at (2.5,7.75) {$A$};
\draw (5,7) to[short] (3.75,7);
\draw (0,7) to[short] (1.25,7);
\draw (0,6) to[short] (5,6);
\draw (3.75,8.5) to[short] (6.25,8.5);
\draw (-1.25,8.5) to[short] (1.25,8.5);
\draw [short] (5,7) .. controls (5.75,6.75) and (5.75,6.25) .. (5,6);
\draw [short] (0,7) .. controls (-0.75,6.75) and (-0.75,6.25) .. (0,6);
\end{circuitikz}
}%
\caption{Partial Trace: $\Tr_{\reg{Y}}(A_{\reg{XY}})$}
\end{subfigure}
\hfill 
\begin{subfigure}{0.3\textwidth} 
\centering
\resizebox{1\textwidth}{!}{%
\begin{circuitikz}
\tikzstyle{every node}=[font=\LARGE]
\draw  (1.25,9) rectangle (3.75,6.5);
\node [font=\LARGE] at (2.5,7.75) {$A$};
\draw (5,7) to[short] (3.75,7);
\draw (0,7) to[short] (1.25,7);
\draw (0,5.25) to[short] (6.25,5.25);
\draw (-1.25,6) to[short] (5,6);
\draw (3.75,8.5) to[short] (6.25,8.5);
\draw (-1.25,8.5) to[short] (1.25,8.5);
\draw [short] (5,7) .. controls (5.75,6.75) and (5.75,6.25) .. (5,6);
\draw [short] (0,7) .. controls (-0.75,6.75) and (-0.75,5.5) .. (0,5.25);
\end{circuitikz}
}%
\caption{Partial Transpose: $\ptrans_{\reg{Y}}(A_{\reg{XY}})$}
\end{subfigure}
\hfill
\begin{subfigure}{0.3\textwidth} 
\centering
\resizebox{1\textwidth}{1.5cm}{%
\begin{circuitikz}
\tikzstyle{every node}=[font=\LARGE]
\draw (4.5,8.5) to[short] (2.5,8.5);
\draw (-0.75,8.5) to[short] (1.25,8.5);
\draw (4.5,7) to[short] (2.5,7);
\draw (-0.75,7) to[short] (1.25,7);
\draw [short] (1.25,8.5) .. controls (1.75,8.25) and (1.75,7.25) .. (1.25,7);
\draw [short] (2.5,8.5) .. controls (2,8.25) and (2,7.25) .. (2.5,7);
\end{circuitikz}
}%
\caption{$\projector{\Omega}_{\reg{XX'}}$}
\end{subfigure}
\caption{Examples of tensor network diagrams.}
\end{figure}
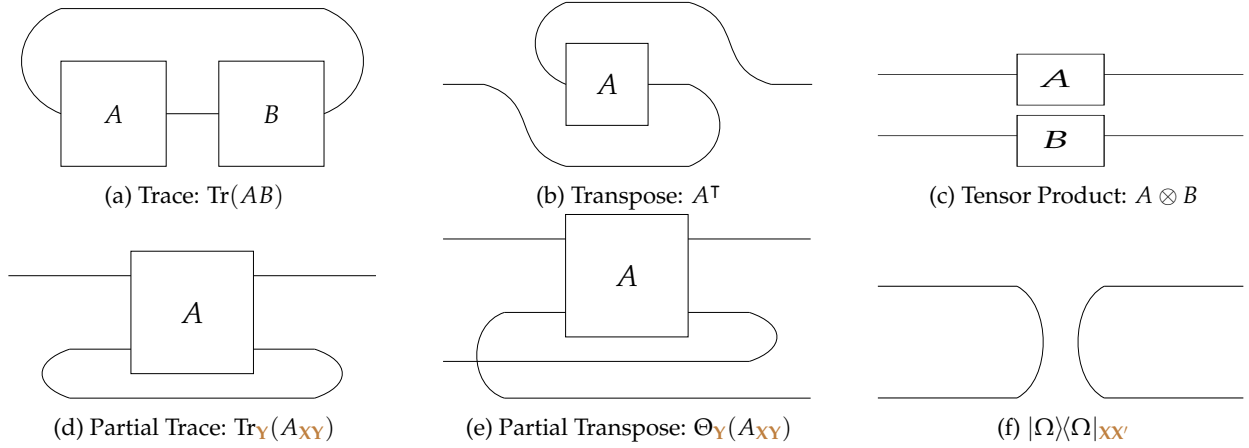 

\fi

\begin{fact}[Ricochet Property]
\label{fact:Ricochet}
Let $\hilbert_\regX$ and $\hilbert_{\regX'}$ be Hilbert spaces of equal dimension $d_X$. For any $A \in \mathrm{L}(\hilbert_\regX)$, $(A_\regX \otimes \id_{\reg{X'}}) \ket{\Omega}_{\reg{XX'}} = (\id_\regX \otimes A^\intercal_{\reg{X'}}) \ket{\Omega}_{\reg{XX'}}$. Furthermore, by re-ordering registers, we have the following variants:
\begin{align*}
& \sum_{x\in[d_X]} A_\regX \ketbra{x}{\psi}_\regX \otimes \ketbra{x}{\phi}_{\reg{X'}}
= \sum_{x\in[d_X]} \ketbra{x}{\psi}_\regX \otimes A^\intercal_{\reg{X'}} \ketbra{x}{\phi}_{\reg{X'}} \\
& \sum_{x\in[d_X]} A_\regX \ketbra{x}{\psi}_\regX \otimes \ketbra{\phi}{x}_{\reg{X'}}
= \sum_{x\in[d_X]} \ketbra{x}{\psi}_\regX \otimes \ketbra{\phi}{x}_{\reg{X'}} A_{\reg{X'}} \\
& \sum_{x\in[d_X]} \ketbra{\psi}{x}_\regX A_\regX \otimes \ketbra{x}{\phi}_{\reg{X'}}
= \sum_{x\in[d_X]} \ketbra{\psi}{x}_\regX \otimes A_{\reg{X'}} \ketbra{x}{\phi}_{\reg{X'}} \\
& \sum_{x\in[d_X]} \ketbra{\psi}{x}_\regX A_\regX \otimes \ketbra{\phi}{x}_{\reg{X'}}
= \sum_{x\in[d_X]} \ketbra{\psi}{x}_\regX \otimes \ketbra{\phi}{x}_{\reg{X'}} A^\intercal_{\reg{X'}}
\end{align*}
where $\ket{\psi}$ and $\ket{\phi}$ are arbitrary (possibly unnormalized) vectors.
\end{fact}

\subsection{Ma-Huang's Path Recording Framework}

Before we recall the isometries described by~\cite{MH25}, we first set up some notation. A \emph{relation} $R$ is defined as a \emph{multiset} $R = \{(x_1,y_1), \dots, (x_t, y_t) \}$ of ordered pairs $(x_i, y_i) \in [N]^2$. The \emph{domain} of a relation $R$ is the \emph{set} $\Dom(R) := \set{x: x \in [N], \exists y\ s.t.\ (x,y) \in R }$. The \emph{image} of a relation $R$ is the \emph{set} $\Im(R) := \set{y: y \in [N], \exists x\ s.t.\ (x,y) \in R }$. For a relation $R = \{(x_1,y_1), \dots, (x_t, y_t) \}$, define the corresponding \emph{relation state} $\ket{R}$ to be 
\[
\ket{R}_\regR 
:= \frac{ \sum_{\pi \in \Symgp_t} \ket{x_{\pi^{-1}(1)}}_{\reg{R_{X,1}}} \ket{y_{\pi^{-1}(1)}}_{\reg{R_{Y,1}}} \dots \ket{x_{\pi^{-1}(t)}}_{\reg{R_{X,t}}} \ket{y_{\pi^{-1}(t)}}_{\reg{R_{Y,t}}} }{\sqrt{t! \prod_{(x,y)\in[N]^2} \num(R,(x,y))!}},
\]
where $\num(R,(x,y))$ denotes the number of times the tuple $(x,y)$ appears in $R$. If the tuples in a relation $R = \{(x_1,y_1), \dots, (x_t, y_t) \}$ are all distinct, then 
\[
\ket{R}_\regR = \frac{1}{\sqrt{t!}} \sum_{\pi \in \Symgp_t} \ket{x_{\pi^{-1}(1)}}_{\reg{R_{X,1}}} \ket{y_{\pi^{-1}(1)}}_{\reg{R_{Y,1}}} \dots \ket{x_{\pi^{-1}(t)}}_{\reg{R_{X,t}}} \ket{y_{\pi^{-1}(t)}}_{\reg{R_{Y,t}}}.
\]
We define the following two partial isometries: let $x,y \in [N]$ and $L, R$ be two relations, 
\begin{align*}
V^L \cdot \ket{x}_{\regA} \ket{L}_{\regL} \ket{R}_{\regR} 
& = \frac{1}{\sqrt{N - |\Im(L \cup R)|}} \sum_{y \notin \Im(L \cup R)} \ket{y}_{\reg{A}} \ket{L \cup \set{(x,y)}}_{\regL} \ket{R}_{\regR}, \\
V^R \cdot \ket{y}_{\regA} \ket{L}_{\regL} \ket{R}_{\regR} 
& = \frac{1}{\sqrt{N - |\Dom(L \cup R)|}} \sum_{x \notin \Dom(L \cup R)} \ket{x}_{\regA} \ket{L}_{\regL} \ket{R \cup \set{(x,y)}}_{\regR}.
\end{align*}

\noindent Using $V^L$ and $V^R$, they define the \emph{path-recording oracle}, which is a partial isometry:
\[
V := V^L \cdot (\id - V^R \cdot V^{R,\dagger}) + (\id - V^L \cdot V^{L,\dagger}) \cdot V^{R,\dagger}.
\]
They then showed that the following theorem:

\begin{theorem}[{\cite[Theorem~8]{MH25}}]
\label{thm:MH24}
Let $\cA$ be an oracle quantum algorithm that makes a total of $t$ forward and inverse queries, and let $\rho^{\cA,t}_{\reg{AB}}$ be the final density matrix when $\cA$ is interacting with a Haar random unitary, $\ket{\cA^V_t}_{\reg{ABLR}}$ be the final state when $\cA$ is interacting with $V$. Then 
\[
\TD( \rho^{\cA,t}_{\reg{AB}}, \Tr_{\reg{LR}}(\projector{\cA^V_t}_{\reg{ABLR}}) ) \leq O\qty(\frac{t^2}{N^{1/8}}).
\]
\end{theorem}

\section{LOCC Indistinguishability of Haar Random Unitaries}
\label{sec:LOCC}
In this section, we show that any two-party LOCC adversary that makes a polynomial number of queries is unable to distinguish whether they are given oracle access to the same Haar random unitary or to independent Haar random unitaries in various settings. We formalize two-party LOCC adversaries with oracle access using the following definition.

\begin{definition}[LOCC channels \wrt fixed oracles]
Let $\hilbert_{\reg{A_{in}}}$, $\hilbert_{\reg{A_{out}}}$, $\hilbert_{\reg{B_{in}}}$, $\hilbert_{\reg{B_{out}}}$, $\hilbert_{\reg{O_A}}$, and $\hilbert_{\reg{O_B}}$ be Hilbert spaces and let $U \in \Unitary(\hilbert_{\reg{O_A}}), V \in \Unitary(\hilbert_{\reg{O_B}})$ be unitaries, $\cE^{U,V} \in \mathrm{C}(\hilbert_{\reg{O_A}} \otimes \hilbert_{\reg{A_{in}}} \otimes \hilbert_{\reg{O_B}} \otimes\hilbert_{\reg{B_{in}}}, \hilbert_{\reg{A_{out}}} \otimes \hilbert_{\reg{B_{out}}})$ a channel. The channel $\cE^{U,V}$ is an \emph{LOCC channels \wrt $(U, V)$} if $\cE^{U,V}$ can be obtained from the composition of any finite number of (1) one-way right LOCC channels, (2) one-way left LOCC channels, (3) unitary channels $\cE_U \otimes \id$ and (4) unitary channels $\id \otimes \cE_V$. Namely, either $\cE^{U,V}$ is a one-way right LOCC channel, a one-way left LOCC channel, a unitary channel $\cE_U \otimes \id$, a unitary channel $\id \otimes \cE_V$, or there exists an integer $r \geq 2$, Hilbert spaces $\hilbert_{\reg{A_1}},\dots,\hilbert_{\reg{A_{r-1}}}$ and $\hilbert_{\reg{B_1}},\dots,\hilbert_{\reg{B_{r-1}}}$, and channels
\begin{align*}
    \cE_1 & \in \mathrm{C}(\hilbert_{\reg{O_A}} \otimes \hilbert_{\reg{A_{in}}} \otimes \hilbert_{\reg{O_B}} \otimes\hilbert_{\reg{B_{in}}}, \hilbert_{\reg{O_A}} \otimes \hilbert_{\reg{A_1}} \otimes \hilbert_{\reg{O_B}} \otimes\hilbert_{\reg{B_1}}), \\
    \cE_2 & \in \mathrm{C}(\hilbert_{\reg{O_A}} \otimes \hilbert_{\reg{A_1}} \otimes \hilbert_{\reg{O_B}} \otimes\hilbert_{\reg{B_1}}, \hilbert_{\reg{O_A}} \otimes \hilbert_{\reg{A_2}} \otimes \hilbert_{\reg{O_B}} \otimes\hilbert_{\reg{B_2}}), \\
    & \hspace{0.3\textwidth} \vdots \\
    \cE_r & \in \mathrm{C}(\hilbert_{\reg{O_A}} \otimes \hilbert_{\reg{A_{r-1}}} \otimes \hilbert_{\reg{O_B}} \otimes\hilbert_{\reg{B_{r-1}}}, \hilbert_{\reg{O_A}} \otimes \hilbert_{\reg{A_{out}}} \otimes \hilbert_{\reg{O_B}} \otimes\hilbert_{\reg{B_{out}}}),
\end{align*}
each of which is either a one-way right LOCC channel, a one-way left LOCC channel, a unitary channel $\cE_U \otimes \id$, or a unitary channel $\id \otimes \cE_V$ such that $\cE^{U,V} = \cE_r \circ \dots \circ \cE_1$.
\end{definition}

Note that the number and sequence of channels could depend on the choice of $(U,V)$. For simplicity, we pad dummy rounds and queries and consider channels of the following ``normal'' form. For every $(U,V)$,
\begin{equation*}
\cE^{U,V} = \cE_U \otimes \cE_V \circ \cE_r \circ \dots \circ \cE_U \otimes \cE_V \circ \cE_2 \circ \cE_U \otimes \cE_V \circ \cE_1
\end{equation*}
where each $\cE_i$ is either a one-way right LOCC channel or a one-way left LOCC channel. In other words, $\cE^{U,V}$ is composed by one-way LOCC channels and unitary channels $\cE_U \otimes \cE_V$ alternatively. From~\Cref{prop:sep_kraus}, each $\cE_i$ can be written as $\cE_i(X) = \sum_{a_i \in \Sigma_i} (A_{i,a_i} \otimes B_{i,a_i}) X (A_{i,a_i} \otimes B_{i,a_i})^\dagger$ for an alphabet $\Sigma_i$ and a set $\set{A_{i,a_i} \otimes B_{i,a_i}}_{a_i \in \Sigma_i}$ of operators which are all independent of $U$ and $V$. Expanding $\cE^{U,V}$, we have\footnote{We use $\prod$ to denote sequential composition of operations, \ie $\prod_{i=1}^n A_i := A_1 \circ A_2 \circ \dots \circ A_n$. Note that order matters, as the composition of operators is generally non-commutative.}
\begin{align}
\label{eq:normal_form}
\cE^{U,V}(X) = 
\sum_{(a_1,\dots,a_r) \in \Sigma_1 \times \dots \times \Sigma_r} \qty( \prod_{i=r}^1 UA_{i,a_i} \otimes VB_{i,a_i} ) X \qty( \prod_{i'=1}^r A^\dagger_{i',a_{i'}}U^\dagger \otimes B^\dagger_{i',a_{i'}}V^\dagger )
\end{align}
for every $X \in \mathrm{L}(\hilbert_{\reg{O_A}} \otimes \hilbert_{\reg{A_{in}}} \otimes \hilbert_{\reg{O_B}} \otimes\hilbert_{\reg{B_{in}}})$.

\subsection{LOCC Indistinguishability against Adaptive Queries}
In this section, we consider the distinguishing game in~\Cref{def:HUDgame}. 

\begin{definition}[Haar Unitary Distinguishing Game]
\label{def:HUDgame}
A \emph{Haar unitary distinguishing game}, parametrized by dimension $N \in \N$, is played by the challenger $\challenger$ and a two-party adversary $(\alice^{(\cdot)}, \bob^{(\cdot)})$.
\begin{protocol}{$\underline{\mathsf{HUD}(N, \alice, \bob)}$}
{
\begin{enumerate}
    \item $\challenger$ samples two independent Haar unitaries $U,V\sim\haarunitaries(N)$.
    \item $\challenger$ samples a challenge bit $b \gets \bit$. If $b = 0$, $\alice$ and $\bob$ are both given oracle access to $U$. Otherwise, $\alice$ is given oracle access to $U$ and $\bob$ is given oracle access to $V$.
    \item $\alice$ and $\bob$ can perform an arbitrary finite number of rounds of classical communication, during which they can make queries adaptively.
    \item $(\alice, \bob)$ outputs a guess $b'\in\bit$.
    \item $\challenger$ outputs $1$ if and only if $b' = b$.
\end{enumerate}
}
\end{protocol}
In particular, $(\alice^{(\cdot)}, \bob^{(\cdot)})$ start with the all-zero state, perform an LOCC channel $\cE$ \wrt $(U,V)$, and perform a two-outcome POVM of the form $\set{ M^0_{\regA} \otimes N^0_{\regB}, M^1_{\regA} \otimes N^1_{\regB} }$ to generate $b' \in \bit$.
\end{definition}

\noindent Before proving the main result of this section (\Cref{thm:adap_LOCC}), we first introduce several lemmas. 

\begin{lemma}[{\cite[Lemma~16]{CGY24}}, implicitly in~{\cite[Lemma~16, Theorem~17]{AKY25}}]
\label{lem:product_perm}
Let $N, t \in \N$ and $\hilbert_\reg{A} = \hilbert_\reg{B} = (\C^N)^{\otimes t}$. Then for any PSD operators $M_{\reg{A}}, N_{\reg{B}} \succeq 0$, it holds that
\ifllncs
\else
\begin{align*}
\sum_{\pi\in\Symgp_{2t}} \Tr( P_N(\pi)_{\reg{AB}} \cdot M_{\reg{A}} \otimes N_{\reg{B}} ) 
\geq \sum_{\substack{\pi_A \in \Symgp_t \\ \pi_B \in \Symgp_t}} \Tr( P_N(\pi_A)_{\reg{A}} \otimes P_N(\pi_B)_{\reg{B}}  \cdot M_{\reg{A}} \otimes N_{\reg{B}} )
\geq 0.
\end{align*}
\fi
\end{lemma}

\begin{lemma}[{\cite[Lemma~1]{SHH25}}]
\label{lem:multiplicative_twirling}
Let $k,N\in\N$ such that $k^2 \leq N$ and $\hilbert_\regA = (\C^N)^{\otimes k}$, and define the completely positive maps
\begin{align*}
\Phi_a(X_\regA) & := \frac{1}{N^k} \sum_{\pi\in\Symgp_k} P_N(\pi)_{\reg{A}} \cdot \Tr( P_N(\pi)^\intercal_{\reg{A}} \cdot X_{\regA}) \\
\Phi_H(X_\regA) & := \Ex_{U \sim \haarunitaries(N)} [ U^{\otimes k}_{\reg{A}} \cdot X_{\regA} \cdot U^{\dagger, \otimes k}_{\regA} ]
\end{align*}
for any $X_\regA \in \mathrm{L}(\hilbert_\regA)$. Then $(1-\veps) \cdot \Phi_a \preceq \Phi_H \preceq (1+\veps) \cdot \Phi_a$, where $\Phi \preceq \Phi'$ denotes that $\Phi' - \Phi$ is a completely positive map. That is, for any Hilbert space $\hilbert_\regB$ and any $\rho_{\reg{AB}} \succeq 0$,
\begin{align*}
    0 \preceq & (1-\veps) \cdot \frac{1}{N^k} \sum_{\pi\in\Symgp_k} P_N(\pi)_{\reg{A}} \otimes \Tr_{\reg{A}}(P_N(\pi)^\intercal_{\reg{A}} \otimes \id_{\reg{B}} \cdot \rho_{\reg{AB}}) \\
    \preceq & 
    \Ex_{U\sim\mu(N)}[ U^{\otimes k}_{\reg{A}} \otimes \id_{\reg{B}} \cdot \rho_{\reg{AB}} \cdot U^{\dagger, \otimes k}_{\reg{A}} \otimes \id_{\reg{B}} ] \\
    \preceq &
    (1+\veps) \cdot \frac{1}{N^k} \sum_{\pi\in\Symgp_k} P_N(\pi)_{\reg{A}} \otimes \Tr_{\reg{A}}(P_N(\pi)^\intercal_{\reg{A}} \otimes \id_{\reg{B}} \cdot \rho_{\reg{AB}}).
\end{align*}
\end{lemma}

\begin{fact}
\label{fact:PSD_trace}
Let $A,B,M \succeq 0$ such that $A \succeq B$. Then $\Tr(AM) \geq \Tr(BM) \geq 0$.
\end{fact}

\begin{fact}
\label{fact:gate_teleportation}
Let $A \in \mathrm{L}(\hilbert_{\reg{X}},\hilbert_{\reg{Y}}), B \in \mathrm{L}(\hilbert_{\reg{Y}})$, and $C \in \mathrm{L}(\hilbert_{\reg{Y}},\hilbert_{\reg{X}})$. Then
\begin{equation*}
ABC 
= \bra{\Omega_Y} \otimes \id 
\cdot \bigg(
B \otimes (\id \otimes A \cdot \projector{\Omega_Y} \cdot \id \otimes C)
\bigg)
\cdot \ket{\Omega_Y} \otimes \id,
\end{equation*}
\ifllncs
\else
where $\ket{\Omega_Y} := \sum_{i \in [\dim(\hilbert_{\regY})]} \ket{i}\ket{i}$ denotes the non-normalized maximally entangled state (see~\Cref{fig:gate_teleportation}).
\fi
\end{fact}

\ifllncs
\else
\begin{figure}[H]
\centering
\resizebox{0.75\textwidth}{!}{%
\begin{circuitikz}
\tikzstyle{every node}=[font=\LARGE]
\draw  (-22,-76) rectangle (-19.5,-78.5);
\draw  (-18.25,-76) rectangle (-15.75,-78.5);
\draw  (-14.5,-76) rectangle (-12,-78.5);
\draw (-23.25,-77.25) to[short] (-22,-77.25);
\draw (-19.5,-77.25) to[short] (-18.25,-77.25);
\draw (-15.75,-77.25) to[short] (-14.5,-77.25);
\draw (-12,-77.25) to[short] (-10.75,-77.25);
\node [font=\LARGE] at (-13.25,-77.25) {$C$};
\node [font=\LARGE] at (-17,-77.25) {$B$};
\node [font=\LARGE] at (-20.75,-77.25) {$A$};
\draw  (-6.25,-78.5) rectangle (-3.75,-81);
\draw  (-3.75,-73.5) rectangle (-1.25,-76);
\draw  (-1.25,-78.5) rectangle (1.25,-81);
\draw [ color={rgb,255:red,255; green,0; blue,0}, line width=2pt, short] (-3.75,-79.75) .. controls (-2.5,-79.25) and (-2.5,-77.75) .. (-3.75,-77.25);
\draw (-6.25,-77.25) to[short] (-3.75,-77.25);
\draw (-7.5,-79.75) to[short] (-6.25,-79.75);
\draw (-6.25,-74.75) to[short] (-3.75,-74.75);
\draw (-1.25,-74.75) to[short] (1.25,-74.75);
\draw (-1.25,-77.25) to[short] (1.25,-77.25);
\draw (2.5,-79.75) to[short] (1.25,-79.75);
\draw [ color={rgb,255:red,255; green,0; blue,0}, line width=2pt, short] (-1.25,-79.75) .. controls (-2.5,-79.25) and (-2.5,-77.75) .. (-1.25,-77.25);
\draw [ color={rgb,255:red,255; green,0; blue,0}, line width=2pt, short] (-6.25,-77.25) .. controls (-7.5,-76.75) and (-7.5,-75.25) .. (-6.25,-74.75);
\draw [ color={rgb,255:red,255; green,0; blue,0}, line width=2pt, short] (1.25,-77.25) .. controls (2.5,-76.75) and (2.5,-75.25) .. (1.25,-74.75);
\node [font=\LARGE] at (0,-79.75) {$C$};
\node [font=\LARGE] at (-2.5,-74.75) {$B$};
\node [font=\LARGE] at (-5,-79.75) {$A$};
\node [font=\LARGE, color={rgb,255:red,255; green,0; blue,0}] at (3,-76) {$\ket{\Omega_Y}$};
\node [font=\LARGE, color={rgb,255:red,255; green,0; blue,0}] at (-1,-78) {$\bra{\Omega_Y}$};
\node [font=\LARGE, color={rgb,255:red,255; green,0; blue,0}] at (-8,-76) {$\bra{\Omega_Y}$};
\node [font=\LARGE, color={rgb,255:red,255; green,0; blue,0}] at (-3.75,-78) {$\ket{\Omega_Y}$};
\node [font=\Huge] at (-9.25,-77.25) {$=$};
\end{circuitikz}
}%
\caption{A graphical expression of~\Cref{fact:gate_teleportation}.}
\label{fig:gate_teleportation}
\end{figure}
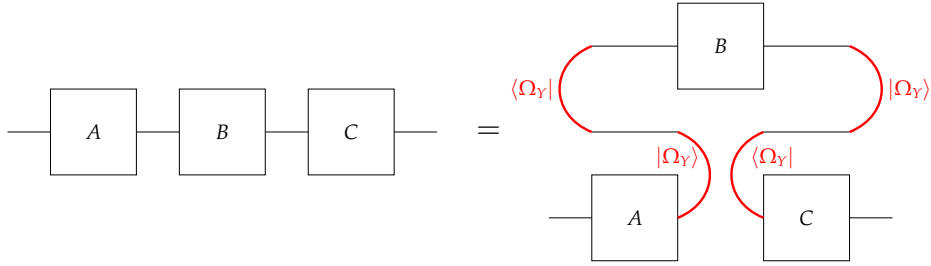
\fi

\begin{lemma}
\label{lem:PSDness}
For any PSD operators $M_{\reg{XY}}, \rho_{\reg{YZ}} \succeq 0$, the operator
$\wt{M}_{\reg{XZ}} :=$ $\Tr_{\reg{Y}}($ $M_{\reg{XY}} \otimes \id_{\reg{Z}} \cdot \id_{\reg{X}} \otimes \ptrans_{\reg{Z}}(\rho_{\reg{YZ}}) )$ is PSD.
\end{lemma}
\ifllncs
\else
\begin{proof}
Since a PSD operator has non-negative eigenvalues and the sum of PSD operators remains PSD, it suffices to consider the case where $M_{\reg{XY}}$ and $\rho_{\reg{YZ}}$ are both rank-$1$ projections. \\

\noindent Suppose $M_{\reg{XY}} = \projector{\psi}_{\reg{XY}}$ and $\rho_{\reg{YZ}} = \projector{\phi}_{\reg{YZ}}$. Consider the Schmidt decomposition of $\ket{\psi}_{\reg{XY}}$:
\begin{align*}
    \ket{\psi}_{\reg{XY}} = \sum_{i=1}^r \lambda_i \ket{u_i}_{\reg{X}} \ket{v_i}_{\reg{Y}}.
\end{align*}
Now, extend the set of orthonormal vectors $\set{\ket{v_i}}_{i\in[r]}$ to an orthonormal basis $\set{\ket{v_i}}_{i\in[d_Y]}$ of $\cH_\reg{Y}$, where $d_Y := \dim(\cH_\reg{Y})$. Next, expand $\ket{\phi}_{\reg{YZ}}$ as
\begin{align*}
    \ket{\phi}_{\reg{YZ}} = \sum_{i\in[d_Y],j\in[d_Z]} \alpha_{ij} \ket{v_i}_{\reg{Y}} \ket{j}_{\reg{Z}},
\end{align*}
where $d_Z := \dim(\cH_\reg{Z})$ and $\set{\ket{j}}_{i\in[d_Z]}$ is the computational basis of $\cH_\reg{Z}$. \\

\noindent A direct calculation yields
\begin{align*}
    \wt{M}_{\reg{XZ}} 
    & = \sum_{i,i'\in[r]} \sum_{j,j'\in[d_Z]} \lambda_i \alpha^*_{ij} \lambda^*_{i'} \alpha_{i'j'} \ketbra{u_i}{u_{i'}}_{\reg{X}} \otimes \ketbra{j}{j'}_{\reg{Z}} \\
    & = \qty( \sum_{i\in[r],j\in[d_Z]} \lambda_i \alpha^*_{ij} \ket{u_i}_{\reg{X}} \ket{j}_{\reg{Z}} ) 
    \qty( \sum_{i' \in [r], j' \in [d_Z]} \lambda^*_{i'} \alpha_{i'j'} \bra{u_{i'}}_{\reg{X}} \bra{j'}_{\reg{Z}} ).
\end{align*}
Since $\wt{M}_{\reg{XZ}}$ can be written as $\projector{\xi}$ for some vector $\ket{\xi}$, it is PSD.
\end{proof}
\fi

Finally, we have the following main lemma.
\begin{lemma}
\label{lem:iden_vs_indep_twirl}
Let $N, t \in \N$ such that $(2t)^2 \leq N$ and let $\hilbert_{\reg{A_1}} = \hilbert_{\reg{B_1}} = (\C^N)^{\otimes t}$ and $\hilbert_{\reg{A_2}}, \hilbert_{\reg{B_2}}$ finite-dimensional Hilbert spaces. For any PSD operators $\rho_{\reg{A_1A_2}}$, $\sigma_{\reg{B_1B_2}}$, $M_{\reg{A_1A_2}}$, $N_{\reg{B_1B_2}}$ $\succeq 0$, it holds that
\begin{align*}
    & \Ex_{U \sim \mu(N)} \qty[ \Tr \qty( U^{\otimes t}_{\reg{A_1}} \rho_{\reg{A_1A_2}} U^{\dagger, \otimes t}_{\reg{A_1}} \cdot M_{\reg{A_1A_2}} ) ] 
    \cdot \Ex_{V \sim \mu(N)}\qty[ \Tr \qty( V^{\otimes t}_{\reg{B_1}} \sigma_{\reg{B_1B_2}} V^{\dagger, \otimes t}_{\reg{B_1}} \cdot N_{\reg{B_1B_2}} ) ] \\
    \leq & \,
    \frac{(1+\veps)^2}{1-\veps} \cdot \Ex_{U\sim\mu(N)} \qty[ \Tr( U^{\otimes t}_{\reg{A_1}} \rho_{\reg{A_1A_2}} U^{\dagger, \otimes t}_{\reg{A_1}} \otimes U^{\otimes t}_{\reg{B_1}} \sigma_{\reg{B_1B_2}} U^{\dagger, \otimes t}_{\reg{B_1}} \cdot M_{\reg{A_1A_2}} \otimes N_{\reg{B_1B_2}} ) ],
\end{align*}
where $\veps := (2t)^2/N$.
\end{lemma}
\begin{proof}[Proof of~\Cref{lem:iden_vs_indep_twirl}]
By using~\Cref{lem:multiplicative_twirling} twice, we have
\begin{align}
\label{eq:indep_twirling}
    & \Ex_{U \sim \mu(N)} \qty[ U^{\otimes t}_{\reg{A_1}} \rho_{\reg{A_1A_2}} U^{\dagger, \otimes t}_{\reg{A_1}} \cdot M_{\reg{A_1A_2}} ) ] \otimes \Ex_{V \sim \mu(N)}\qty[  V^{\otimes t}_{\reg{B_1}} \sigma_{\reg{B_1B_2}} V^{\dagger, \otimes t}_{\reg{B_1}} \cdot N_{\reg{B_1B_2}} ] \nonumber \\ 
    & \preceq  \frac{(1+\veps)^2}{N^{2t}} \sum_{\pi_A \in \Symgp_t} P_N(\pi_A)_{\reg{A_1}} \otimes \Tr_{\reg{A_1}}(P_N(\pi_A)^\intercal_{\reg{A_1}} \cdot \rho_{\reg{A_1A_2}}) \nonumber \\
    & \hspace{.3\textwidth} \otimes \sum_{\pi_B \in \Symgp_t} P_N(\pi_B)_{\reg{B_1}} \otimes \Tr_{\reg{B_1}}(P_N(\pi_B)^\intercal_{\reg{B_1}} \cdot \sigma_{\reg{B_1B_2}}).
\end{align}
Next, we have the following claim to relate~\Cref{eq:indep_twirling} and~\Cref{lem:multiplicative_twirling}:
\begin{myclaim}
\label{claim:perms_ineq}
\ifllncs
\else
\begin{align*}
    & \mathrm{Tr}\Bigg( \sum_{\substack{\pi_A \in \Symgp_t \\ \pi_B \in \Symgp_t}}
    P_N(\pi_A)_{\reg{A_1}} \otimes P_N(\pi_B)_{\reg{B_1}} \otimes \Tr_{\reg{A_1B_1}}( P_N(\pi_A)^\intercal_{\reg{A_1}} \otimes P_N(\pi_B)^\intercal_{\reg{B_1}} \cdot \rho_{\reg{A_1A_2}} \otimes \sigma_{\reg{B_1B_2}}) \cdot M_{\reg{A_1A_2}} \otimes N_{\reg{B_1B_2}} \Bigg) \\ 
    & \leq \Tr( \sum_{\pi \in \Symgp_{2t}} P_N(\pi)_{\reg{A_1B_1}} \otimes \Tr_{\reg{A_1B_1}}(P_N(\pi)^\intercal_{\reg{A_1B_1}} \cdot \rho_{\reg{A_1A_2}} \otimes \sigma_{\reg{B_1B_2}}) \cdot M_{\reg{A_1A_2}} \otimes N_{\reg{B_1B_2}} ).
\end{align*}
\fi
\end{myclaim}
\begin{proof}[Proof of~\Cref{claim:perms_ineq}]
First, let $\hilbert_{\reg{A'_1}} = \hilbert_{\reg{A_1}}, \hilbert_{\reg{B'_1}} = \hilbert_{\reg{B_1}}$, and define the operators
\begin{align*}
\wt{M}_{\reg{A_1A'_1}} & := \Tr_{\reg{A_2}}( M_{\reg{A_1A_2}}\otimes\id_{\reg{A'_1}} \cdot \id_{\reg{A_1}} \otimes \ptrans_{\reg{A'_1}}(\rho_{\reg{A'_1A_2}}) ) \\
\wt{N}_{\reg{B_1B'_1}} & := \Tr_{\reg{B_2}}( N_{\reg{B_1B_2}}\otimes\id_{\reg{B'_1}} \cdot  \id_{\reg{B_1}} \otimes \ptrans_{\reg{B'_1}}(\sigma_{\reg{B'_1B_2}}) ).
\end{align*}
From~\Cref{lem:PSDness}, we have $\wt{M}_{\reg{A_1A'_1}}, \wt{N}_{\reg{B_1B'_1}} \succeq 0$. Moreover, for every $\pi \in \Symgp_{2t}$,
\ifllncs
\else
one can verify either using direct calculation or tensor network diagrams (see~\Cref{fig:equivalence} for a graphical proof) that
\fi

\begin{align}
\label{eq:tensor_network_1}
& \Tr( P_N(\pi)_{\reg{A_1B_1}} \otimes \Tr_{\reg{A_1B_1}}(P_N(\pi)^\intercal_{\reg{A_1B_1}} \cdot \rho_{\reg{A_1A_2}} \otimes \sigma_{\reg{B_1B_2}}) \cdot M_{\reg{A_1A_2}} \otimes N_{\reg{B_1B_2}} ) \\
\label{eq:tensor_network_2}
& = \Tr( P_N(\pi)_{\reg{A_1B_1}} \otimes P_N(\pi)_{\reg{A'_1B'_1}} \cdot \wt{M}_{\reg{A_1A'_1}} \otimes \wt{N}_{\reg{B_1B'_1}} ) \\
\label{eq:tensor_network_3}
& = \Tr( P_{N^2}(\pi)_{\reg{A_1A'_1B_1B'_1}} \cdot \wt{M}_{\reg{A_1A'_1}} \otimes \wt{N}_{\reg{B_1B'_1}} ),
\end{align}
where we equivalently view $P_N(\pi)_{\reg{A_1B_1}} \otimes P_N(\pi)_{\reg{A'_1B'_1}}$ as $P_{N^2}(\pi)_{\reg{A_1A'_1B_1B'_1}}$ by re-ordering the registers to obtain the last equality.

\ifllncs
\else
\begin{figure}[ht]
    \centering
\begin{subfigure}[t]{0.45\textwidth}
\centering
\resizebox{1\textwidth}{!}{%
\begin{circuitikz}
\tikzstyle{every node}=[font=\LARGE]
\draw  (1.25,9) rectangle (3.75,4);
\draw  (6.25,10.75) rectangle (7.5,8.25);
\draw  (6.25,4.75) rectangle (7.5,2.25);
\node [font=\LARGE] at (7.25,9.5) {};
\node [font=\LARGE] at (7,9.5) {$\rho$};
\node [font=\LARGE] at (6.75,3.5) {$\sigma$};
\node [font=\LARGE] at (2.5,6.5) {$P_N(\pi)$};
\draw  (12.5,10.75) rectangle (15,2.25);
\node [font=\LARGE] at (13.75,6.5) {$P_N(\pi)$};
\draw  (9.25,10.75) rectangle (10.5,8.25);
\draw  (9.25,4.75) rectangle (10.5,2.25);
\node [font=\LARGE] at (10.25,9.5) {};
\node [font=\LARGE] at (9.75,9.5) {$M$};
\node [font=\LARGE] at (9.75,3.5) {$N$};
\draw [short] (3.75,8.25) .. controls (5.5,10.5) and (2,11) .. (-1.25,10.25);
\draw [short] (3.75,4.5) .. controls (5.5,2.5) and (2,1.75) .. (-1.25,2.75);
\draw [short] (1.25,8.25) .. controls (-0.5,10.5) and (3,11.25) .. (6.25,10.25);
\draw [short] (1.25,4.5) .. controls (-0.5,2.5) and (3,1.75) .. (6.25,2.75);
\draw [short] (-1.25,10.25) .. controls (-3.5,12.5) and (9.75,12.75) .. (7.5,10.25);
\draw [short] (-1.25,2.75) .. controls (-3.5,0.25) and (9.75,0.25) .. (7.5,2.75);
\draw [short] (7.5,4) .. controls (7.75,4) and (8.5,4) .. (9.25,4);
\node [font=\LARGE] at (8.75,8.5) {};
\node [font=\LARGE] at (8.75,8.5) {};
\draw [short] (7.5,9) .. controls (7.75,9) and (8.5,9) .. (9.25,9);
\draw [short] (6.25,5.25) .. controls (7.75,5.25) and (8.5,5.25) .. (10.5,5.25);
\node [font=\LARGE] at (9.25,3) {};
\node [font=\LARGE] at (9.25,3) {};
\node [font=\LARGE] at (9.25,3) {};
\node [font=\LARGE] at (8.75,5.25) {};
\draw [short] (6.25,7.75) .. controls (7.75,7.75) and (8.5,7.75) .. (10.5,7.75);
\draw [short] (6.25,9) .. controls (5,8.75) and (5,8) .. (6.25,7.75);
\draw [short] (6.25,5.25) .. controls (5,5) and (5,4.25) .. (6.25,4);
\draw [short] (10.5,7.75) .. controls (11.75,8) and (11.75,8.75) .. (10.5,9);
\draw [short] (10.5,4) .. controls (11.75,4.25) and (11.75,5) .. (10.5,5.25);
\node [font=\LARGE] at (8.75,7.25) {};
\node [font=\LARGE] at (8.75,7.25) {};
\draw [short] (10.5,10.25) .. controls (11,10.25) and (11.5,10.25) .. (12.5,10.25);
\draw [short] (10.5,2.75) .. controls (11,2.75) and (11.5,2.75) .. (12.5,2.75);
\draw [short] (9.25,1.5) .. controls (11.5,1.5) and (12.25,1.5) .. (15,1.5);
\draw [short] (9.25,11.5) .. controls (11.5,11.5) and (12.25,11.5) .. (15,11.5);
\draw [short] (9.25,11.5) .. controls (8,11.25) and (8,10.5) .. (9.25,10.25);
\draw [short] (9.25,2.75) .. controls (8,2.5) and (8,1.75) .. (9.25,1.5);
\draw [short] (15,10.25) .. controls (16.25,10.5) and (16.25,11.25) .. (15,11.5);
\draw [short] (15,1.5) .. controls (16.25,1.75) and (16.25,2.5) .. (15,2.75);
\node [font=\LARGE] at (0.75,8.25) {$A_1$};
\node [font=\LARGE] at (0.75,4.75) {$B_1$};
\node [font=\LARGE] at (5.75,10.75) {$A_1$};
\node [font=\LARGE] at (11.5,10.75) {$A_1$};
\node [font=\LARGE] at (11.5,3.25) {$B_1$};
\node [font=\LARGE] at (5.75,3.25) {$B_1$};
\node [font=\LARGE] at (8.25,4.5) {$B_2$};
\node [font=\LARGE] at (8.25,8.5) {$A_2$};
\end{circuitikz}
}%
\caption{Tensor network diagram of~\Cref{eq:tensor_network_1}.}
\label{fig:LHS}
\end{subfigure}
\hfill
\begin{subfigure}[t]{0.45\textwidth}
\centering
\resizebox{1\textwidth}{!}{%
\begin{circuitikz}
\tikzstyle{every node}=[font=\LARGE]
\draw  (17.5,12.75) rectangle (20,0.25);
\draw [ color={rgb,255:red,255; green,0; blue,0} , line width=2pt ] (5,10.75) rectangle (6.25,8.25);
\draw [ color={rgb,255:red,0; green,0; blue,255} , line width=2pt ] (5,4.75) rectangle (6.25,2.25);
\node [font=\LARGE] at (6,9.5) {};
\node [font=\LARGE, color={rgb,255:red,255; green,0; blue,0}] at (5.75,9.5) {$\rho$};
\node [font=\LARGE, color={rgb,255:red,0; green,0; blue,255}] at (5.5,3.5) {$\sigma$};
\node [font=\LARGE] at (18.75,6.5) {$P_N(\pi)$};
\draw  (12.5,10.75) rectangle (15,2.25);
\node [font=\LARGE] at (13.75,6.5) {$P_N(\pi)$};
\draw [ color={rgb,255:red,255; green,0; blue,0} , line width=2pt ] (9.25,10.75) rectangle (10.5,8.25);
\draw [ color={rgb,255:red,0; green,0; blue,255} , line width=2pt ] (9.25,4.75) rectangle (10.5,2.25);
\node [font=\LARGE] at (10.25,9.5) {};
\node [font=\LARGE, color={rgb,255:red,255; green,0; blue,0}] at (9.75,9.5) {$M$};
\node [font=\LARGE, color={rgb,255:red,0; green,0; blue,255}] at (9.75,3.5) {$N$};
\draw [ color={rgb,255:red,0; green,0; blue,255}, line width=2pt, short] (6.25,4) .. controls (7.25,4) and (7.75,4) .. (9.25,4);
\node [font=\LARGE] at (8.75,8.5) {};
\node [font=\LARGE] at (8.75,8.5) {};
\draw [ color={rgb,255:red,255; green,0; blue,0}, line width=2pt, short] (6.25,9) .. controls (7.25,9) and (7.75,9) .. (9.25,9);
\draw [ color={rgb,255:red,0; green,0; blue,255}, line width=2pt, short] (3.75,5.25) .. controls (7,5.25) and (7.5,5.25) .. (11.25,5.25);
\node [font=\LARGE] at (9.25,3) {};
\node [font=\LARGE] at (9.25,3) {};
\node [font=\LARGE] at (9.25,3) {};
\node [font=\LARGE] at (8.75,5.25) {};
\draw [ color={rgb,255:red,255; green,0; blue,0}, line width=2pt, short] (3.75,7.75) .. controls (7,7.75) and (7.5,7.75) .. (11.25,7.75);
\draw [ color={rgb,255:red,255; green,0; blue,0}, line width=2pt, short] (3.75,9) .. controls (2.5,8.75) and (2.5,8) .. (3.75,7.75);
\draw [ color={rgb,255:red,0; green,0; blue,255}, line width=2pt, short] (3.75,5.25) .. controls (2.5,5) and (2.5,4.25) .. (3.75,4);
\draw [ color={rgb,255:red,255; green,0; blue,0}, line width=2pt, short] (11.25,7.75) .. controls (12.5,8) and (12.5,8.75) .. (11.25,9);
\draw [ color={rgb,255:red,0; green,0; blue,255}, line width=2pt, short] (11.25,4) .. controls (12.5,4.25) and (12.5,5) .. (11.25,5.25);
\node [font=\LARGE] at (8.75,7.25) {};
\node [font=\LARGE] at (8.75,7.25) {};
\draw [short] (8.75,1.5) .. controls (11.25,1.5) and (12,1.5) .. (15,1.5);
\draw [short] (8.75,11.5) .. controls (11.25,11.5) and (12,11.5) .. (15,11.5);
\draw [short] (8.75,11.5) .. controls (7.5,11.25) and (7.5,10.5) .. (8.75,10.25);
\draw [short] (8.75,2.75) .. controls (7.5,2.5) and (7.5,1.75) .. (8.75,1.5);
\draw [short] (15,10.25) .. controls (16.25,10.5) and (16.25,11.25) .. (15,11.5);
\draw [short] (15,1.5) .. controls (16.25,1.75) and (16.25,2.5) .. (15,2.75);
\node [font=\LARGE] at (20.75,11.75) {$A'_1$};
\node [font=\LARGE] at (20.75,1.5) {$B'_1$};
\node [font=\LARGE] at (3.75,10.25) {$A'_1$};
\node [font=\LARGE] at (11.5,10.75) {$A_1$};
\node [font=\LARGE] at (11.5,3.25) {$B_1$};
\node [font=\LARGE] at (3.75,2.75) {$B'_1$};
\node [font=\LARGE] at (8.25,4.5) {$B_2$};
\node [font=\LARGE] at (8.25,8.5) {$A_2$};

\draw [short] (3.75,1) .. controls (2.5,0.5) and (2.5,-0.5) .. (3.75,-1);
\draw [short] (3.75,14) .. controls (2.5,13.5) and (2.5,12.5) .. (3.75,12);
\draw [ color={rgb,255:red,0; green,0; blue,255}, line width=2pt, short] (5,2.75) .. controls (3.75,2) and (3.75,1.25) .. (5,0.5);
\draw [ color={rgb,255:red,255; green,0; blue,0}, line width=2pt, short] (5,12.5) .. controls (3.75,11.75) and (3.75,11) .. (5,10.25);
\draw [ color={rgb,255:red,255; green,0; blue,0}, line width=2pt, short] (6.25,10.25) .. controls (7.5,10.75) and (7.5,11.5) .. (6.25,12);
\draw [ color={rgb,255:red,255; green,0; blue,0}, , line width=2pt](3.75,12) to[short] (6.25,12);
\draw [ color={rgb,255:red,255; green,0; blue,0}, , line width=2pt](5,12.5) to[short] (17.5,12.5);
\draw [short] (20,12.5) .. controls (21.25,12.75) and (21.25,13.75) .. (20,14);
\draw (3.75,14) to[short] (20,14);
\draw [ color={rgb,255:red,0; green,0; blue,255}, , line width=2pt](5,0.5) to[short] (17.5,0.5);
\draw [short] (20,-1) .. controls (21.25,-0.5) and (21.25,0.25) .. (20,0.75);
\draw (3.75,-1) to[short] (20,-1);
\draw [ color={rgb,255:red,0; green,0; blue,255}, line width=2pt, short] (6.25,1) .. controls (7.5,1.5) and (7.5,2.25) .. (6.25,2.75);
\draw [ color={rgb,255:red,0; green,0; blue,255}, , line width=2pt](3.75,1) to[short] (6.25,1);
\draw [ color={rgb,255:red,255; green,0; blue,0}, , line width=2pt](10.5,9) to[short] (11.25,9);
\draw [ color={rgb,255:red,0; green,0; blue,255}, , line width=2pt](10.5,4) to[short] (11.25,4);
\draw [ color={rgb,255:red,255; green,0; blue,0}, , line width=2pt](3.75,9) to[short] (5,9);
\draw [ color={rgb,255:red,0; green,0; blue,255}, , line width=2pt](3.75,4) to[short] (5,4);
\draw [ color={rgb,255:red,255; green,0; blue,0}, , line width=2pt](8.75,10.25) to[short] (9.25,10.25);
\draw [ color={rgb,255:red,0; green,0; blue,255}, , line width=2pt](8.75,2.75) to[short] (9.25,2.75);
\draw [ color={rgb,255:red,255; green,0; blue,0}, , line width=2pt](10.5,10.25) to[short] (11.25,10.25);
\draw (11.25,10.25) to[short] (12.5,10.25);
\draw [ color={rgb,255:red,0; green,0; blue,255}, , line width=2pt](10.5,2.75) to[short] (11.25,2.75);
\draw (11.25,2.75) to[short] (12.5,2.75);
\end{circuitikz}
}%
\caption{Tensor network diagram of~\Cref{eq:tensor_network_2}. The \textcolor{red}{red} part represents $\wt{M}_{\reg{A_1A'_1}}$ and the \textcolor{blue}{blue} part represents $\wt{N}_{\reg{B_1B'_1}}$.}
\label{fig:RHS}
\end{subfigure}
\caption{A graphical expression of~\Cref{eq:tensor_network_1,eq:tensor_network_2}.}
\label{fig:equivalence}
\end{figure}
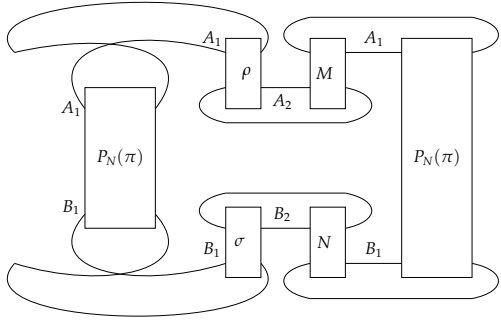
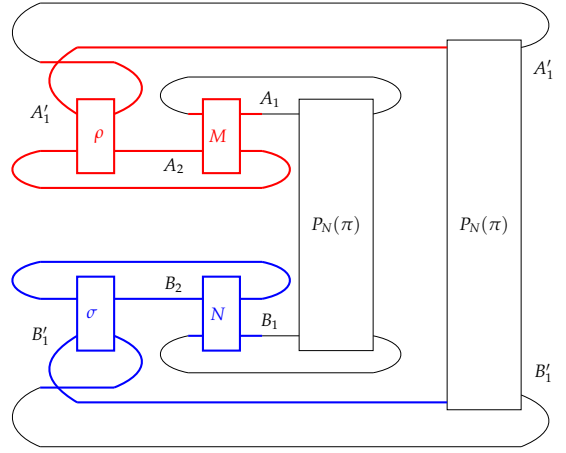
\fi

\noindent Summing over $\pi \in \Symgp_{2t}$ and $(\pi_A,\pi_B) \in \Symgp_t \times \Symgp_t$ in~\Cref{eq:tensor_network_3} respectively and invoking~\Cref{lem:product_perm} completes the proof of~\Cref{claim:perms_ineq}.
\end{proof}
\ \\
\noindent Finally, we have
\ifllncs
\else
\begin{align*}
    & \Ex_{U \sim \mu(N)} \qty[ \Tr \qty( U^{\otimes t}_{\reg{A_1}} \rho_{\reg{A_1A_2}} U^{\dagger, \otimes t}_{\reg{A_1}} \cdot M_{\reg{A_1A_2}} ) ] \cdot \Ex_{V \sim \mu(N)}\qty[ \Tr \qty( V^{\otimes t}_{\reg{B_1}} \sigma_{\reg{B_1B_2}} V^{\dagger, \otimes t}_{\reg{B_1}} \cdot N_{\reg{B_1B_2}} ) ] \\
    & \leq
    \frac{(1+\veps)^2}{N^{2t}} \mathrm{Tr}\Bigg( \sum_{\substack{\pi_A \in \Symgp_t \\ \pi_B \in \Symgp_t}}
    P_N(\pi_A)_{\reg{A_1}} \otimes P_N(\pi_B)_{\reg{B_1}} \otimes \Tr_{\reg{A_1B_1}}( P_N(\pi_A)^\intercal_{\reg{A_1}} \otimes P_N(\pi_B)^\intercal_{\reg{B_1}} \cdot \rho_{\reg{A_1A_2}} \otimes \sigma_{\reg{B_1B_2}}) \cdot M_{\reg{A_1A_2}} \otimes N_{\reg{B_1B_2}} \Bigg) \\
    & \leq
    \frac{(1+\veps)^2}{N^{2t}} \mathrm{Tr} \Bigg( \sum_{\pi \in \Symgp_{2t}} P_N(\pi)_{\reg{A_1B_1}} \otimes \Tr_{\reg{A_1B_1}}(P_N(\pi)^\intercal_{\reg{A_1B_1}} \cdot \rho_{\reg{A_1A_2}} \otimes \sigma_{\reg{B_1B_2}}) \cdot M_{\reg{A_1A_2}} \otimes N_{\reg{B_1B_2}} \Bigg) \\
    & \leq \frac{(1+\veps)^2}{1-\veps} \cdot \Ex_{U\sim\mu(N)}\qty[ \Tr( U^{\otimes t}_{\reg{A_1}} \rho_{\reg{A_1A_2}} U^{\dagger, \otimes t}_{\reg{A_1}} \otimes U^{\otimes t}_{\reg{B_1}} \sigma_{\reg{B_1B_2}} U^{\dagger, \otimes t}_{\reg{B_1}} \cdot M_{\reg{A_1A_2}} \otimes N_{\reg{B_1B_2}} ) ],
\end{align*}
\fi
where the first inequality follows from~\Cref{eq:indep_twirling} and~\Cref{fact:PSD_trace}; the second inequality follows from~\Cref{claim:perms_ineq}; the last inequality follows from~\Cref{lem:multiplicative_twirling} and~\Cref{fact:PSD_trace}. This completes the proof of~\Cref{lem:iden_vs_indep_twirl}.
\end{proof}

\begin{theorem}
\label{thm:adap_LOCC}
For any two-party adversary $(\alice^{(\cdot)},\bob^{(\cdot)})$, where $\alice$ and $\bob$ do not share any entanglement initially, and each makes $t$ queries,
\begin{align*}
    \Pr[\mathsf{HUD}(N, \alice, \bob) = 1] 
    \leq \frac{1}{2} + \frac{3\veps+\veps^2}{2(1+\veps)^2} 
    = \frac{1}{2} + O\qty(\frac{t^2}{N}),
\end{align*}
where $\veps := \frac{(2t)^2}{N}$.
\end{theorem}
\begin{proof}[Proof of~\Cref{thm:adap_LOCC}]
For short, let $p_{b'|b}$ be the probability that $(\alice,\bob)$ outputs $b'$ conditioned on $\challenger$ picks $b$ as the challenge. It is sufficient to show that the success probability $(p_{0|0} + p_{1|1})/2$ is at most $\frac{1}{2} + \frac{3\veps+\veps^2}{2(1+\veps)^2}$. Suppose $(\alice, \bob)$ starts with the initial state $\projector{0}_{\regA} \otimes \projector{0}_{\regB}$, performs an LOCC channel $\cE^{U,V}$ of the form~\Cref{eq:normal_form}, and performs a two-outcome POVM of the form $\set{ M^0_{\regA} \otimes N^0_{\regB}, M^1_{\regA} \otimes N^1_{\regB} }$ such that $M^0_{\regA}$, $N^0_{\regB}$, $M^1_{\regA}$, $N^1_{\regB} \succeq 0$ and $M^0_{\regA} \otimes N^0_{\regB} + M^1_{\regA} \otimes N^1_{\regB} = \id$ to output $b' \in \bit$. \\

\noindent First, using the idea from~\cite{AMR20} (see also~\cite{Kretschmer21,MPSY24,SHH25}), which utilizes gate teleportation and post-selection, we can express each $p_{b'|b}$ in the following form:

\begin{lemma}
\label{lem:adaptive-to-selective}
There exist PSD operators $P^0_{A}, P^1_{A}, Q^0_{B}$, and $Q^1_{B}$, an alphabet $\Sigma$, and a collection of PSD operators $\set{\rho_a \otimes \sigma_a:\ \rho_a, \sigma_a \succeq 0}_{a \in \Sigma}$ such that
\begin{align*}
    & p_{0|0} = \sum_{a \in \Sigma} \Ex_{U \sim \haarunitaries(N)} \qty[ \Tr( U^{\otimes t} \rho_a U^{\dagger, \otimes t} \otimes U^{\otimes t} \sigma_a U^{\dagger, \otimes t} \cdot P^0_{A} \otimes Q^0_{B} ) ], \\
    & p_{0|1} = \sum_{a \in \Sigma} \Ex_{U \sim \haarunitaries(N)} \qty[ \Tr( U^{\otimes t} \rho_a U^{\dagger, \otimes t} \cdot P^0_{A} ) ] \cdot \Ex_{V \sim \haarunitaries(N)} \qty[ \Tr\qty( V^{\otimes t} \sigma_a V^{\dagger, \otimes t} \cdot Q^0_{B} ) ], \\
    & p_{1|0} = \sum_{a \in \Sigma} \Ex_{U \sim \haarunitaries(N)} \qty[ \Tr( U^{\otimes t} \rho_a U^{\dagger, \otimes t} \otimes U^{\otimes t} \sigma_a U^{\dagger, \otimes t} \cdot P^1_{A} \otimes Q^1_{B} ) ], \\
    & p_{1|1} = \sum_{a \in \Sigma} \Ex_{U \sim \haarunitaries(N)} \qty[ \Tr( U^{\otimes t} \rho_a U^{\dagger, \otimes t} \cdot P^1_{A} ) ] \cdot \Ex_{V \sim \haarunitaries(N)} \qty[ \Tr\qty( V^{\otimes t} \sigma_a V^{\dagger, \otimes t} \cdot Q^1_{B} ) ].
\end{align*}
\end{lemma}
\begin{proof}[Proof of~\Cref{lem:adaptive-to-selective}]
Using~\Cref{eq:normal_form}, we expand $p_{0|0}$ as follows:
\begin{align}
\label{eq:p00}
p_{0|0} = & \sum_{(a_1,\dots,a_t) \in \Sigma_1 \times \dots \times \Sigma_t}  
\Ex_{U \sim \haarunitaries(N)} \Bigg[ \Tr \qty( \prod_{i=t}^1 UA_{i,a_i} \projector{0}_{\regA} \prod_{i'=1}^t A^\dagger_{i',a_{i'}}U^\dagger \cdot M^0_{\regA} ) \nonumber \\
& \hspace{0.2\textwidth} \cdot \Tr \qty( \prod_{j=t}^1 UB_{j,a_j} \projector{0}_{\regB} \prod_{j'=1}^t B^\dagger_{j',a_{j'}}U^\dagger \cdot N^0_{\regB} ) \Bigg].
\end{align}

\noindent Now, fix $\bfa := (a_1,\dots,a_t) \in \Sigma_1 \times \dots \times \Sigma_t$ and $U \in \Unitary(N)$. Applying~\Cref{fact:gate_teleportation} recursively on~\Cref{eq:p00}, we have
\begin{align}
\label{eq:trace_A}
& \Tr \qty( \prod_{i=t}^1 U A_{i,a_i} \projector{0}_{\regA} \prod_{i'=1}^t A^\dagger_{i',a_{i'}} U^\dagger \cdot M^0_{\regA} ) \nonumber \\
= & \mathrm{Tr} \Bigg( \qty( \id \otimes U^{\otimes t} ) \cdot \underbrace{ \projector{0}_{\regA} \otimes 
\bigotimes_{i=1}^t (\id \otimes A_{i,a_i}) \projector{\Omega_i} (\id \otimes A^\dagger_{i,a_i}) }_{ =: \rho_{\bfa}} \cdot \qty( \id \otimes U^{\dagger, \otimes t} ) \cdot \underbrace{ \bigotimes_{i=1}^t \projector{\Omega_i} \otimes M^0_{\regA} }_{ =: P^0_{A}} \Bigg),
\end{align}
where $\ket{\Omega_i}$ denotes the non-normalized maximally entangled state of an appropriate dimension.
\ifllncs
\else
where $\ket{\Omega_i}$ denotes the non-normalized maximally entangled state of an appropriate dimension (see~\Cref{fig:flattening} for an example).

\begin{figure}[!ht]
\centering
\resizebox{.5\textwidth}{!}{%
\begin{circuitikz}
\tikzstyle{every node}=[font=\LARGE]
\draw [ color={rgb,255:red,0; green,0; blue,255} , line width=2pt ] (2.5,-83.5) rectangle (5,-86);
\draw [ color={rgb,255:red,0; green,0; blue,255} , line width=2pt ] (5,-78.5) rectangle (7.5,-81);
\draw [ color={rgb,255:red,0; green,0; blue,255} , line width=2pt ] (7.5,-83.5) rectangle (10,-86);
\draw [ color={rgb,255:red,0; green,0; blue,255}, line width=2pt, short] (5,-84.75) .. controls (6.25,-84.25) and (6.25,-82.75) .. (5,-82.25);
\draw [ color={rgb,255:red,0; green,0; blue,255}, , line width=2pt](1.25,-82.25) to[short] (5,-82.25);
\draw [ color={rgb,255:red,0; green,0; blue,255}, , line width=2pt](1.25,-79.75) to[short] (5,-79.75);
\draw [ color={rgb,255:red,0; green,0; blue,255}, , line width=2pt](7.5,-79.75) to[short] (11.25,-79.75);
\draw [ color={rgb,255:red,0; green,0; blue,255}, , line width=2pt](7.5,-82.25) to[short] (11.25,-82.25);
\draw [ color={rgb,255:red,0; green,0; blue,255}, line width=2pt, short] (7.5,-84.75) .. controls (6.25,-84.25) and (6.25,-82.75) .. (7.5,-82.25);
\draw [ color={rgb,255:red,255; green,0; blue,0}, line width=2pt, short] (-2.5,-82.25) .. controls (-3.75,-81.75) and (-3.75,-80.25) .. (-2.5,-79.75);
\draw [ color={rgb,255:red,255; green,0; blue,0}, line width=2pt, short] (15,-82.25) .. controls (16.25,-81.75) and (16.25,-80.25) .. (15,-79.75);
\node [font=\LARGE] at (8.75,-84.75) {$A_1^\dagger$};
\node [font=\LARGE] at (6.25,-79.75) {$\projector{0}$};
\node [font=\LARGE] at (3.75,-84.75) {$A_1$};

\draw [ color={rgb,255:red,0; green,0; blue,255}, , line width=2pt](1.25,-84.75) to[short] (2.5,-84.75);
\draw [ color={rgb,255:red,0; green,0; blue,255}, , line width=2pt](10,-84.75) to[short] (11.25,-84.75);

\draw [ line width=2pt ] (-1.75,-83.5) rectangle (0.75,-86);
\draw [ line width=2pt ] (11.75,-83.5) rectangle (14.25,-86);
\node [font=\LARGE] at (13,-84.75) {$U^\dagger$};
\node [font=\LARGE] at (-0.5,-84.75) {$U$};

\node [font=\LARGE] at (-3.75,-89.75) {$\sqrt{M^0_A}$};
\node [font=\LARGE] at (16.25,-89.75) {$\sqrt{M^0_A}$};
\draw [ color={rgb,255:red,255; green,0; blue,0} , line width=2pt ] (15,-88.5) rectangle (17.5,-91);
\draw [ color={rgb,255:red,255; green,0; blue,0} , line width=2pt ] (-5,-88.5) rectangle (-2.5,-91);
\draw [short] (-5,-92.25) .. controls (-6.25,-91.75) and (-6.25,-90.25) .. (-5,-89.75);
\draw [short] (17.5,-92.25) .. controls (18.75,-91.75) and (18.75,-90.25) .. (17.5,-89.75);
\draw (-5,-92.25) to[short] (17.5,-92.25);
\draw (11.25,-79.75) to[short] (15,-79.75);
\draw (11.25,-82.25) to[short] (15,-82.25);
\draw (-2.5,-79.75) to[short] (1.25,-79.75);
\draw (-2.5,-82.25) to[short] (1.25,-82.25);
\draw (11.25,-84.75) to[short] (11.75,-84.75);
\draw (14.25,-84.75) to[short] (15,-84.75);
\draw (0.75,-84.75) to[short] (1.25,-84.75);
\draw (-2.5,-84.75) to[short] (-1.75,-84.75);
\draw [ color={rgb,255:red,0; green,0; blue,255} , line width=2pt ] (2.5,-88.5) rectangle (5,-91);
\draw [ color={rgb,255:red,0; green,0; blue,255} , line width=2pt ] (7.5,-88.5) rectangle (10,-91);
\draw [ color={rgb,255:red,0; green,0; blue,255}, line width=2pt, short] (5,-89.75) .. controls (6.25,-89.25) and (6.25,-87.75) .. (5,-87.25);
\draw [ color={rgb,255:red,0; green,0; blue,255}, , line width=2pt](1.25,-87.25) to[short] (5,-87.25);
\draw [ color={rgb,255:red,0; green,0; blue,255}, , line width=2pt](7.5,-87.25) to[short] (11.25,-87.25);
\draw [ color={rgb,255:red,0; green,0; blue,255}, line width=2pt, short] (7.5,-89.75) .. controls (6.25,-89.25) and (6.25,-87.75) .. (7.5,-87.25);
\draw [ color={rgb,255:red,255; green,0; blue,0}, line width=2pt, short] (-2.5,-87.25) .. controls (-3.75,-86.75) and (-3.75,-85.25) .. (-2.5,-84.75);
\draw [ color={rgb,255:red,255; green,0; blue,0}, line width=2pt, short] (15,-87.25) .. controls (16.25,-86.75) and (16.25,-85.25) .. (15,-84.75);
\node [font=\LARGE] at (8.75,-89.75) {$A_2^\dagger$};
\node [font=\LARGE] at (3.75,-89.75) {$A_2$};

\draw [ color={rgb,255:red,0; green,0; blue,255}, , line width=2pt](1.25,-89.75) to[short] (2.5,-89.75);
\draw [ color={rgb,255:red,0; green,0; blue,255}, , line width=2pt](10,-89.75) to[short] (11.25,-89.75);

\draw [ line width=2pt ] (-1.75,-88.5) rectangle (0.75,-91);
\draw [ line width=2pt ] (11.75,-88.5) rectangle (14.25,-91);
\node [font=\LARGE] at (13,-89.75) {$U^\dagger$};
\node [font=\LARGE] at (-0.5,-89.75) {$U$};
\draw (11.25,-87.25) to[short] (15,-87.25);
\draw (-2.5,-87.25) to[short] (1.25,-87.25);
\draw (11.25,-89.75) to[short] (11.75,-89.75);
\draw (14.25,-89.75) to[short] (15,-89.75);
\draw (0.75,-89.75) to[short] (1.25,-89.75);
\draw (-2.5,-89.75) to[short] (-1.75,-89.75);
\end{circuitikz}
}%
\caption{Tensor network diagram of~\Cref{eq:trace_A} for $r = 2$. The \textcolor{blue}{blue} part represents $\rho_\bfa$, and the \textcolor{red}{red} part represents $P^0_A$.}
\label{fig:flattening}
\end{figure}
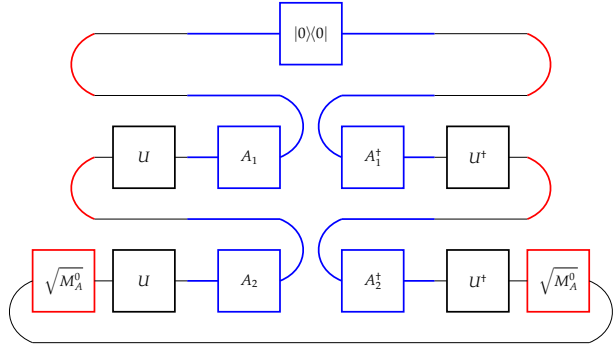
\fi

\ \\
\noindent Similarly, we have
\begin{align}
\label{eq:trace_B}
& \Tr \qty( \prod_{j = t}^1 U B_{j,a_j} \projector{0}_{\regA} \prod_{j' = 1}^t B^\dagger_{j',a_{j'}} U^\dagger \cdot N^0_{\regB} ) \nonumber \\
& = \mathrm{Tr} \Bigg( (\id \otimes U^{\otimes t}) \cdot \underbrace{ \projector{0}_{\regA} \otimes 
\bigotimes_{j=1}^t (\id \otimes B_{j,a_j}) \projector{\Omega_j} (\id \otimes B^\dagger_{j,a_j}) }_{ =: \sigma_{\bfa}} \cdot (\id \otimes U^{\dagger, \otimes t}) \cdot \underbrace{ \bigotimes_{j=1}^t \projector{\Omega_j} \otimes N^0_{\regB} }_{ =: Q^0_{B}} \Bigg).
\end{align}

\noindent Plugging~\Cref{eq:trace_A,eq:trace_B} to~\Cref{eq:p00}, we have
\begin{align*}
    p_{0|0} & = 
    \sum_{\bfa} \Ex_{U \sim \haarunitaries(N)} \qty[ \Tr \qty( U^{\otimes t} \rho_\bfa U^{\dagger, \otimes t} \cdot P^0_A ) \cdot \Tr \qty( U^{\otimes t} \sigma_\bfa U^{\dagger, \otimes t} \cdot Q^0_B ) ] \\
    & = \sum_{\bfa} \Ex_{U \sim \haarunitaries(N)} \qty[ \Tr \qty( U^{\otimes t} \rho_\bfa U^{\dagger, \otimes t} \otimes U^{\otimes t} \sigma_\bfa U^{\dagger, \otimes t} \cdot P^0_A \otimes Q^0_B ) ].
\end{align*}
The cases of $p_{0|1}, p_{1|0}$, and $p_{1|1}$ can be proven similarly. This completes the proof of~\Cref{lem:adaptive-to-selective}.
\end{proof}

\noindent From~\Cref{lem:adaptive-to-selective,lem:iden_vs_indep_twirl}, we have the following linear program:
\begin{equation}
\label{eq:LP}
\begin{aligned}
    \text{maximize} \quad & \frac{p_{0|0} + p_{1|1}}{2} \\
    \text{subject to} \quad 
    & p_{0|0}, p_{0|1}, p_{1|0}, p_{1|1} \geq 0, \\
    & p_{0|0} + p_{1|0} =  p_{0|1} + p_{1|1} = 1, \\
    & \frac{(1+\veps)^2}{1-\veps} p_{0|0} \geq p_{0|1}, \\
    & \frac{(1+\veps)^2}{1-\veps} p_{1|0}\geq p_{1|1}.
\end{aligned}
\end{equation}
Solving~\eqref{eq:LP} we obtain that the maximum is $\frac{1}{2} + \frac{3\veps+\veps^2}{2(1+\veps)^2}$. This completes the proof of~\Cref{thm:adap_LOCC}.
\end{proof}

\subsection{LOCC Indistinguishability against Non-Adaptive and Inverse Queries}
In this section, we consider the distinguishing game in~\Cref{def:NAInvHUDgame}.

\begin{definition}[Non-Adaptive Invertible Haar Unitary Distinguishing Game]
\label{def:NAInvHUDgame}
A \emph{non-adaptive invertible Haar unitary distinguishing game}, parametrized by dimension
$N \in \N$, is played by the challenger $\challenger$ and a two-party adversary $(\alice^{(\cdot)}, \bob^{(\cdot)})$.
\begin{protocol}{$\underline{\NAInvHUD(N, \alice, \bob)}$}
{
\vspace{.5em}
\ \\
\noindent {\bf{Phase~1:}}
\begin{enumerate}
    \item $\challenger$ samples two independent Haar unitaries $U,V\sim\haarunitaries(N)$.
    \item $\challenger$ samples a challenge bit $b \gets \bit$. If $b = 0$, $\alice$ and $\bob$ are both given oracle access to $U$ {\bf and its inverse $U^\dagger$}. Otherwise, $\alice$ is given oracle access to $U$ and {\bf its inverse $U^\dagger$}, and $\bob$ is given oracle access to $V$ {\bf and its inverse $V^\dagger$}.
    \item $\alice$ and $\bob$ each makes {\bf one round of non-adaptive queries}.
\end{enumerate}
\noindent {\bf{Phase~2:}}
\begin{enumerate}
    \item $\alice$ and $\bob$ {\bf both lose oracle access}. They can perform perform an arbitrary finite number of rounds of classical communication.
    \item $(\alice,\bob)$ outputs a guess $b'\in\bit$.
    \item $\challenger$ outputs $1$ if and only if $b' = b$. 
\end{enumerate}
}
\end{protocol}
\end{definition}

\noindent The main technical tool is the following approximation formula (with additive error) for \emph{mixed twirling}.
\begin{theorem}[Mixed twirling approximation]
\label{thm:mixed_twirling}
Let $k,N \in \N$, $\hilbert_{\regA} = \hilbert_{\regB} = (\C^N)^{\otimes k}$, and $\hilbert_{\regC}$ an arbitrary Hilbert space. For any density matrix $\rho_{\reg{ABC}}$, define
\begin{align*}
    \rho^{\sf twirl}_{\reg{ABC}} := \Ex_{U \sim \haarunitaries(N)}[(U^{\otimes k}_{\reg{A}}\otimes U^{\dagger, \otimes k}_{\reg{B}}) \cdot \rho_{\reg{ABC}} \cdot (U^{\otimes k}_{\reg{A}} \otimes U^{\dagger, \otimes k}_{\reg{B}})^\dagger]
\end{align*}
\begin{align*}
    \rho^{\sf approx}_{\reg{ABC}} := \frac{1}{N^{2k}} \sum_{\pi, \tau \in \Symgp_k} 
    P_N(\pi)_{\regA} \otimes P_N(\tau)_{\regB} \otimes
    \Tr_{\reg{AB}}( P_N(\pi)^\intercal_{\regA} \otimes P_N(\tau)^\intercal_{\regB} \cdot \rho_{\reg{ABC}} ).
\end{align*}
Then $\TD(\rho^{\sf twirl}_{\reg{ABC}}, \rho^{\sf approx}_{\reg{ABC}}) \leq O \qty( k^2/N^{1/8} )$.
\end{theorem}
\begin{proof}
Let $\reg{A} := (\reg{A_1},\dots,\reg{A_k})$, $\reg{B} := (\reg{B_1},\dots,\reg{B_k})$ such that $\hilbert_{\reg{A_i}} = \hilbert_{\reg{B_i}} = \C^N$ for all $i \in [k]$. By convexity, it is sufficient to consider the case when $\rho_{\reg{ABC}}$ is pure, \ie $\rho_{\reg{ABC}} = \projector{\psi}_{\reg{ABC}}$. Recall the definition of the path recording isometry $V$:
\begin{align*}
    V & = V^L \cdot \qty( \id - V^R \cdot V^{R, \dagger} ) + \qty( \id - V^L \cdot V^{L, \dagger} ) \cdot V^{R,\dagger} \\
    V^{\dagger} & = V^R \cdot \qty( \id - V^L \cdot V^{L, \dagger} ) + \qty( \id - V^R \cdot V^{R, \dagger} ) \cdot V^{L,\dagger}.
\end{align*}

\noindent Consider the following sequence of operators defined over registers $\reg{ABC}$, where the change from the previous one is highlighted in \textcolor{red}{red}.

\newhybrid{1} Let $\rho_1 := \rho^{\sf twirl}_{\reg{ABC}}$.

\newhybrid{2} We replace $(U,U^\dagger)$ in~\hybrid{1} with $(V,V^\dagger)$. Define the following subnormalized state
\begin{align*}
    \ket{\Adversary^{V,V^\dagger}_k}_{\reg{ABCLR}}
    & := \prod_{j = k}^1 V^{\dagger}_{\reg{B_j LR}}
    \cdot \prod_{i = k}^1 V_{\reg{A_i LR}}
    \cdot \ket{\psi}_{\reg{ABC}} 
    \ket{\varnothing}_{\reg{L}} 
    \ket{\varnothing}_{\reg{R}},
\end{align*}
and let $\rho_2 := \Tr_{\reg{LR}} \qty( \projector{\Adversary^{V,V^\dagger}_k}_{\reg{ABCLR}} )$. \\

\noindent By~\Cref{thm:MH24}, we have $\TD \qty( \rho_1, \rho_2 ) = O\qty( k^2/N^{1/8} )$.

\newhybrid{3} Define the following subnormalized state
\begin{align*}
    \ket{\Adversary^{\textcolor{red}{V^L}, V^\dagger}_k}_{\reg{ABCLR}}
    & := \prod_{j = k}^1 V^\dagger_{\reg{B_j LR}} 
    \cdot \prod_{i = k}^1 \textcolor{red}{V^L_{\reg{A_i LR}}}
    \cdot \ket{\psi}_{\reg{ABC}} 
    \ket{\varnothing}_{\reg{L}} 
    \ket{\varnothing}_{\reg{R}},
\end{align*}
and let $\rho_3 := \Tr_{\reg{LR}} \qty( \projector{\Adversary^{V^L,V^\dagger}_k}_{\reg{ABCLR}} )$.

\noindent Since the term $V^{R,\dagger}$ in $V$ vanishes $\ket{\varnothing}_R$, we have $\ket{\Adversary^{V,V^\dagger}_k}_{\reg{ABCLR}} = \ket{\Adversary^{V^L,V^\dagger}_k}_{\reg{ABCLR}}$, which implies $\rho_2 = \rho_3$.

\newhybrid{4} For $i \in [k]$, define the partial isometry\footnote{To see this, one can verify that it preserves orthogonality by checking that the inner product between any two orthogonal vectors remains zero after applying the transformation.}
\begin{align*}
& \wt{V}^L_{\reg{A_i B LR}}:
\ket{x}_{\reg{A_i}} \ket{\vec{z}}_{\reg{B}} \ket{L}_{\reg{L}} \ket{R}_{\reg{R}}
\mapsto \\
& \quad \sum_{ \substack{ y \in [N]: \\ y \notin \Im(L \cup R) \cup \Supp(\vec{z}) }} 
\frac{1}{\sqrt{N - |\Im(L \cup R) \cup \Supp(\vec{z})|}} \ket{y}_{\reg{A_i}} \ket{\vec{z}}_{\reg{B}} \ket{L \cup \set{(x,y)}}_{\reg{L}} \ket{R}_{\reg{R}},
\end{align*}
and define the following subnormalized state
\begin{align*}
    \ket{\Adversary^{\textcolor{red}{\wt{V}^L}, V^\dagger}_k}_{\reg{ABCLR}}
    & := \prod_{j = k}^1 V^\dagger_{\reg{B_j LR}} 
    \cdot \prod_{i = k}^1 \textcolor{red}{\wt{V}^L_{\reg{A_i BLR}}}
    \cdot \ket{\psi}_{\reg{ABC}} 
    \ket{\varnothing}_{\reg{L}} 
    \ket{\varnothing}_{\reg{R}},
\end{align*}
and let $\rho_4 := \Tr_{\reg{LR}} \qty( \projector{\Adversary^{\wt{V}^L,V^\dagger}_k}_{\reg{ABCLR}} )$.

\myparagraph{Closeness between Hybrids~3 and 4} The only difference between $V^L$ and $\wt{V}^L$ is that $\wt{V}^L$ further disallows $y$ from taking values in the set $\Supp(\vec{z})$. By using a similar argument as in~\cite[Appendix~B]{MH25}, we can obtain
\begin{equation}
\label{eq:Hyb3_Hyb4}
\norm{ \ket{\Adversary^{V^L,V^\dagger}_k}_{\reg{ABCLR}} 
- \ket{\Adversary^{\wt{V}^L,V^\dagger}_k}_{\reg{ABCLR}} }
\leq O\qty( \frac{k^2}{N - k} ), 
\end{equation}
which implies that $\TD( \rho_3, \rho_4 ) \leq O\qty( k^2/(N - k) )$. 
\ifllncs
\else
The proof of~\Cref{eq:Hyb3_Hyb4} is deferred to~\Cref{app:eq:Hyb3_Hyb4}.
\fi
\newhybrid{5} Define the following subnormalized state
\begin{align*}
    \ket{\Adversary^{\wt{V}^L,\textcolor{red}{V^R}}_k}_{\reg{ABCLR}}
    & := \prod_{j = k}^1 \textcolor{red}{V^R_{\reg{B_j LR}}}
    \cdot \prod_{i = k}^1 \wt{V}^L_{\reg{A_i LR}}
    \cdot \ket{\psi}_{\reg{ABC}} 
    \ket{\varnothing}_{\reg{L}} 
    \ket{\varnothing}_{\reg{R}},
\end{align*}
and let $\rho_5 := \Tr_{\reg{LR}} \qty( \projector{\Adversary^{\wt{V}^L,V^R}_k}_{\reg{ABCLR}} )$.

\myparagraph{Equivalence between Hybrids~4 and 5} Consider the following decomposition of $V^{\dagger}_{\reg{B_j}\reg{LR}}$ for each $j \in [k]$:
\begin{equation*}
    V^{\dagger}_{\reg{B_j}\reg{LR}} 
    = V^R_{\reg{B_j}\reg{LR}} + 
    \underbrace{ \qty( \id_{\reg{B_j}\reg{LR}} - V^R_{\reg{B_j}\reg{LR}} \cdot \qty( V^L_{\reg{B_j}\reg{LR}} + V^{R,\dagger}_{\reg{B_j}\reg{LR}} ) ) \cdot V^{L,\dagger}_{\reg{B_j}\reg{LR}} }_{(i)},
\end{equation*}
where the term $(i)$ can be viewed as an error. Thus, we can expand $\prod_{j = k}^1 V^{\dagger}_{\reg{B_j LR}}$ as
\ifllncs
\else
\begin{align*}
    & \prod_{j = k}^1 \qty( V^R_{\reg{B_j LR}} + 
    \qty( \id_{\reg{B_j}\reg{LR}} - V^R_{\reg{B_j}\reg{LR}} \cdot \qty( V^L_{\reg{B_j LR}} + V^{R,\dagger}_{\reg{B_j LR}} ) ) \cdot V^{L,\dagger}_{\reg{B_j LR}} ) \\
    & = \prod_{j = k}^1 V^R_{\reg{B_j LR}}
    + \underbrace{ \sum_{\ell = 1}^k \qty( \prod_{m = k}^{\ell + 1} V^{\dagger}_{\reg{B_m LR}} \cdot \qty( \id_{\reg{B_\ell LR}} - V^R_{\reg{B_\ell LR}} \cdot \qty( V^L_{\reg{B_\ell LR}} + V^{R,\dagger}_{\reg{B_\ell LR}} ) ) \cdot V^{L,\dagger}_{\reg{B_\ell LR}} 
    \cdot \prod_{n = \ell - 1}^{1} V^R_{\reg{B_n LR}} ) }_{(ii)}.
\end{align*}
\fi

\noindent Similarly, the term $(ii)$ can be viewed as an error. To simplify notation, we use the shorthand for the summand in $(ii)$. For each $\ell \in [k]$, define
\ifllncs
\else
\begin{align*}
E_{\ell,\reg{BLR}} := \prod_{m = k}^{\ell + 1} V^{\dagger}_{\reg{B_m LR}} \cdot \qty( \id_{\reg{B_\ell LR}} - V^R_{\reg{B_\ell LR}} \cdot \qty( V^L_{\reg{B_\ell LR}} + V^{R,\dagger}_{\reg{B_\ell LR}} ) ) \cdot V^{L,\dagger}_{\reg{B_\ell LR}} 
\cdot \prod_{n = \ell - 1}^{1} V^R_{\reg{B_n LR}}.
\end{align*}
\fi

\noindent Therefore, we can write $\ket{\Adversary^{\wt{V}^L,V^\dagger}_k}$ in~\hybrid{4} as
\ifllncs
\else
\begin{align*}
    \ket{\Adversary^{\wt{V}^L,V^\dagger}_k} 
    & = \qty( \prod_{j = k}^1 V^R_{\reg{B_j LR}}
    + \sum_{\ell = 1}^t E_{\ell,\reg{BLR}} ) \cdot 
    \prod_{i = k}^1 \wt{V}^L_{\reg{A_i BLR}}
    \cdot \ket{\psi}_{\reg{ABC}} 
    \ket{\varnothing}_{\reg{L}} 
    \ket{\varnothing}_{\reg{R}} \\
    & = \ket{\Adversary^{\wt{V}^L,V^R}_k}_{\reg{ABCLR}} 
    + \sum_{\ell = 1}^t \qty( E_{\ell,\reg{BLR}} \cdot 
    \prod_{i = k}^1 \wt{V}^L_{\reg{A_i BLR}}
    \cdot \ket{\psi}_{\reg{ABC}} 
    \ket{\varnothing}_{\reg{L}} 
    \ket{\varnothing}_{\reg{R}} ).
\end{align*}
\fi
Notice that for all $\ell \in [k]$,
\begin{align*}
    V^{L,\dagger}_{\reg{B_\ell LR}} 
    \cdot \prod_{n = \ell - 1}^{1} V^R_{\reg{B_n LR}}
    \cdot \prod_{i = k}^1 \wt{V}^L_{\reg{A_i BLR}} 
    \cdot \ket{\psi}_{\reg{ABC}} \ket{\varnothing}_{\reg{L}} \ket{\varnothing}_{\reg{R}}
    = 0.
\end{align*}
This is because, by construction, $\prod_{n = \ell - 1}^{1} V^R_{\reg{B_n LR}} 
    \cdot \prod_{i = k}^1 \wt{V}^L_{\reg{A_i BLR}} \cdot \ket{\psi}_{\reg{ABC}} \ket{\varnothing}_{\reg{L}} \ket{\varnothing}_{\reg{R}}$ lies in the subspace spanned by
\begin{equation*}
\set{ \ket{\vec{z}}_{\regB} \otimes \ket{L}_{\reg{L}}:
\Supp(\vec{z}) \cap \Im(L) = \varnothing },
\end{equation*}
which is not in the image of $V^L_{\reg{B_\ell LR}}$. This implies $\ket{\Adversary^{V,V^\dagger}_k}_{\reg{ABCLR}} = \ket{\Adversary^{\wt{V}^L,V^R}_k}_{\reg{ABCLR}}$, and hence $\rho_4 = \rho_5$.

\newhybrid{6} Define the following normalized state
\begin{align*}
    \ket{\Adversary^{\textcolor{red}{V^L},V^R}_k}_{\reg{ABCLR}}
    & := \prod_{j = k}^1 V^R_{\reg{B_j LR}} 
    \cdot \prod_{i = k}^1 \textcolor{red}{V^L_{\reg{A_i LR}}}
    \cdot \ket{\psi}_{\reg{ABC}} 
    \ket{\varnothing}_{\reg{L}} 
    \ket{\varnothing}_{\reg{R}},
\end{align*}
and let $\rho_6 := \Tr_{\reg{LR}} \qty( \projector{\Adversary^{V^L,V^R}_k}_{\reg{ABCLR}} )$.

\noindent Using the same argument as in~\Cref{eq:Hyb3_Hyb4}, $\TD( \rho_5, \rho_6) \leq O\qty( k^2/(N - k) )$.

\newhybrid{7} For each $j \in [k]$, define the operator
\begin{align*}
    \wt{V}^R_{\reg{B_j LR}}:
    \ket{z}_{\reg{B_j}} \ket{L}_{\reg{L}} \ket{R}_{\reg{R}}
    \mapsto \sum_{\substack{w \in [N]: \\ w \notin \Dom(L \cup R)}} 
    \frac{1}{\sqrt{N}} \ket{w}_{\reg{B_j}} \ket{L}_{\reg{L}} \ket{R \cup \set{(w,z)}}_{\reg{R}}.
\end{align*}

\noindent Define the following subnormalized state
\begin{align*}
\ket{ \Adversary^{ V^L, \textcolor{red}{\wt{V}^R} }_k }_{\reg{ABCLR} }
:= \prod_{j = k}^1 \textcolor{red}{\wt{V}^R_{\reg{B_j LR}}}
\cdot \prod_{i = k}^1 V^L_{\reg{A_i LR}} 
\cdot \ket{\psi}_{\reg{ABC}} \ket{\varnothing}_{\reg{L}} \ket{\varnothing}_{\reg{R}},
\end{align*}
and let $\rho_7 := \Tr_{\reg{LR}} \qty( \projector{\Adversary^{V^L,\wt{V}^R}_k}_{\reg{ABCLR}} )$.

\noindent The only difference between $V^R$ and $\wt{V}^R$ is the constant factor. 
\ifllncs
\else
In~\Cref{app:eq:Hyb6_Hyb7}, we show  
\fi
\begin{equation}
\label{eq:Hyb6_Hyb7}
\norm{ \ket{ \Adversary^{V^L,V^R}_k }_{\reg{ABCLR} } 
- \ket{ \Adversary^{V^L,\wt{V}^R}_k }_{\reg{ABCLR} } } 
\leq O\qty(\frac{k^2}{N - k}),
\end{equation}
which implies $\TD( \rho_6, \rho_7) \leq O\qty( k^2/(N-k) )$.

\newhybrid{8} \noindent Define the unnormalized state
\ifllncs
\else
\begin{multline*}
    \ket{P}_{\reg{ABCLR}} := \\
    \frac{1}{N^k \cdot k!} 
    \sum_{ \substack{
    \vec{w},\vec{x},\vec{y},\vec{z} \in [N]^k \\ 
    \pi, \tau \in \sym_k
        }
    } 
    \ketbra{\vec{z}}{\vec{x}}_{\reg{A}} \otimes \ketbra{\vec{w}}{\vec{y}}_{\reg{B}} \cdot \ket{\psi}_{\reg{ABC}} \otimes P_{N}(\pi) \ket{\vec{x}}_{\reg{L_X}} \otimes P_{N}(\pi)\ket{\vec{z}}_{\reg{L_Y}} \otimes P_{N}(\tau) \ket{\vec{w}}_{\reg{R_X}} \otimes P_{N}(\tau)\ket{\vec{y}}_{\reg{R_Y}},
\end{multline*}
\fi
and let $\rho_8 := \Tr_{\reg{LR}} \qty( \projector{P}_{\reg{ABCLR}} )$. \\

\noindent Before proving the closeness between $\rho_7$ and $\rho_8$, we first show that $\ket{P}$ is a purification of $\rho^{\sf approx}$.

\myparagraph{Equivalence between $\rho_8$ and $\rho^{\sf approx}$} By the Ricochet property (\Cref{fact:Ricochet}), we can write $\ket{P}$ as follows:
\ifllncs
\else
\begin{multline*}
    \ket{P}_{\reg{ABCLR}} = \\
    \frac{1}{N^k \cdot k!} \sum_{\substack{ \vec{w},\vec{x},\vec{y},\vec{z} \in [N]^k \\ \pi, \tau \in \sym_k}} 
    P_{N}(\pi)^\intercal_{\regA} \ketbra{\vec{z}}{\vec{x}}_{\reg{A}} P_{N}(\pi)_{\regA} \otimes 
    P_{N}(\tau)^\intercal_{\regB} \ketbra{\vec{w}}{\vec{y}}_{\reg{B}} P_{N}(\tau)_{\regB} 
    \cdot \ket{\psi}_{\reg{ABC}} 
    \ket{\vec{x}}_{\reg{L_X}} 
    \ket{\vec{z}}_{\reg{L_Y}}  
    \ket{\vec{w}}_{\reg{R_X}}
    \ket{\vec{y}}_{\reg{R_Y}}.
\end{multline*}
\fi
A direct calculation yields
\begin{align*}
\rho_8
& = \Tr_{\reg{LR}} \qty[ \projector{P}_{\reg{ABCLR}} ] \\
& = \frac{1}{N^{2k} (k!)^2} \sum_{\substack{\vec{w},\vec{x},\vec{y},\vec{z} \in[N]^k \\ \pi, \pi', \tau, \tau' \in \sym_k}} 
P_{N}(\pi)^\intercal_{\regA} \ketbra{\vec{z}}{\vec{x}}_{\reg{A}} P_{N}(\pi)_{\regA} 
\otimes 
P_{N}(\tau)^\intercal_{\regB} \ketbra{\vec{w}}{\vec{y}}_{\reg{B}} P_{N}(\tau)_{\regB}
\cdot \projector{\psi}_{\reg{ABC}} \\
& \hspace{12em} \cdot P_{N}(\pi')^\intercal_{\regA} \ketbra{\vec{x}}{\vec{z}}_{\reg{A}} P_{N}(\pi')_{\regA} 
\otimes 
P_{N}(\tau')^\intercal_{\regB} \ketbra{\vec{y}}{\vec{w}}_{\reg{B}} P_{N}(\tau')_{\regB} \\
& = \frac{1}{N^{2k} (k!)^2} \sum_{\pi, \pi', \tau, \tau' \in \sym_k} 
P_{N}(\pi)^\intercal_{\regA} P_{N}(\pi')_{\regA} \otimes P_{N}(\tau)^\intercal_{\regB} P_{N}(\tau')_{\regB} \\
& \hspace{12em} \otimes \Tr_{\reg{AB}} \qty( P_{N}(\pi')^\intercal_{\regA} P_{N}(\pi)_{\regA} \otimes P_{N}(\tau')^\intercal_{\regB} P_{N}(\tau)_{\regB} \cdot \projector{\psi}_{\reg{ABC}} ) \\
& = \frac{1}{N^{2k} (k!)^2} \sum_{\pi, \pi', \tau, \tau' \in \sym_k} 
P_{N}(\pi^{-1} \circ \pi')_{\regA} \otimes P_{N}(\tau^{-1} \circ \tau')_{\regB} \\
& \hspace{12em} \otimes \Tr_{\reg{AB}} \qty( P_{N}(\pi^{-1} \circ \pi')^\intercal_{\regA} \otimes P_{N}(\tau^{-1} \circ \tau')^\intercal_{\regB} \cdot \projector{\psi}_{\reg{ABC}} ) \\
& = \frac{1}{N^{2k}} \sum_{\pi, \tau\in \sym_k} 
P_{N}(\pi)_{\regA} \otimes P_{N}(\tau)_{\regB} \cdot \Tr_{\reg{AB}} \qty( P_{N}(\pi)^\intercal_{\regA} \otimes P_{N}(\tau)^\intercal_{\regB} \cdot \projector{\psi}_{\reg{ABC}} ) 
= \rho^{\sf approx}_{\reg{ABC}}.
\end{align*}

\myparagraph{Closeness between Hybrids~7 and~8} We expand $\ket{ \Adversary^{V^L,\wt{V}^R}_k}$ in~\hybrid{7} by using the definitions of $V^L$ and $\wt{V}^R$, and obtain
\ifllncs
\else
\begin{align*}
\ket{ \Adversary^{V^L,\wt{V}^R}_k}_{\reg{ABCLR} } 
& = \frac{1}{\sqrt{N^{\downarrow k} N^k}} 
\sum_{\substack{
    \vec{x}, \vec{y} \in[N]^k, \\ 
    \vec{z} \in [N]^k_\dist, \\
    \vec{w} \in [ N \setminus \Supp(\vec{x}) ]^k_\dist
}} 
    \ketbra{\vec{z}}{\vec{x}}_{\reg{A}} \otimes \ketbra{\vec{w}}{\vec{y}}_{\reg{B}} \cdot \ket{\psi}_{\reg{ABC}} \otimes \ket{\set{(x_i,z_i)}_{i \in [k]}}_{\reg{L}} \otimes \ket{\set{(w_j,y_j)}_{j \in [k]}}_{\reg{R}} \\
& = \frac{1}{\sqrt{N^{\downarrow k} N^k} \cdot k!} 
\sum_{\substack{
    \vec{x}, \vec{y} \in[N]^k, \\ 
    \vec{z} \in [N]^k_\dist, \\
    \vec{w} \in [N \setminus \Supp(\vec{x}) ]^k_\dist \\ 
    \pi, \tau \in \sym_k
}} 
    \ketbra{\vec{z}}{\vec{x}}_{\reg{A}} \otimes \ketbra{\vec{w}}{\vec{y}}_{\reg{B}} \cdot \ket{\psi}_{\reg{ABC}} \\
& \hspace{.25\textwidth} \otimes P_{N}(\pi) \ket{\vec{x}}_{\reg{L_X}} \otimes P_{N}(\pi)\ket{\vec{z}}_{\reg{L_Y}} \otimes P_{N}(\tau) \ket{\vec{w}}_{\reg{R_X}} \otimes P_{N}(\tau)\ket{\vec{y}}_{\reg{R_Y}}.
\end{align*}
\fi

\noindent Define the following projector $\Pi_{\reg{LR}}$ onto the subspace spanned by
\[
\set{ \ket{\vec{x}}_{\reg{L_X}}\ket{\vec{z}}_{\reg{L_Y}}\ket{\vec{w}}_{\reg{R_X}}\ket{\vec{y}}_{\reg{R_Y}}: \vec{x}, \vec{y} \in[N]^k, \vec{z} \in [N]^k_\dist, \vec{w} \in [N \setminus \Supp(\vec{x}) ]^k_\dist },
\] 
which allows us to write
\begin{equation*}
    \ket{ \Adversary^{V^L,\wt{V}^R}_k }_{\reg{ABCLR}}
    = \sqrt{\frac{N^k}{N^{\downarrow k}}} \cdot \Pi_{\reg{LR}} \cdot \ket{P}_{\reg{ABCLR}}.
\end{equation*}

\noindent Therefore,
\begin{align*}
\norm{ \rho_7 - \rho_8 }_1
= & \norm{ \Tr_{\reg{LR}} \qty( \projector{\Adversary^{V^L,\wt{V}^R}_k}_{\reg{ABCLR} } ) - \Tr_{\reg{LR}} \qty( \projector{P}_{\reg{ABCLR} } )  }_1 \\
= & \norm{ \frac{N^k}{N^{\downarrow k}} \Tr_{\reg{LR}} \qty( \Pi_{\reg{LR}} \projector{P}_{\reg{ABCLR}} \Pi_{\reg{LR} } ) - \Tr_{\reg{LR}} \qty( \projector{P}_{\reg{ABCLR}} ) }_1.
\end{align*}
Since $\Pi_{\reg{LR}}$ acts as the identity on $\reg{ABC}$, we can use the cyclicity of $\Tr_{\reg{LR}}$ to obtain
\begin{align}
\label{eq:partial_trace_cyclicity}
\Tr_{\reg{LR}} \qty( \projector{P}_{\reg{ABCLR} } ) 
& = \Tr_{\reg{LR}} \qty( \Pi_{\reg{LR}} \projector{P}_{\reg{ABCLR}} \Pi_{\reg{LR} } ) \nonumber \\
& + \Tr_{\reg{LR}} \qty( (\id - \Pi_{\reg{LR}}) \projector{P}_{\reg{ABCLR}} (\id - \Pi_{\reg{LR}}) ).
\end{align}
Using~\Cref{eq:partial_trace_cyclicity} the triangle inequality, we obtain
\begin{align*}
& \norm{ \rho_7 - \rho_8 }_1 \\
& \leq \qty( \frac{N^k}{N^{\downarrow k}} - 1 ) \cdot \norm{ \Tr_{\reg{LR}} \qty( \Pi_{\reg{LR}} \projector{P}_{\reg{ABCLR}} \Pi_{\reg{LR} } ) }_1 \\
& \hspace{.2\textwidth} + \norm{ \Tr_{\reg{LR}} \qty( (\id - \Pi_{\reg{LR}}) \projector{P}_{\reg{ABCLR}} (\id - \Pi_{\reg{LR}}) ) }_1 \\
& = \qty( \frac{N^k}{N^{\downarrow k}} - 2 ) \cdot \Tr \qty( \Pi_{\reg{LR}} \projector{P}_{\reg{ABCLR}} \Pi_{\reg{LR} } )
+ \Tr \qty( \projector{P}_{\reg{ABCLR}}) \\
& = \qty( \frac{N^k}{N^{\downarrow k}} - 2 ) \cdot \frac{N^{\downarrow k}}{N^k} \cdot \norm{ \ket{ \Adversary^{V^L,\wt{V}^R}_k }_{\reg{ABCLR}} }^2
+ \Tr \qty( \projector{P}_{\reg{ABCLR}}),
\end{align*}
where we use the fact that $\norm{M}_1 = \Tr(M)$ for $M \succeq 0$. By~\Cref{eq:Hyb6_Hyb7} and the fact that $\ket{ \Adversary^{V^L,V^R}_k }_{\reg{ABCLR}}$ in~\hybrid{6} is normalized, we obtain 
\begin{align*}
\norm{ \ket{ \Adversary^{V^L,\wt{V}^R}_k }_{\reg{ABCLR}} } \leq 1 + O\qty( \frac{k^2}{N-k} ).
\end{align*}
Using~\Cref{lem:multiplicative_twirling} twice, we obtain
\begin{align*}
\Tr \qty( \projector{P}_{\reg{ABCLR}}) 
= \Tr(\rho^{\sf approx}_{\reg{ABC}})
\leq \frac{1}{(1 - \veps)^2}
\leq 1 - O(\veps),
\end{align*}
where $\veps := k^2/N$. Hence, after plugging the above bounds and using the fact that $N^k/N^{\downarrow k} = 1 + O(k^2/N)$, we obtain $\norm{ \rho_7 - \rho_8 }_1 \leq O\qty( k^2/(N-k) )$. Collecting all hybrids and noting that the dominant error term arises from~\hybrid{1} and ~\hybrid{2} completes the proof.
\end{proof}
\vspace{.5em}
\noindent Combining~\Cref{lem:multiplicative_twirling} and~\Cref{thm:mixed_twirling}, we have the following corollary, which states that under \emph{non-adaptive} queries, oracle access to a Haar random unitary and its inverse is indistinguishable from forward-query access to two \emph{independent} Haar random unitaries.
\begin{corollary}
\label{cor:UUdagger_UV}
Let $\Phi_0, \Phi_1 \in \mathrm{C}((\C^N)^{\otimes 2k},(\C^N)^{\otimes 2k})$ be channels defined by
\begin{align*}
\Phi_0(\cdot) & := \Ex_{U \sim \haarunitaries(N)} [ U^{\otimes k} \otimes U^{\dagger, \otimes k}(\cdot)(U^{\otimes k} \otimes U^{\dagger, \otimes k})^\dagger ], \\
\Phi_1(\cdot) & := \Ex_{U, V \sim \haarunitaries(N)} [ U^{\otimes k} \otimes V^{\otimes k}(\cdot)(U^{\otimes k} \otimes V^{\otimes k})^\dagger ].
\end{align*}
Then $\|\Phi_0 - \Phi_1\|_\diamond \leq O(k^2/N^{1/8})$.
\end{corollary}

By~\Cref{cor:UUdagger_UV} and a hybrid argument, we have the following theorem.
\begin{theorem}
\label{thm:NAInvHUD}
For any two-party adversary $(\alice,\bob)$ who does not share any entanglement initially, and where $\alice$ and $\bob$ each make $t$ forward queries and $t$ inverse queries,
\begin{align*}
\Pr[\NAInvHUD(N, \alice, \bob) = 1] 
\leq \frac{1}{2} + O\qty( \frac{t^2}{ N^{1/8} } ).
\end{align*}
\end{theorem}
\ifllncs
\else
\begin{proof}
Consider the following sequence of (information-theoretic) hybrids in which we use $(\cO_1,\cO_2:\cO_3,\cO_4)$ to denote $\alice$'s and $\bob$'s oracle access. That is, $\alice$ has oracle access to $(\cO_1,\cO_2)$ and $\bob$ has oracle access to $(\cO_3,\cO_4)$. The changes between hybrids are highlighted in red:
\begin{align*}
        & ( \cU, \cU^\dagger : \cU, \cU^\dagger ) &  \\
\approx & ( \cU, \textcolor{red}{\cX} : \cU, \textcolor{red}{\cX} ) \hspace{5em} & \text{(by \Cref{cor:UUdagger_UV})} \\
\approx & ( \cU, \cX : \textcolor{red}{\cV}, \cX ) \hspace{5em} & \text{(by \Cref{thm:adap_LOCC})} \\
\approx & ( \cU, \cX : \cV, \textcolor{red}{\cY} ) \hspace{5em} & \text{(by \Cref{thm:adap_LOCC})} \\
\approx & ( \cU, \textcolor{red}{\cU^\dagger} : \cV, \cY ) \hspace{5em} & \text{(by \Cref{cor:UUdagger_UV})} \\
\approx & ( \cU, \cU^\dagger : \cV, \textcolor{red}{\cV^\dagger} ) \hspace{5em} & \text{(by \Cref{cor:UUdagger_UV})}
\end{align*}
where we treat $(\alice,\bob)$ as a single entity when applying~\Cref{cor:UUdagger_UV}. Hence, the distinguishing advantage is at most $O\qty( t^2 / N^{1/8} )$.
\end{proof}
\fi

\subsection{Non-Interactive LOCC Indistinguishability against Adaptive and Inverse Queries}

\begin{definition}[Non-Interactive Invertible Haar Unitary Distinguishing Game]
\label{def:NIInvHUDgame}
A \emph{non-interactive invertible Haar unitary distinguishing game}, parametrized by dimension
$N \in \N$, is played by the challenger $\challenger$ and a two-party adversary $(\alice^{(\cdot)}, \bob^{(\cdot)})$.
\begin{protocol}{$\underline{\NIInvHUD(N, \alice, \bob)}$}
{
\begin{enumerate}
    \item $\challenger$ samples two independent Haar unitaries $U,V\sim\haarunitaries(N)$.
    \item $\challenger$ samples a challenge bit $b \gets \bit$. If $b = 0$, $\alice$ and $\bob$ are both given oracle access to $U$ and its inverse $U^\dagger$. Otherwise, $\alice$ is given oracle access to $U$ and its inverse $U^\dagger$, and $\bob$ is given oracle access to $V$ and its inverse $V^\dagger$.
    \item $\alice$ sends {\bf one round} of messages to $\bob$.
    \item After receiving the message, $\bob$ can still make queries, and outputs a guess $b'\in\bit$.
    \item $\challenger$ outputs $1$ if and only if $b' = b$. 
\end{enumerate}
}
\end{protocol}
\end{definition}

\begin{theorem}
\label{thm:oneround_LOCC}
For any two-party adversary $(\alice,\bob)$ in $\NIInvHUD$ who does not share any entanglement initially, and where $\alice$ sends $m$ bits of message and $\bob$ makes $t$ forward queries and $t$ inverse queries,
\begin{align*}
\Pr[\NIInvHUD(N, \alice, \bob) = 1] 
\leq \frac{1}{2} + O\qty( \sqrt{\frac{t^2 (m + \log N)}{N}} ).
\end{align*}
\end{theorem}
Note that the above bound is independent of the number of queries made by $\alice$. 
\ifllncs
\else
The proof of~\Cref{thm:oneround_LOCC}, which follows a similar idea from~\cite{Kretschmer21} and relies on the concentration of the Haar measure, can be found in~\Cref{app:oneround_LOCC}.
\fi
\section{Black-box Separations between PRUs and QCCC primitives}
\label{sec:separation}

In this section, we show (fully) black-box separations between PRUs and various QCCC primitives.
\ifllncs
\else
Our approach follows the general framework in~\cite{AGL24}. We first recall the definition of fully black-box reductions~\cite{RTV04,BBF13} and their quantum generalization. The definitions below are taken from~\cite{HY20,CM24TCC,CCS24}.

\begin{definition}[Quantum primitives]
A quantum primitive $\cP$ is a pair $(\cF_\cP,$ $\cR_\cP)$, where $\cF_\cP$ is a set of quantum channels $\cI$, and $\cR_\cP$ is a relation over pairs $(\cI, \cA)$ of quantum channels $\cI \in \cF_\cP$ and $\cA$. A quantum channel $\cI$ implements $\cP$ or is an implementation of $\cP$ if $\cI \in \cF_\cP$. If $\cI \in \cF_\cP$ is QPT, then $\cI$ is an efficient implementation of $\cP$. A quantum channel $\cA$ $\cP$-breaks $\cI \in \cF_\cP$ if $(\cI, \cA) \in \cR_\cP$. A secure implementation of $\cP$ is an implementation $\cI$ of $\cP$ such that no QPT quantum channel $\cP$-breaks $\cI$. The primitive $\cP$ quantumly exists if there exists an efficient and secure implementation of $\cP$.
\end{definition}

\begin{definition}[Quantum primitives relative to oracle]
Let $\cP = (\cF_\cP, \cR_\cP)$ be a quantum primitive, and $O$ be a quantum oracle. An oracle quantum channel $\cI$ implements $\cP$ relative to $O$ or is an implementation of $\cP$ relative to $O$ if $\cI^O \in \cF_\cP$. If $\cI^O \in \cF_\cP$ is QPT, then $\cI$ is an efficient implementation of $\cP$ relative to $O$. A quantum channel $\cA$ $\cP$-breaks $\cI \in \cF_\cP$ relative to $O$ if $(I^O, \cA^O) \in \cR_\cP$. A secure implementation of $\cP$ is an implementation $\cI$ of $\cP$ relative to $O$ such that no QPT channel $\cP$-breaks $\cI$ relative to $O$. The primitive $\cP$ quantumly exists relative to $O$ if there exists an efficient and secure implementation of $\cP$ relative to $O$.
\end{definition}

\begin{definition}[Quantum fully black-box reductions]
A pair $(\sfC, \sfR)$ of QPT oracle quantum channels is a \emph{quantum fully-black-box reduction} from a quantum primitive $\cP = (\cF_\cP, \cR_\cP)$ to a quantum primitive $\cQ = (\cF_\cQ, \cR_\cQ)$ if the following two conditions are satisfied:
\begin{enumerate}
\item \textbf{(Correctness.)} For every implementation $\cI \in \cF_\cQ$, we have $\sfC^\cI \in \cF_\cP$.
\item \textbf{(Security.)} For every implementation $\cI \in \cF_\cQ$ and every quantum channel $\cA$, if $\cA$ $\cP$-breaks $C^\cI$, then $\sfR^{\cA,\cI}$ $\cQ$-breaks $\cI$.
\end{enumerate}
\end{definition}
\fi

\noindent The following is a quantum analog of a result by Impagliazzo and Rudich~\cite{IR89} (formalized in~\cite{RTV04}), and is formally stated in~\cite{CM24TCC}.

\begin{theorem}[{\cite[Theorem~4.1]{CM24TCC}}, rephrased]
\label{thm:relativization}
Suppose there exists a fully black-box reduction $(\sfC,\sfR)$ from primitive $\cQ$ to primitive $\cP$ such that $\sfC$ has unitary access to the implementation of $\cP$. Then for any unitary oracle $\cO$, if $\cP$ exists relative to $\cO$, then $\cQ$ also exists relative to $\cO$. In particular, if $\sfG^\cO$ is an implementation of $\cP$, then $\sfC^{\cO,\sfG^\cO}$ is an implementation of $\cQ$.
\end{theorem}

\begin{definition}[The oracle $\cU$]
The oracle $\cU = \set{U_\secp}_{\secp \in \N}$ is an infinite sequence of Haar random unitaries, where $U_\secp \sim \haarunitaries(2^\secp)$ for each $\secp \in \N$.
\end{definition}

\begin{lemma}[Existence of PRUs against any polynomial-query attackers~\cite{ABGL24}]
\label{lem:exist_PRU}
Relative to $\cU$, there exists a construction $\sfF^{(\cdot)}$ of PRUs where $\sfF$ makes only \emph{forward} queries to $\cU$, and the distinguishing advantage of any computationally unbounded adversary that makes a polynomial number of \emph{forward} queries to $\sfF^{\cU}$ and $\cU$ is negligible.
\end{lemma}

\subsection{QCCC Key Agreements}
\ifllncs
\else
\begin{definition}[QCCC key agreements relative to an oracle]
A \emph{QCCC key agreement relative to an oracle $\cO$} is a two-party interactive protocol consisting of a pair of uniform QPT oracle algorithms $(\alice,\bob)$, where $\alice$ and $\bob$ each take as input the security parameter $1^\secp$, are allowed to make a polynomial number of queries, communicate classically, and output the keys $k_\alice\in\bit$ and $k_\bob\in\bit$, respectively.
\begin{itemize}
    \item \textbf{Completeness.} There exists a negligible function $\negl$ such that for all $\secp \in \N$,
    \[
    \Pr \qty[
    k_\alice = k_\bob: 
    (k_\alice, k_\bob, \tau) \gets \langle \alice^\cO(1^\secp), \bob^\cO(1^\secp) \rangle
    ] 
    \geq 1 - \negl(\secp),
    \]
    where $\tau$ denotes the transcript.
    \item \textbf{Security.} For any QPT eavesdropper $\eve^\cO$, there exists a negligible function $\negl$ such that for all $\secp \in \N$,
    \[
    \Pr \qty[
    k_\eve = k_\bob:
    \substack{
        (k_\alice, k_\bob, \tau) \gets \langle \alice^\cO(1^\secp), \bob^\cO(1^\secp) \rangle, \\
        k_E \gets \eve^\cO(1^\secp,\tau) }
    ] 
    \leq \frac{1}{2} + \negl(\secp).
    \]
\end{itemize}
\end{definition}
\noindent In the plain model, completeness and security are defined analogously without an oracle.
\fi

\begin{lemma}[Breaking any key agreement relative to $\cU$ with polynomially many queries]
\label{lem:nexist_KA}
Relative to $\cU$, for any construction of QCCC key agreements that makes only \emph{forward} queries to $\cU$, there exists an eavesdropper, not necessarily time-efficient, who makes a polynomial number of \emph{forward} queries and breaks the security.
\end{lemma}
\ifllncs
\else
The proof of~\Cref{lem:nexist_KA} follows the approach in~\cite{AGL24} closely, and can be found in~\Cref{app:lem:nexist_KA}.
\fi

\begin{theorem}
\label{thm:QBB_KA}
There does not exist a quantum fully black-box reduction $(\sfC,\sfR)$ from QCCC key agreements to PRUs such that (1) $\sfC$ has unitary access to the PRU, but without access to the inverse (2) $\sfR$ has unitary access to the adversary and the PRU, but without access to the inverse.
\end{theorem}
\begin{proof}
For the sake of contradiction, suppose such a pair $(\sfC,\sfR)$ exists. Let $\sfF^\cU$ be the PRU construction in~\Cref{lem:exist_PRU}. By~\Cref{thm:relativization}, $\sfC^{\cU,\sfF^\cU}$ is a construction of QCCC key agreements. According to~\Cref{lem:nexist_KA}, there exists an eavesdropper $\sfE^\cU$ who makes a polynomial number of queries and breaks the security of the key agreement construction $\sfC^{\cU,\sfF^\cU}$. Furthermore, by definition, the reduction $\sfR^{\cU,\sfF^\cU,\sfE^\cU}$ can break the security of PRU construction $\sfF^\cU$ by making a polynomial number of queries to $\cU$, $\sfF^\cU$, and $\sfE^\cU$, which together amount to a polynomial number of queries to  $\cU$ and $\sfF^\cU$. However, this contradicts~\Cref{lem:exist_PRU}.
\end{proof}

\subsection{Interactive QCCC Commitments}
\ifllncs
\else
\begin{definition}[Interactive QCCC commitments relative to an oracle]
An \emph{interactive QCCC commitment relative to an oracle $\cO$} is a two-party interactive protocol consisting of a pair of uniform QPT oracle algorithms $(\sen,\rec)$ (where $\sen$ is the sender and $R$ is the receiver). Each of $\sen$ and $\rec$ are allowed to make queries and communicate classically.
\begin{itemize}
    \item \textbf{Commit phase:} In the (possibly interactive) commit phase, $\sen$ takes as input the security parameter $1^\secp$ and an input bit $b\in\bit$, and $\rec$ takes as input the security parameter $1^\secp$. We denote the execution of the commit phase by
    $(\rho^{b,\tau}_{\reg{S}},\sigma^\tau_{\reg{R}},\tau) \gets \commit\langle \sen^\cO(1^\secp,b), \rec^\cO(1^\secp) \rangle$, 
    where $\rho^{b,\tau}_{\reg{S}}$ (\resp $\sigma^\tau_{\reg{R}}$) is the state of $\sen$ (\resp $\rec$) after the commit phase, and $\tau$ denotes the transcript.\footnote{We can express them as a product state due to~\Cref{lem:cond_indep}.}
    \item \textbf{Reveal phase:} In the (possibly interactive) reveal phase, the output is $\mu \in \set{0,1,\bot}$ indicating the receiver's output bit or abort. We denote the execution of the reveal phase by $\mu \gets \reveal\langle \sen^\cO(1^\secp,\rho^{b,\tau}_{\reg{S}}), \rec^\cO(1^\secp, \sigma^\tau_{\reg{R}}) \rangle$.
\end{itemize}
The scheme satisfies the following conditions.
\begin{itemize}
    \item \textbf{Completeness.} There exists a negligible function $\negl$ such that for all $\secp \in \N$,
    \[
    \Pr_{\cO} \left[
    \mu = b:\, 
    \substack{ 
    b \gets \bit, \\
    (\rho^{b,\tau}_{\reg{S}},\sigma^\tau_{\reg{R}},\tau) \gets \commit \inner{ \sen^\cO(1^\secp,b)}{\rec^\cO(1^\secp)}, \\
    \mu \gets \reveal \inner{\sen^\cO(1^\secp,\rho^{b,\tau}_{\reg{S}})}{\rec^\cO(1^\secp, \sigma^\tau_{\reg{R}})}
    }
    \right] 
    \geq 1 - \negl(\secp).
    \]
    \item \textbf{Statistical/Computational hiding.} For any computationally unbound (\resp QPT) malicious receiver $\rec^*$ who makes a polynomial number of queries, there exists a negligible function $\negl$ such that for all $\secp \in \N$,
    \[
    \Pr_\cO \left[
    b' = b:\,
    \substack{ 
    b \gets \bit, \\
    (\rho^{b,\tau}_{\reg{S}},\sigma^\tau_{\reg{R}},\tau) \gets \commit \inner{\sen^\cO(1^\secp,b)}{\rec^{*,\cO}(1^\secp)}, \\
    b' \gets \rec^{*,\cO}(1^\secp, \sigma^\tau_{\reg{R}})
    }
    \right] 
    \leq \frac{1}{2} + \negl(\secp).
    \]
    \item \textbf{Statistical/Computational sum binding.} For any computationally unbounded (\resp QPT) malicious sender $\sen^*$ who makes a polynomial number of queries, there exists a negligible function $\negl$ such that for all $\secp\in\N$,
    \[
    \Pr_\cO \left[
    \mu = ch:\,
    \substack{
    (\rho^\tau_{\reg{S}},\sigma^\tau_{\reg{R}},\tau) \gets \commit \inner{\sen^{*,\cO}(1^\secp)}{\rec^\cO(1^\secp)}, \\
    ch \gets \bit, \\
    \mu \gets \reveal \inner{\sen^{*,\cO}(1^\secp,\,ch,\,\rho^\tau_{\reg{S}})}{\rec^\cO(1^\secp,\,\sigma^\tau_{\reg{R}})}
    }
    \right] 
    \leq \frac{1}{2} + \negl(\secp).
    \]
\end{itemize}
\end{definition}
\noindent In the plain model, completeness, hiding, and binding are defined analogously without an oracle.
\fi

\begin{lemma}[Breaking any commitment relative to $\cU$ with polynomially many queries]
\label{lem:nexist_int_com}
Relative to $\cU$, for any construction of interactive QCCC key commitments that makes only \emph{forward} queries to $\cU$ and satisfy completeness, there exists either (1) a malicious receiver or (2) a malicious sender (or possibly both) who, while not necessarily time-efficient, makes a polynomial number of queries and breaks the security.
\end{lemma}
\ifllncs
\else
Similarly, the proof of~\Cref{lem:nexist_int_com} follows the approach in~\cite{AGL24} closely, and can be found in~\Cref{app:lem:nexist_int_com}.
\fi

\begin{theorem}
\label{thm:QBB_IntCom}
There does not exist a quantum fully black-box reduction $(\sfC,\sfR)$ from interactive QCCC commitments to PRUs such that (1) $\sfC$ has unitary access to the PRU, but without access to the inverse (2) $\sfR$ has unitary access to the adversary and the PRU
\end{theorem}
\begin{proof}
The proof is similar to that of~\Cref{thm:QBB_KA}, which we omit here.
\end{proof}

\section*{Acknowledgments}
PA, AG and YTL are supported by the National Science Foundation under the grants FET-2329938, CAREER-2341004 and, FET-2530160.

\printbibliography

\newpage

\ifllncs
\else
    \appendix
    \section{Omitted Proofs in~\Cref{sec:LOCC}}

\ifllncs
\subsection{Proof of~\Cref{lem:PSDness}}
\label{app:PSDness}

\begin{lemma}[\Cref{lem:PSDness}, restated]
For any PSD operators $M_{\reg{XY}}, \rho_{\reg{YZ}} \succeq 0$, the operator
$\wt{M}_{\reg{XZ}} :=$ $\Tr_{\reg{Y}}($ $M_{\reg{XY}} \otimes \id_{\reg{Z}} \cdot \id_{\reg{X}} \otimes \ptrans_{\reg{Z}}(\rho_{\reg{YZ}}) )$ is PSD.
\end{lemma}
\begin{proof}[Proof of~\Cref{lem:PSDness}]
Since a PSD operator has non-negative eigenvalues and the sum of PSD operators remains PSD, it suffices to consider the case where $M_{\reg{XY}}$ and $\rho_{\reg{YZ}}$ are both rank-$1$ projections. \\

\noindent Suppose $M_{\reg{XY}} = \projector{\psi}_{\reg{XY}}$ and $\rho_{\reg{YZ}} = \projector{\phi}_{\reg{YZ}}$. Consider the Schmidt decomposition of $\ket{\psi}_{\reg{XY}}$:
\begin{align*}
    \ket{\psi}_{\reg{XY}} = \sum_{i=1}^r \lambda_i \ket{u_i}_{\reg{X}} \ket{v_i}_{\reg{Y}}.
\end{align*}
Now, extend the set of orthonormal vectors $\set{\ket{v_i}}_{i\in[r]}$ to an orthonormal basis $\set{\ket{v_i}}_{i\in[d_Y]}$ of $\cH_\reg{Y}$, where $d_Y := \dim(\cH_\reg{Y})$. Next, expand $\ket{\phi}_{\reg{YZ}}$ as
\begin{align*}
    \ket{\phi}_{\reg{YZ}} = \sum_{i\in[d_Y],j\in[d_Z]} \alpha_{ij} \ket{v_i}_{\reg{Y}} \ket{j}_{\reg{Z}},
\end{align*}
where $d_Z := \dim(\cH_\reg{Z})$ and $\set{\ket{j}}_{i\in[d_Z]}$ is the computational basis of $\cH_\reg{Z}$. \\

\noindent A direct calculation yields
\begin{align*}
    \wt{M}_{\reg{XZ}} 
    & = \sum_{i,i'\in[r]} \sum_{j,j'\in[d_Z]} \lambda_i \alpha^*_{ij} \lambda^*_{i'} \alpha_{i'j'} \ketbra{u_i}{u_{i'}}_{\reg{X}} \otimes \ketbra{j}{j'}_{\reg{Z}} \\
    & = \qty( \sum_{i\in[r],j\in[d_Z]} \lambda_i \alpha^*_{ij} \ket{u_i}_{\reg{X}} \ket{j}_{\reg{Z}} ) 
    \qty( \sum_{i' \in [r], j' \in [d_Z]} \lambda^*_{i'} \alpha_{i'j'} \bra{u_{i'}}_{\reg{X}} \bra{j'}_{\reg{Z}} ).
\end{align*}
Since $\wt{M}_{\reg{XZ}}$ can be written as $\projector{\xi}$ for some vector $\ket{\xi}$, it is PSD.
\end{proof}
\fi

\subsection{Proof of~\Cref{eq:Hyb3_Hyb4}}
\label{app:eq:Hyb3_Hyb4}

For convenience, we restate~\Cref{eq:Hyb3_Hyb4} below:
\begin{equation*}
\norm{ \ket{\Adversary^{V^L,V^\dagger}_k}_{\reg{ABCLR}} 
- \ket{\Adversary^{\wt{V}^L,V^\dagger}_k}_{\reg{ABCLR}} }
\leq O\qty( \frac{k^2}{N - k} ).
\end{equation*}
\begin{proof}[Proof of~~\Cref{eq:Hyb3_Hyb4}]
Recall the definitions:
\begin{align*}
V^L_{\reg{A_i LR}} \otimes \id_{\reg{B}}
\ket{x}_{\reg{A_i}} \ket{\vec{z}}_{\reg{B}} \ket{L}_{\reg{L}} \ket{R}_{\reg{R}}
& \coloneqq \sum_{ \substack{ y \in [N]: \\ y \notin \Im(L \cup R) }} 
\frac{1}{\sqrt{N - |\Im(L \cup R)|}} \ket{y}_{\reg{A_i}} \ket{\vec{z}}_{\reg{B}} \ket{L \cup \set{(x,y)}}_{\reg{L}} \ket{R}_{\reg{R}}, \\
\wt{V}^L_{\reg{A_i B LR}}
\ket{x}_{\reg{A_i}} \ket{\vec{z}}_{\reg{B}} \ket{L}_{\reg{L}} \ket{R}_{\reg{R}}
& \coloneqq \sum_{ \substack{ y \in [N]: \\ y \notin \Im(L \cup R) \cup \Supp(\vec{z}) }} 
\frac{1}{\sqrt{N - |\Im(L \cup R) \cup \Supp(\vec{z})|}} \ket{y}_{\reg{A_i}} \ket{\vec{z}}_{\reg{B}} \ket{L \cup \set{(x,y)}}_{\reg{L}} \ket{R}_{\reg{R}}.
\end{align*}
One can verify that for every distinct pair $(x,\vec{z},L,R)$ and $(x',\vec{z'},L',R')$, the following holds:
\begin{itemize}
    \item $V^L_{\reg{A_i LR}} \otimes \id_{\reg{B}} \ket{x}_{\reg{A_i}} \ket{\vec{z}}_{\reg{B}} \ket{L}_{\reg{L}} \ket{R}_{\reg{R}}$ is orthogonal to $V^L_{\reg{A_i LR}} \otimes \id_{\reg{B}} \ket{x}_{\reg{A_i}} \ket{\vec{z'}}_{\reg{B}} \ket{L'}_{\reg{L}} \ket{R'}_{\reg{R}}$,
    \item $\wt{V}^L_{\reg{A_i B LR}} \ket{x'}_{\reg{A_i}} \ket{x}_{\reg{A_i}} \ket{\vec{z}}_{\reg{B}} \ket{L}_{\reg{L}} \ket{R}_{\reg{R}}$ is orthogonal to $\wt{V}^L_{\reg{A_i B LR}} \ket{x'}_{\reg{A_i}} \ket{\vec{z'}}_{\reg{B}} \ket{L'}_{\reg{L}} \ket{R'}_{\reg{R}}$,
    \item $V^L_{\reg{A_i LR}} \otimes \id_{\reg{B}} \ket{x}_{\reg{A_i}} \ket{\vec{z}}_{\reg{B}} \ket{L}_{\reg{L}} \ket{R}_{\reg{R}}$ is orthogonal to $\wt{V}^L_{\reg{A_i B LR}} \ket{x'}_{\reg{A_i}} \ket{\vec{z'}}_{\reg{B}} \ket{L'}_{\reg{L}} \ket{R'}_{\reg{R}}$.
\end{itemize}
Hence, for any normalized state 
\[
\ket{\psi}_{\reg{A_i B LR}} = \sum_{x,\vec{z},L,R} \alpha_{x\vec{z}LR} 
\ket{x}_{\reg{A_i}} \ket{\vec{z}}_{\reg{B}} \ket{L}_{\reg{L}} \ket{R}_{\reg{R}},
\]
it holds that
\begin{align*}
& \norm{ (V^L_{\reg{A_i LR}} \otimes \id_{\reg{B}} - \wt{V}^L_{\reg{A_i B LR}}) \ket{\psi}_{\reg{A_i B LR}} }^2 \\
& = \sum_{x,\vec{z},L,R} |\alpha_{x\vec{z}LR}|^2 \cdot 
\norm{ (V^L_{\reg{A_i LR}} \otimes \id_{\reg{B}} - \wt{V}^L_{\reg{A_i B LR}}) \ket{x}_{\reg{A_i}} \ket{\vec{z}}_{\reg{B}} \ket{L}_{\reg{L}} \ket{R}_{\reg{R}} }^2 \\
& \leq \max_{x,\vec{z},L,R: \alpha_{x\vec{z}LR} \neq 0} 
\norm{ (V^L_{\reg{A_i LR}} \otimes \id_{\reg{B}} - \wt{V}^L_{\reg{A_i B LR}}) \ket{x}_{\reg{A_i}} \ket{\vec{z}}_{\reg{B}} \ket{L}_{\reg{L}} \ket{R}_{\reg{R}} }^2,
\end{align*}
where the equality follows from the Pythagorean theorem. \\

\noindent For any $(x,\vec{z},L,R)$, we have
\begin{align*}
& (V^L_{\reg{A_i LR}} \otimes \id_{\reg{B}} - \wt{V}^L_{\reg{A_i B LR}}) \ket{x}_{\reg{A_i}} \ket{\vec{z}}_{\reg{B}} \ket{L}_{\reg{L}} \ket{R}_{\reg{R}} = \\
& \sum_{ \substack{ y \in [N]: \\ y \notin \Im(L \cup R) \cup \Supp(\vec{z}) }} 
\qty( \frac{1}{\sqrt{N - |\Im(L \cup R) \cup \Supp(\vec{z})|}} - \frac{1}{\sqrt{N - |\Im(L \cup R)|}} ) \cdot \ket{y}_{\reg{A_i}} \ket{\vec{z}}_{\reg{B}} \ket{L \cup \set{(x,y)}}_{\reg{L}} \ket{R}_{\reg{R}} \\
& + \sum_{ \substack{ y \in [N]: \\ y \notin \Im(L \cup R) \land y \in \Supp(\vec{z}) }} 
\frac{1}{\sqrt{N - |\Im(L \cup R) \cup \Supp(\vec{z})|}} 
\ket{y}_{\reg{A_i}} \ket{\vec{z}}_{\reg{B}} \ket{L \cup \set{(x,y)}}_{\reg{L}} \ket{R}_{\reg{R}}.
\end{align*}
Hence, the squared norm of $(V^L_{\reg{A_i LR}} \otimes \id_{\reg{B}} - \wt{V}^L_{\reg{A_i B LR}}) \ket{x}_{\reg{A_i}} \ket{\vec{z}}_{\reg{B}} \ket{L}_{\reg{L}} \ket{R}_{\reg{R}}$ is 
\begin{align*}
& \sum_{ \substack{ y \in [N]: \\ y \notin \Im(L \cup R) \cup \Supp(\vec{z}) }} \qty( \frac{1}{\sqrt{N - |\Im(L \cup R) \cup \Supp(\vec{z})|}} - \frac{1}{\sqrt{N - |\Im(L \cup R)|}} )^2 \\
& + \sum_{ \substack{ y \in [N]: \\ y \notin \Im(L \cup R) \land y \in \Supp(\vec{z}) }} 
\frac{1}{N - |\Im(L \cup R) \cup \Supp(\vec{z})|} \\
& \leq \sum_{ \substack{ y \in [N]: \\ y \notin \Im(L \cup R) \cup \Supp(\vec{z}) } } \frac{|\Im(L \cup R) \cup \Supp(\vec{z})| - |\Im(L \cup R)|}{(N - |\Im(L \cup R) \cup \Supp(\vec{z})|)(N - |\Im(L \cup R)|)} \\
& + \sum_{ \substack{ y \in [N]: \\ y \notin \Im(L \cup R) \land y \in \Supp(\vec{z}) } }
\frac{1}{N - |\Im(L \cup R) \cup \Supp(\vec{z})|} \\
& \leq \frac{|\Supp(\vec{z})|}{N - |\Im(L \cup R)|} + \frac{|\Supp(\vec{z})|}{N - |\Im(L \cup R) \cup \Supp(\vec{z})|} \\
& \leq \frac{2|\Supp(\vec{z})|}{N - |\Im(L \cup R) \cup \Supp(\vec{z})|},
\end{align*}
where the first inequality follows from $(1/\sqrt{a} - 1/\sqrt{b})^2 \leq (b - a)/(ab)$ for $0 < a \leq b$. Hence, we obtain
\begin{align*}
& \norm{ (V^L_{\reg{A_i LR}} \otimes \id_{\reg{B}} - \wt{V}^L_{\reg{A_i B LR}}) \ket{\psi}_{\reg{A_i B LR}} } \\
& \leq \frac{|\Supp(\vec{z})|}{N - |\Im(L \cup R)|} + \frac{|\Supp(\vec{z})|}{N - |\Im(L \cup R) \cup \Supp(\vec{z})|}.
\end{align*}

\noindent Finally, using that fact that $\ket{\Adversary^{V^L,V^\dagger}_k}_{\reg{ABCLR}}$ and $\ket{\Adversary^{\wt{V}^L,V^\dagger}_k}_{\reg{ABCLR}}$ are supported by $(L,R,\vec{z})$ such that $|L| = k$, $|R| = 0$, and $|\vec{z}| \leq k$ and a standard hybrid over oracle calls, we obtain 
\begin{align*}
\norm{ \ket{\Adversary^{V^L,V^\dagger}_k}_{\reg{ABCLR}} 
- \ket{\Adversary^{\wt{V}^L,V^\dagger}_k}_{\reg{ABCLR}} }
\leq \sum_{i = 1}^k O\qty( \frac{k}{N - k} )
= O\qty( \frac{k^2}{N - k} ),
\end{align*}
which completes the proof.
\end{proof}

\subsection{Proof of~\Cref{eq:Hyb6_Hyb7}}
\label{app:eq:Hyb6_Hyb7}
For convenience, we restate~\Cref{eq:Hyb6_Hyb7} below:
\begin{equation*}
\norm{ \ket{ \Adversary^{V^L,V^R}_k }_{\reg{ABCLR} } 
- \ket{ \Adversary^{V^L,\wt{V}^R}_k }_{\reg{ABCLR} } }
\leq O\qty(\frac{k^2}{N - k}).
\end{equation*}
\begin{proof}[Proof of~\Cref{eq:Hyb6_Hyb7}]
Recall the definitions:
\begin{align*}
& V^R_{\reg{B_j LR}} \ket{z}_{\reg{B_j}} \ket{L}_{\reg{L}} \ket{R}_{\reg{R}}
\coloneqq \sum_{\substack{w \in [N]: \\ w \notin \Dom(L \cup R)}} 
\frac{1}{\sqrt{ N - |\Dom(L \cup R)}| } \ket{w}_{\reg{B_j}} \ket{L}_{\reg{L}} \ket{R \cup \set{(w,z)}}_{\reg{R}}, \\
& \wt{V}^R_{\reg{B_j LR}} \ket{z}_{\reg{B_j}} \ket{L}_{\reg{L}} \ket{R}_{\reg{R}}
\coloneqq \sum_{\substack{w \in [N]: \\ w \notin \Dom(L \cup R)}} 
\frac{1}{\sqrt{N}} \ket{w}_{\reg{B_j}} \ket{L}_{\reg{L}} \ket{R \cup \set{(w,z)}}_{\reg{R}}. 
\end{align*}
Similar to the proof of~\Cref{eq:Hyb3_Hyb4} in~\Cref{app:eq:Hyb3_Hyb4}, one can verify that orthogonality holds for distinct pair $(z,L,R)$ and $(z',L',R')$. Using the same argument implies
\begin{equation*}
\norm{ \ket{ \Adversary^{V^L,V^R}_k }_{\reg{ABCLR} } 
- \ket{ \Adversary^{V^L,\wt{V}^R}_k }_{\reg{ABCLR} } }
\leq O\qty(\frac{k^2}{N - k})
\end{equation*}
as desired.
\end{proof}

\ifllncs
\subsection{Proof of~\Cref{thm:NAInvHUD}}
\label{app:thm:NAInvHUD}

\begin{theorem}[\Cref{thm:NAInvHUD}, restated]
For any two-party adversary $(\alice,\bob)$ who does not share any entanglement initially, and where $\alice$ and $\bob$ each make $t$ forward queries and $t$ inverse queries,
\begin{align*}
\Pr[\NAInvHUD(N, \alice, \bob) = 1] 
\leq \frac{1}{2} + O\qty( \frac{t^2}{ N^{1/8} } ).
\end{align*}
\end{theorem}
\begin{proof}[Proof of~\Cref{thm:NAInvHUD}]
Consider the following sequence of (information-theoretic) hybrids in which we use $(\cO_1,\cO_2:\cO_3,\cO_4)$ to denote $\alice$'s and $\bob$'s oracle access. That is, $\alice$ has oracle access to $(\cO_1,\cO_2)$ and $\bob$ has oracle access to $(\cO_3,\cO_4)$. The changes between hybrids are highlighted in red:
\begin{align*}
        & ( \cU, \cU^\dagger : \cU, \cU^\dagger ) &  \\
\approx & ( \cU, \textcolor{red}{\cX} : \cU, \textcolor{red}{\cX} ) \hspace{5em} & \text{(by \Cref{cor:UUdagger_UV})} \\
\approx & ( \cU, \cX : \textcolor{red}{\cV}, \cX ) \hspace{5em} & \text{(by \Cref{thm:adap_LOCC})} \\
\approx & ( \cU, \cX : \cV, \textcolor{red}{\cY} ) \hspace{5em} & \text{(by \Cref{thm:adap_LOCC})} \\
\approx & ( \cU, \textcolor{red}{\cU^\dagger} : \cV, \cY ) \hspace{5em} & \text{(by \Cref{cor:UUdagger_UV})} \\
\approx & ( \cU, \cU^\dagger : \cV, \textcolor{red}{\cV^\dagger} ) \hspace{5em} & \text{(by \Cref{cor:UUdagger_UV})}
\end{align*}
where we treat $(\alice,\bob)$ as a single entity when applying~\Cref{cor:UUdagger_UV}. Hence, the distinguishing advantage is at most $O\qty( t^2 / N^{1/8} )$.
\end{proof}
\fi

\subsection{Proof of~\Cref{thm:oneround_LOCC}}
\label{app:oneround_LOCC}

\begin{theorem}[\Cref{thm:oneround_LOCC}, restated]
For any two-party adversary $(\alice,\bob)$ in $\NIInvHUD$ who does not share any entanglement initially, and where $\alice$ sends $m$ bits of message and $\bob$ makes $t$ forward queries and $t$ inverse queries,
\begin{align*}
\Pr[\NIInvHUD(N, \alice, \bob) = 1] 
\leq \frac{1}{2} + O\qty( \sqrt{\frac{t^2 (m + \log N)}{N}} ).
\end{align*}
\end{theorem}

\begin{theorem}[Special case of~{\cite[Theorem~5.17]{Meckes19}}]
\label{thm:Haar_concentration}
Suppose that $f : \Unitary(N) \to \mathbb{R}$ is $L$-Lipschitz with respect to the Frobenius norm. Then for every $\veps > 0$,
\[
\Pr_{U \sim \haarunitaries(N)} \left[ f(U) \geq \Ex_{V \sim \haarunitaries(N)}[f(V)] + \veps \right] \leq \exp \left( - \frac{(N - 2) \veps^2}{24L^2} \right).
\]
\end{theorem}

\begin{lemma}[{\cite[Lemma~9]{Kretschmer21}}]
\label{lem:diamond_to_Frobenius}
Let $U, V \in \Unitary(N)$. Then $\| U(\cdot)U^\dagger - V(\cdot)V^\dagger \|_\diamond \leq 2 \| U - V \|_F$.
\end{lemma}

\begin{lemma}[Adapted from~{\cite[Lemma~28]{Kretschmer21}}]
\label{lem:Lipschitz}
Let $U \in \Unitary(N)$ be a (fixed) unitary oracle and $\alice^{U,U^\dagger}$ be a quantum algorithm that makes $t$ queries to each $U$ and $U^\dagger$. Define $f: \Unitary(N) \to \R$ by $f(U) := \Pr[ 1 \gets \alice^{U,U^\dagger} ]$. Then $f$ is $4t$-Lipschitz in the Frobenius norm.
\end{lemma}
\begin{proof}
Suppose that $U, V \in U(N)$ and $\norm{ U - V }_F = \veps$. Since the Frobenious norm is invariant under the adjoint operator, we obtain $\norm{ U^\dagger - V^\dagger }_F = \veps$. By~\Cref{lem:diamond_to_Frobenius}, we obtain $\norm{ U - V }_\diamond = \norm{ U^\dagger - V^\dagger }_\diamond \leq 2\veps$. Using sub-additivity of the diamond norm under composition implies that, as quantum channels, $\norm{ \alice^{U,U^\dagger} - \alice^{V,V^\dagger} }_\diamond \leq 4 t \veps$. From~\Cref{fact:diamondnorm}, the distinguishing advantage $| f(U) - f(V) |/2$ is at most $2 t \veps$, proving the lemma.
\end{proof}

\vspace{.5em} Now, we are ready to prove~\Cref{thm:oneround_LOCC}.

\begin{proof}[Proof of~\Cref{thm:oneround_LOCC}]
Since $\bob$ never sends messages to $\alice$, we can assume that $\bob$ starts making queries after receiving the message $\tau \in \bit^m$ from $\alice$. For any fixed $\tau \in \bit^m$ and $U \in \Unitary(N)$, we can view $\bob(\tau)^{U,U^\dagger}$ as a (non-interactive) oracle algorithm and define $f_\tau: \Unitary(N) \to [0,1]$ as $f_\tau(U) := \Pr[ 1 \gets \bob(\tau)^{U,U^\dagger} ]$. By~\Cref{lem:Lipschitz}, we have $f_\tau(U)$ is $2t$-Lipschitz in the Frobenius norm. Next, let $\veps > 0$ be a constant to be defined later. For each $\tau \in \bit^m$, we define the set of ``bad'' unitaries as follows:
\begin{align*}
    \bad_\tau := \bigg\{ U \in \Unitary(N): f_\tau(U) \geq \Ex_{V \sim \haarunitaries(N)}[ f_\tau(V) ] + \veps \bigg\}
    \subseteq \Unitary(N),
\end{align*}
and let $\bad := \bigcup_{\tau \in \bit^m} \bad_\tau \subseteq \Unitary(N)$. \\

\noindent Conditioning on the challenge bit $b$ being $1$ (in which $\alice$ and $\bob$ are interacting with two independent Haar unitaries and their inverse), the probability of $\bob$ outputting $b' = 1$ is given by
\begin{align}
\label{eq:p11}
    p_{1 \mid 1} := \sum_{\tau \in \bit^m} \Pr[ \tau ] \cdot \Ex_{V \sim \haarunitaries(N)}[ f_\tau(V) ].
\end{align}
On the other hand, when $b = 0$, the probability is given by
\begin{align*}
\label{eq:p10}
p_{1 \mid 0}
:= \Ex_{U \sim \haarunitaries(N)} \qty[ \sum_{\tau \in \bit^m} \Pr[\tau \mid U ] \cdot f_\tau(U) ] 
\end{align*}
Using elementary properties of expectation, we have
\begin{align*}
& p_{1 \mid 0} 
= \Pr_{U \sim \haarunitaries(N)}[U \notin \bad] \cdot \Ex_{U \sim \haarunitaries(N)} \qty[ \sum_{\tau \in \bit^m} \Pr[\tau \mid U ] \cdot f_\tau(U) \bigg\vert U \notin \bad ] \\
& \hspace{.1\textwidth} + \Pr_{U \sim \haarunitaries(N)}[U \in \bad] \cdot \Ex_{U \sim \haarunitaries(N)}\qty[ \sum_{\tau \in \bit^m} \Pr[\tau \mid U ] \cdot f_\tau(U) \bigg\vert U \in \bad ] \\
& \leq \Pr_{U \sim \haarunitaries(N)}[U \notin \bad] \cdot \Ex_{U \sim \haarunitaries(N)} \qty[ \sum_{\tau \in \bit^m} \Pr[\tau \mid U ] \cdot \qty( \Ex_{V \sim \haarunitaries(N)}[ f_\tau(V) ] + \veps ) \bigg\vert U \notin \bad ] \\
& \hspace{.1\textwidth} + \Pr[U \in \bad] \cdot \Ex_{U \sim \haarunitaries(N)}\qty[ \sum_{\tau \in \bit^m} \Pr[\tau \mid U ] \bigg\vert U \in \bad ]\\
& \leq \sum_{\tau \in \bit^m} \qty( \Ex_{V \sim \haarunitaries(N)}[ f_\tau(V) ] + \veps ) \cdot \Ex_{U \sim \haarunitaries(N)} \qty[ \Pr[\tau \mid U ] ]
+ \Pr_{U \sim \haarunitaries(N)}[U \in \bad] \\
& = \sum_{\tau \in \bit^m} \qty( \Ex_{V \sim \haarunitaries(N)}[ f_\tau(V) ] \cdot \Pr[\tau] ) + \veps + \Pr_{U \sim \haarunitaries(N)}[U \in \bad] \\
& \leq p_{1|1} + \veps + 2^m e^{- \frac{(N - 2) \veps^2}{96t^2}},
\end{align*}
where the first inequality follows from the definition of $\bad$ and the fact that $f_\tau(U) \leq 1$, and the last equality follows from~\Cref{thm:Haar_concentration} and the union bound. Now, we use $K := (N-2)/(96t^2)$ as a shorthand and let $\veps := \sqrt{(m + \log K)/K}$, we have 
\[
\veps + 2^m e^{- \frac{(N - 2) \veps^2}{96t^2}} 
= \veps + 2^m e^{-K\veps^2}
\leq O\qty( \sqrt{\frac{m + \log K}{K}} )
= O\qty( \sqrt{\frac{t^2 (m + \log N)}{N}} ).
\]
We define $p_{0|0}$ and $p_{0|1}$ analogously. Applying the same argument, we obtain
\[
p_{0|0} \leq p_{0|1} + \veps + 2^m e^{- \frac{(N - 2) \veps^2}{96t^2}}.
\]
Hence, the distinguishing advantage is at most $O\qty( \sqrt{\frac{t^2 (m + \log N)}{N}} )$.
\end{proof}

    \section{Omitted Proofs in~\Cref{sec:separation}}
We first introduce some technical tools for this section.

\begin{theorem}[Quantum process tomography, rephrased~{\cite[Theorem~1.1 and Proposition~2.4]{HKOT23}}]
\label{thm:tomography}
There is a quantum algorithm that, given any $\eps, \delta \in (0,1]$ and black-box access to an unknown $d$-dimensional unitary $Z \in \Unitary(d)$, makes $O((d^2/\eps) \cdot \log(1/\delta))$ queries, and outputs a classical description of a unitary $\wt{Z} \in \Unitary(d)$ such that $\Ex \qty[ \norm{ Z(\cdot)Z^\dagger - \wt{Z}(\cdot)\wt{Z}^\dagger }^2_\diamond ] \leq \veps^2$ with probability at least $1 - \delta/2$. In particular, it implies $\Ex \qty[ \norm{ Z(\cdot)Z^\dagger - \wt{Z}(\cdot)\wt{Z}^\dagger }_\diamond ]$ $\leq \veps + \delta$.\footnote{To see this, we use the fact that $\norm{ U(\cdot)U^\dagger - V(\cdot)V^\dagger }_\diamond \leq 2$ for any unitaries $U$ and $V$, and Jensen's inequality to obtain: $\Ex \qty[ \norm{ Z(\cdot)Z^\dagger - \wt{Z}(\cdot)\wt{Z}^\dagger }_\diamond ]$ $\leq$ $(1 - \delta/2) \cdot \sqrt{\veps^2}$ $+ \delta/2 \cdot 2 \leq \veps + \delta$.}
\end{theorem}

\begin{lemma}[Conditional independence] 
\label{lem:cond_indep}
For any two-party interactive QCCC protocol $(\alice,\bob)$ in the plain model, the final joint state can be written of the form
\[
\sum_{\tau} p_{\tau} \projector{\tau}_{\reg{T}} \otimes \rho_{\tau,\regA} \otimes \sigma_{\tau,\regB},
\]
where register $\reg{T}$ records the transcript $\tau$, and $\set{p_\tau}_\tau$ is a distribution. In particular, conditioning on any transcript in the support, the joint state of $(\alice,\bob)$ is a product state.
\end{lemma}

\Cref{lem:cond_indep} can be proved by induction on the number of rounds.

\begin{lemma}[{\cite[Lemma~9.10]{AGL24}}] 
\label{lem:SD_lemma}
Let $\bfP_{BT},\bfQ_{BT}$ be two distributions over $\bit \times \cT$ for some finite set $\cT$. Consider the following experiments:
\begin{protocolbox}
\begin{minipage}[t]{0.48\textwidth}  
  \vspace{0pt}  
  \noindent $\mathbf{Exp.0:}$
\begin{enumerate}
    \item Sample $(b,\tau) \gets \bfP_{BT}$.
    \item Set $b'$ to the more likely bit according to $\bfP_{B \mid T = \tau}$.
    \item Output $(b,b',\tau)$.
\end{enumerate}
\end{minipage}
\hfill
\begin{minipage}[t]{0.48\textwidth}  
\vspace{0pt}  
\noindent $\mathbf{Exp.1:}$
\begin{enumerate}
    \item Sample $(b,\tau) \gets \bfP_{BT}$.
    \item If $\bfQ_{T}(\tau) = 0$ (where $\bfQ_{T}$ denotes the marginal distribution of $\bfQ_{BT}$ on $\cT$), then sample $b'$ uniformly at random. Otherwise, set $b'$ to the more likely bit according to $\bfQ_{B \mid T = \tau}$.
    \item Output $(b,b',\tau)$.
\end{enumerate}
\end{minipage}
\end{protocolbox}
Then it holds that
\[
\Pr_{\mathbf{Exp.1}}[b = b'] 
\geq \Pr_{\mathbf{Exp.0}}[b = b'] - 3 \cdot \SD(\bfP_{BT},\bfQ_{BT}).
\]
\end{lemma}


\subsection{Proof of~\Cref{lem:nexist_KA}}
\label{app:lem:nexist_KA}

\noindent We focus on the following restricted construction of QCCC key agreements relative to $\cU$, where all queries made by Alice and Bob are of length $O(\log \secp)$. 

\begin{definition}[Short QCCC key agreement constructions]
A QCCC key agreement construction relative to $\cU$ is \emph{short} if the maximum length of queries made by Alice and Bob is $O(\log\secp)$.
\end{definition}

\noindent For a short QCCC key agreement construction, a (statistical) eavesdropper can perform process tomography to learn the description of all ``short'' unitaries by making only a polynomial number of queries. Hence, intuitively, Alice and Bob cannot share any secret by using the oracle. To formalize this intuition, most of the effort in this section is dedicated to analyzing the errors introduced by performing process tomography. \\

\begin{lemma}[Breaking any short QCCC key agreement construction] 
\label{lem:break_shortKA}
For any $\eta:\N \to [0,1]$ and any short QCCC key agreement such that Alice and Bob each output the same key with probability at least $1 - \eta$, there exists an eavesdropper (not necessarily time-efficient) who makes a polynomial number of queries and finds Bob's key with probability at least $0.99 - \eta$. This holds even if Alice and Bob are not time-efficient, as long as they make a polynomial number of queries.
\end{lemma}
\begin{proof}
Let $G^{\cU} = (\alice^\cU,\bob^\cU)$ be a short key agreement construction that satisfies the premises in which Alice and Bob each make $q = \poly(\secp)$ queries of length at most  $\ell(\secp) = O(\log \secp)$. Define the following eavesdropper:
\protocol{$E^{\cU}$}
{
on input $1^\secp$ and transcript $\tau \gets \inner{\sfG.\alice^\cU(1^\secp)}{\sfG.\bob^\cU(1^\secp)}$,
  \begin{enumerate}
    \item Perform the process tomography algorithm in~\Cref{thm:tomography} with $\delta = \eps = 1/(1200q\ell)$ on each $\set{U_\kappa}_{\kappa\in[\ell]}$, and obtain the classical description of the unitaries $\wt{U} := \set{ \wt{U}_\kappa }_{\kappa\in[\ell]}$.
    \item Let $k'_\bob$ be the more likely Bob's key according to $\sfG^{\cU}(1^\secp)$ conditioned on the oracle being $\wt{U}$ and the transcript being $\tau$, and output $k'_\bob$ (if $\wt{U}$ and $\tau$ are not consistent, then output a random bit).
  \end{enumerate}
}
From~\Cref{thm:tomography}, $\eve$ makes at most
\[
\sum_{\kappa = 1}^\ell O\qty( \frac{2^{2\ell}}{\veps} \log\qty( \frac{1}{\delta} ) ) 
= \poly(\secp)
\]
number of queries. Fix $U$ and $\wt{U}$. Let $\bfD_U$ and $\bfD_{\wt{U}}$ denote the marginal distributions of $(k_\bob,\tau)$ generated from $G^\cU(1^\secp)$ conditioned on the oracle being $U$ and $\wt{U}$, respectively. From~\Cref{lem:SD_lemma} (setting $\bfP_{BF} = \bfD_U$ and $\bfQ_{BF} = \bfD_{\wt{U}}$), we obtain 
\begin{align}
\label{eq:shortKA_1}
\Pr_{\mathbf{Exp.1}}[ b = b' ] 
\geq \Pr_{\mathbf{Exp.0}}[ b = b' ] 
- 3 \cdot \SD(\bfD_U,\bfD_{\wt{U}}).
\end{align}
First, by~\Cref{fact:diamondnorm} and viewing $\sfG$ as a $2q$-query algorithm, we have
\begin{align}
\label{eq:shortKA_SD}
\SD(\bfD_U,\bfD_{\wt{U}})
\leq 2q \cdot \sum_{\kappa = 1}^\ell \norm{ U_\kappa(\cdot)U_\kappa^\dagger - \wt{U}_\kappa(\cdot)\wt{U}_\kappa^\dagger }_\diamond.
\end{align}

\noindent Notice that $\Pr_{\mathbf{Exp.1}}[ b = b' ]$ is exactly the guessing probability of $\eve$ conditioned on $U$ and $\wt{U}$, \ie
\begin{align}
\label{eq:shortKA_eve}
    \Pr[ k'_\bob = k_\bob \mid U,\,\wt{U} ] = \Pr_{\mathbf{Exp.1}}[ b = b' ]
\end{align}
Next, we relate the probability $\Pr_{\mathbf{Exp.0}}[ b = b' ]$ to the completeness of $\sfG$. Consider the probability
\begin{align*}
\Pr\left[ k_\alice = k_\bob \mid U,\,\wt{U} \right] 
= \Pr\left[ k_\alice = k_\bob: \substack{ 
    (k_\bob, \tau) \gets \bfP_U, \\
    k_\alice \gets \bfk_\alice|_{U,k_\bob,\tau}
}   \right],
\end{align*}
where $\bfk_\alice|_{U,k_\bob,\tau}$ denotes the conditional distribution of $k_\alice$ conditioned on $(U,k_\bob,\tau)$. By~\Cref{lem:cond_indep}, the distribution $\bfk_\alice|_{U,k_\bob,\tau}$ is independent of $k_\bob$, so it can be denoted as $\bfk_\alice|_{U,\tau}$. Hence, we obtain
\begin{align}
\label{eq:shortKA_complete}
\Pr\left[ k_\alice = k_\bob \mid U,\,\wt{U} \right] 
= \Pr\left[ k_\alice = k_\bob: \substack{ 
    (k_\bob, \tau) \gets \bfP_U, \\
    k_\alice \gets \bfk_\alice|_{U,\tau}
}
\right]
\leq \Pr_{\mathbf{Exp.0}}[ b = b' ],
\end{align}
where the inequality holds because the ``optimal strategy'' for letting $k_\alice = k_\bob$ is to set $k_\alice$ as the most likely value to $k_\bob$ conditioned on $\tau$, which corresponds the step~3 in~$\mathbf{Exp.0}$. Combining~\Cref{eq:shortKA_1,eq:shortKA_SD,eq:shortKA_eve,eq:shortKA_complete}, we obtain
\begin{align}
\label{eq:shortKA_2}
\Pr[ k'_\bob = k_\bob \mid U,\,\wt{U} ] 
\geq \Pr\left[ k_\alice = k_\bob \mid U,\,\wt{U} \right] - 6q \cdot \sum_{\kappa = 1}^\ell \norm{ U_\kappa(\cdot)U_\kappa^\dagger - \wt{U}_\kappa(\cdot)\wt{U}_\kappa^\dagger }_\diamond.
\end{align}
Averaging over $\wt{U}$ in~\Cref{eq:shortKA_2}, and applying~\Cref{thm:tomography}, we have
\begin{align*}
\Pr[ k'_\bob = k_\bob \mid U ] 
\geq \Pr\left[ k_\alice = k_\bob \mid U \right] - 6q\ell \cdot (\veps + \delta)
= \Pr\left[ k_\alice = k_\bob \mid U \right] - 0.01,
\end{align*}
Finally, averaging over $U$, we obtain 
\begin{align*}
\Pr[ k'_\bob = k_\bob ] 
\geq \Pr\left[ k_\alice = k_\bob \right] - 0.01
\geq 0.99 - \eta,
\end{align*}
which completes the proof.
\end{proof}

\begin{lemma} 
\label{lem:reduce_to_shortKA}
Suppose there exists a QCCC key agreement relative to $\cU$ such that (1) Alice and Bob output the same key with probability at least $1 - \negl(\secp)$, and (2) any computationally unbounded polynomial-query eavesdropper can find Bob's key with probability at most $1/2 + \negl(\secp)$. Then there exists a short QCCC key agreement relative to $\cU$ such that: (1) Alice and Bob, who are not necessarily time-efficient, each make a polynomial number of queries and output the same key with probability at least $1 - O \qty( 1/\secp )$, and (2) any computationally unbounded polynomial-query eavesdropper can find Bob's key with probability at most $1/2 + O(1/\secp)$.
\end{lemma}
\begin{proof}
Let $\sfG^\cU = (\sfG.\alice^\cU,\sfG.\bob^\cU)$ be a QCCC key agreement construction relative to $\cU$ that satisfies the given premises in which Alice and Bob each make $q = \poly(\secp)$ queries of length at most $L = \poly(\secp)$. Let $\ell(\secp) := \ceil{\log(q^4L^2 + \secp^2)} = O(\log\secp)$. Based on $\sfG$, we define the following short construction:

\myparagraph{Short construction $\sfH$ relative to $\cU$}
\protocol{$\sfH^\cU = (\sfH.\alice^\cU,\sfH.\bob^\cU)$}
{
on input $1^\secp$,
\begin{enumerate}
    \item $\sfH.\alice$ and $\sfH.\bob$ respectively sample a sequence of Haar random unitaries $\set{U_{\alice,\kappa}}_{\ell + 1 \leq \kappa \leq L}$ and \\
    $\set{U_{\bob,\kappa}}_{\ell + 1 \leq \kappa \leq L}$.
    \item $\sfH.\alice$ runs $\sfG.\alice^{(\cdot)}$ and responds to the queries as follows: if the length does not exceed $\ell$, then make a query to $\cU$ to reply; otherwise, simulate the oracle with $\set{U_{\alice,\kappa}}_{\ell + 1 \leq \kappa \leq L}$.
    \item $\sfH.\bob$ is defined analogously.
    \item They output whatever $\sfG.\alice$ and $\sfG.\bob$ output.
\end{enumerate}
}
In other words, $\sfH$ is identical to $\sfG$, except that any query by Alice or Bob exceeding length $\ell$ is answered with locally-simulated \emph{independent} Haar random unitaries.

\myparagraph{Completeness of $\sfH$} We use $(\sfG.\alice,\sfG.\bob)$ to construct a two-party adversary $(\HUD.\alice,\HUD.\bob)$ in~\Cref{thm:adap_LOCC} as follows. First, $\HUD.\alice$ and $\HUD.\bob$ sample and agree on a set of Haar unitaries of length at most $\ell$ through classical communication. Then $(\HUD.\alice,\HUD.\bob)$ runs $(\sfG.\alice,\sfG.\bob)$ by simulating oracles with their oracle access. Suppose $(\HUD.\alice,\HUD.\bob)$ are given identical Haar oracles, their execution is equivalent to $\sfG$; otherwise, it follows $\sfH$. Applying a standard hybrid over index $\kappa$, we have 
\begin{align*}
& \Pr_{\cU} \qty[ k_\alice = k_\bob:\, 
(k_\alice, k_\bob, \tau) \gets \inner{\sfH.\alice^\cU(1^\secp)}{\sfH.\bob^\cU(1^\secp)} ] \\
& \geq \Pr_{\cU} \qty[ k_\alice = k_\bob:\, 
(k_\alice, k_\bob, \tau) \gets \inner{\sfG.\alice^\cU(1^\secp)}{\sfG.\bob^\cU(1^\secp)} ]
- \sum_{\kappa = \ell +1}^L O\qty( \frac{q^2}{2^\kappa} ) \\
& \geq 1 - \negl(\secp) - O\qty( \frac{q^2L}{2^\ell} ) \\
& = 1 - \negl(\secp) - O\qty( \frac{1}{q^2L + \frac{\secp^2}{q^2L}} ) \\
& \geq 1 - \negl(\secp) - O\qty( \frac{1}{\secp} ) 
= 1 - O\qty( \frac{1}{\secp} ).
\end{align*}

\myparagraph{Security of $\sfH$} Let $p = \poly(\secp)$ and $\sfH.\eve^\cU$ be an eavesdropper that makes $p$ queries. Since $\sfH.\alice$ and $\sfH.\bob$ never make queries of length greater than $\ell$, we can assume that $\sfH.\eve^\cU$ does not either, without reducing its advantage. Consider the reduction $\sfG.\eve^\cU$:
\protocol{$\sfG.\eve^\cU$}
{
on input $1^\secp$ and transcript $\tau \gets \inner{\sfG.\alice^\cU(1^\secp)}{\sfG.\bob^\cU(1^\secp)}$,
  \begin{enumerate}
    \item $\sfG.\eve$ runs $\sfH.\eve^{(\cdot)}(1^\secp,\tau)$, responding to queries using $\cU$.
    \item $\sfG.\eve$ outputs whatever $\sfH.\eve$ outputs.
  \end{enumerate}
}

\noindent Similar to the proof of completeness of $\sfH$, we use $\sfG.\alice$, $\sfG.\bob$, and $\sfH.\eve$ to construct a two-party adversary $(\HUD.\alice,\HUD.\bob)$ in~\Cref{thm:adap_LOCC} as follows. First, $\HUD.\alice$ and $\HUD.\bob$ sample and agree on a set of Haar unitaries of length at most $\ell$ through classical communication. Then $(\HUD.\alice,\HUD.\bob)$ runs $(\sfG.\alice,\sfG.\bob)$ by simulating oracles with their oracle access. Finally, $\HUD.\bob$ runs $\sfG.\eve$ and checks whether $\sfG.\eve$ outputs $\sfG.\bob$'s key correctly. Notice that $\HUD.\bob$ does not need to make any queries to simulate $\sfH.\eve$ since $\sfH.\eve$ only makes queries of length at most $\ell$. Suppose $(\HUD.\alice,\HUD.\bob)$ are given identical Haar oracles, their execution is equivalent to the security game played by $\sfG.\eve$; otherwise, it follows the security game played by $\sfH.\eve$. Hence, the probability of $\sfH.\sfE^{\cU}$ finding Bob's key satisfies 
\begin{align*}
& \Pr_{\cU} \qty[ k_\eve = k_\bob:\, 
\substack{
    (k_\alice, k_\bob, \tau) \gets \inner{\sfH.\alice^\cU(1^\secp)}{\sfH.\bob^\cU(1^\secp)}, \\
    k_\eve \gets \sfH.\eve^\cU(1^\secp,\tau)
} ] \\
\leq & \Pr_{\cU} \qty[ k_\eve = k_\bob:\, 
\substack{
    (k_\alice, k_\bob, \tau) \gets \inner{\sfG.\alice^\cU(1^\secp)}{\sfG.\bob^\cU(1^\secp)}, \\
    k_\eve \gets \sfG.\eve^\cU(1^\secp,\tau)
} ]  
+ O\qty( \frac{1}{\secp} ) \\
\leq & \frac{1}{2} + \negl(\secp) + O\qty( \frac{1}{\secp} )
= \frac{1}{2} + O\qty( \frac{1}{\secp} ).
\end{align*}
This proves the security of $\sfH$.
\end{proof}

Now, we are ready to prove~\Cref{lem:nexist_KA}. For convenience, we restate the lemma below:
\begin{lemma}[\Cref{lem:nexist_KA}, restated]
Relative to $\cU$, for any construction of QCCC key agreements that makes only \emph{forward} queries to $\cU$ and satisfy completeness, there exists an eavesdropper, not necessarily time-efficient, who makes a polynomial number of queries and breaks the security.
\end{lemma}
\begin{proof}[Proof of~\Cref{lem:nexist_KA}]
For the sake of contradiction, suppose there was such a key agreement construction. However, by~\Cref{lem:reduce_to_shortKA}, it implies a short key agreement construction that contradicts~\Cref{lem:break_shortKA}.
\end{proof}


\subsection{Proof of~\Cref{lem:nexist_int_com}}
\label{app:lem:nexist_int_com}

\myparagraph{A general property of QCCC commitments} Suppose $\sfG = (\sen,\rec)$ is an interactive QCCC commitment construction in the plain model. Consider the following ``semi-honest'' statistical receiver $\rec^\ddagger$ and sender $\sen^\ddagger$. Receiver $\rec^\ddagger$ runs the commit phase honestly and outputs the more likely input bit conditioned on the transcript $\tau$; sender $\sen^\ddagger$ runs the commit phase honestly on a random input, and upon receiving the challenge bit $ch$, runs the honest sender in the reveal phase with the initial state $\rho^{ch,\tau}$. By inspection, their advantages for breaking statistical hiding and statistical sum binding are
\begin{align}
\label{eq:adv(rec)_plain}
\adv(\rec^\ddagger) 
& = \sum_{\tau \in \cT} \Pr[ \tau  ] \cdot \frac{1}{2} \cdot \bigg| \Pr[ b = 0 \mid \tau ] - \Pr[ b = 1 \mid \tau ] \bigg| \\
\label{eq:adv(sen)_plain}
\adv(\sen^\ddagger) 
& = \sum_{\tau \in \cT} \Pr[ \tau  ] \cdot \frac{1}{2} \cdot \Bigg(
\sum_{ch \in \bit} \Pr \left[
    \mu = ch:\,
    \mu \gets \reveal \inner{\sen(1^\secp,\,\rho^{ch,\,\tau}_{\reg{S}})}{\rec(1^\secp,\,\sigma^\tau_{\reg{R}})}
\right] - 1 \Bigg),
\end{align}
where the probability is over an honest execution of $\sfG$ on a random input bit $b$.

\noindent Next, suppose $\sfG$ has completeness error $\eta$, which by definition can be written as
\begin{align}
\label{eq:completeness_plain}
\eta & = 1 - \Pr \left[
    \mu = b:\, 
    \substack{ 
    b \gets \bit, \\
    (\rho^{b,\,\tau}_{\reg{S}},\,\sigma^{\tau}_{\reg{R}},\,\tau) \gets \commit \inner{ \sen(1^\secp,\,b)}{\rec(1^\secp)}, \\
    \mu \gets \reveal \inner{\sen(1^\secp,\,\rho^{b,\,\tau}_{\reg{S}})}{\rec(1^\secp,\,\sigma^{\tau}_{\reg{R}})}
    }
\right] \nonumber \\
& = 1 - \sum_{\tau\in\cT} \Pr[\tau] \cdot \sum_{b\in\bit} \Pr[b \mid \tau] \cdot \left[
    \mu = b:\
    \mu \gets \reveal \inner{\sen(1^\secp,\,\rho^{b,\,\tau}_{\reg{S}})}{\rec(1^\secp,\,\sigma^{\tau}_{\reg{R}})}
\right]
\end{align}

\noindent We present the following \emph{quantitative} version of the well-known impossibility of achieving both statistically-hiding and statistically-binding. 
\begin{lemma}
\label{lem:Imp_COM_plain}
$\adv(\rec^\ddagger) + \adv(\sen^\ddagger) \geq 1/2 - \eta$.
\end{lemma}
\begin{proof}
For any fixed $\tau$, we claim that
\begin{align}
\label{eq:tradeoff}
& \frac{1}{2} \cdot \bigg| \Pr[ b = 0 \mid \tau ] - \Pr[ b = 1 \mid \tau ] \bigg| \nonumber \\
& \quad + \frac{1}{2} \cdot \Bigg(
\sum_{ch \in \bit} \Pr \left[
    \mu = ch:\,
    \mu \gets \reveal \inner{\sen(1^\secp,\,\rho^{ch,\,\tau}_{\reg{S}})}{\rec(1^\secp,\,\sigma^\tau_{\reg{R}})}
\right] - 1 \Bigg) \nonumber \\
& \geq \Pr[b = 0 \mid \tau] \cdot \left[
    \mu = 0:\
    \mu \gets \reveal \inner{\sen(1^\secp,\,\rho^{0,\,\tau}_{\reg{S}})}{\rec(1^\secp,\,\sigma^{\tau}_{\reg{R}})}
\right] \nonumber \\ 
& \quad + \Pr[b = 1 \mid \tau] \cdot \left[
    \mu = 1:\
    \mu \gets \reveal \inner{\sen(1^\secp,\,\rho^{1,\,\tau}_{\reg{S}})}{\rec(1^\secp,\,\sigma^{\tau}_{\reg{R}})}
\right] - \frac{1}{2}.
\end{align}
Assuming~\Cref{eq:tradeoff} and averaging over $\tau$, we will obtain
\[
\adv(\rec^\ddagger) + \adv(\sen^\ddagger) 
\geq 1 - \eta - \frac{1}{2} = \frac{1}{2} - \eta,
\]
which completes the proof. Hence, it remains to prove~\Cref{eq:tradeoff}. Using the shorthand notation
\begin{align*}
    p & := \Pr[ b = 0 \mid \tau ], \\
    1 - p & \  = \Pr[ b = 1 \mid \tau ], \\
    q & := \Pr\left[
        \mu = 0:\
        \mu \gets \reveal \inner{\sen(1^\secp,\,\rho^{0,\,\tau}_{\reg{S}})}{\rec(1^\secp,\,\sigma^{\tau}_{\reg{R}})}
    \right], \\
    r & := \Pr\left[
        \mu = 1:\
        \mu \gets \reveal \inner{\sen(1^\secp,\,\rho^{1,\,\tau}_{\reg{S}})}{\rec(1^\secp,\,\sigma^{\tau}_{\reg{R}})}
    \right],
\end{align*}
we can equivalently express~\Cref{eq:tradeoff} as 
\begin{align*}
    \frac{|2p - 1|}{2} + \frac{q + r}{2} 
    \geq pq + (1 - p) r - \frac{1}{2},
\end{align*}
which can be verified through elementary arithmetic.
\end{proof}

From~\Cref{lem:Imp_COM_plain}, if the completeness error $\eta$ is negligible, then either $\rec^\ddagger$ or $\sen^\ddagger$ (or possibly both) can break the security. \\

\noindent Now, suppose $G^\cU$ is an interactive QCCC commitment relative to an oracle $\cU$. For any (fixed) oracle $U$, in the spirit of~\Cref{eq:completeness_plain,eq:adv(rec)_plain,eq:adv(sen)_plain}, we define the following quantities over an honest execution of $G^U$ with the oracle being fixed to $U$ on a random input:
\begin{align}
\label{eq:completeness_oracle}
& \eta_U 
:=  1 - \sum_{\tau\in\cT} \Pr[\tau \mid U] \sum_{b\in\bit} \Pr[b \mid U,\,\tau] \cdot \left[
    \mu = b:\
    \mu \gets \reveal \inner{\sen(1^\secp,\,\rho^{b,\,U,\,\tau}_{\reg{S}})}{\rec(1^\secp,\,\sigma^{U,\,\tau}_{\reg{R}})}
\right], \\
\label{eq:adv(rec)_oracle}
& \adv_U(\rec^\ddagger) 
:= \sum_{\tau \in \cT} \Pr[ \tau \mid U ] \cdot \frac{1}{2} \cdot \qty| \Pr[ b = 0 \mid U,\,\tau ] - \Pr[ b = 1 \mid U,\,\tau ] |, \\
\label{eq:adv(sen)_oracle}
& \adv_U(\sen^\ddagger) 
:= \sum_{\tau \in \cT} \Pr[ \tau \mid U ] \cdot \frac{1}{2} \cdot \Bigg(
\sum_{ch \in \bit} \Pr \left[
    \mu = ch:\,
    \mu \gets \reveal \inner{\sen(1^\secp,\,\rho^{ch,\,U,\,\tau}_{\reg{S}})}{\rec(1^\secp,\,\sigma^{U,\,\tau}_{\reg{R}})}
\right] - 1 \Bigg).
\end{align}

\noindent Using the same argument in the proof of~\Cref{lem:Imp_COM_oracle}, we obtain the following lemma:
\begin{lemma}
\label{lem:Imp_COM_oracle}
For any fixed choice of the oracle $U$, it holds that $\adv_U(\rec^\ddagger) + \adv_U(\sen^\ddagger) \geq 1/2 - \eta_U$.
\end{lemma}

\noindent Similarly, we can view $\eta_U$ as the completeness error conditioned on the oracle being $U$. As for $\adv_U(\rec^\ddagger)$ (\resp $\adv_U(\sen^\ddagger)$), it corresponds to $\rec^\ddagger$'s (\resp $\sen^\ddagger$'s) advantage if they have \emph{full knowledge} of oracle $U$, allowing them to (statistically) compute the conditional distribution and the internal states. 

\begin{definition}[Short interactive QCCC commitment constructions]
An interactive QCCC commitment construction relative to $\cU$ is \emph{short} if the maximum length of queries made by the sender and receiver is $O(\log\secp)$.
\end{definition}

\begin{lemma}[Breaking any short interactive QCCC commitment construction]
\label{lem:Imp_short_query_COM}
For any $\eta:\N \to [0,1]$ and any short construction of interactive QCCC commitment relative to $\cU$ that has completeness $1 - \eta$, there exists either 
\begin{enumerate}
    \item a malicious receiver (not necessarily time-efficient) who makes a polynomial number of queries and guesses the input with an advantage at least $0.15 - \eta/3$, or
    \item a malicious sender (not necessarily time-efficient) who makes a polynomial number of queries and opens the challenge $ch$ with an advantage at least $0.15 - \eta/3$
\end{enumerate}
(or possibly both).
\end{lemma}
\begin{proof}
Let $\sfG^\cU = (\sen^\cU, \rec^\cU))$ be a short construction of interactive QCCC commitments that satisfies completeness $1 - \eta$, where $\sen$ and $\rec$ each make $q = \poly(\secp)$ queries of length at most $\ell = O(\log\secp)$.

\myparagraph{Attacking hiding} Define the following malicious receiver $\rec^*$:
\protocol{$\rec^{*,\,\cU}$}
{
on input $1^\secp$,
  \begin{enumerate}
    \item Run $\rec^{(\cdot)}$ honestly, simulating oracle calls to $\set{U_\kappa}_{\kappa\in[\ell]}$ by querying $\cU$ until the commit phase completes, and denote the resulting transcript by $\tau$.
    \item Perform the process tomography algorithm in~\Cref{thm:tomography} with $\delta = \eps = 1/(1200q\ell)$ on each $\set{U_\kappa}_{\kappa\in[\ell]}$, and obtain the classical description of the unitaries $\set{ \wt{U}_\kappa }_{\kappa\in[\ell]}$.
    \item Sample $b'$ conditioned on $(\set{ \wt{U}_\kappa }_{\kappa\in[\ell]},\tau)$, and output $b'$.
  \end{enumerate}
}
\noindent The total number of queries made by $\rec^*$ in steps~1 and~2 is polynomial. We fix the oracle $U = $ $\set{U_\kappa}_{\kappa\in[\ell]}$ and $\wt{U} := $ $\set{ \wt{U}_\kappa }_{\kappa\in[\ell]}$ obtained in step~2. Denote the advantage of $\rec^*$ conditioned on $U$ and $\wt{U}$ by
\begin{align*}
\adv_{U, \wt{U}}(\rec^*) := \Pr[ b' = b \mid U,\,\wt{U} ] - \frac{1}{2},
\end{align*}
where the probability is over the hiding security game played by $\rec^*$. Let $\bfD_U$ denote the joint distribution of $(b,\tau)$ over an honest execution of $\sfG^U$ conditioned on the oracle being $U$. Letting $\bfP_{BT} = \bfD_U$ and $\bfQ_{BT} = \bfD_{\wt{U}}$ in~\Cref{lem:SD_lemma}, we obtain
\begin{align}
\label{eq:exp0_exp1}
\Pr_{\mathbf{Exp.1}}[b = b'] 
\geq \Pr_{\mathbf{Exp.0}}[b = b'] - 3 \cdot \SD(\bfD_U,\,\bfD_{\wt{U}}).
\end{align}
It is not hard to see that the joint distribution of $(b,\,b',\,\tau)$ in $\mathbf{Exp.1}$ in~\Cref{lem:SD_lemma} is identically distributed to that in the hiding security game played by $\rec^*$. Similarly, the joint distribution of $(b,\,b',\,\tau)$ in $\mathbf{Exp.0}$ in~\Cref{lem:SD_lemma} is identically distributed to that in the hiding security game played by $\rec^*$ if $\wt{U} = U$, that is, in the absence of any tomography error. Expressing~\Cref{eq:exp0_exp1} in terms of distinguishing advantage with the notation in~\Cref{eq:adv(rec)_oracle}, we obtain
\begin{align}
\label{eq:adv_rec}
\adv_{U,\wt{U}}(\rec^*) 
& \geq \adv_U(\rec^\ddagger) - 3 \cdot \SD(\bfD_U,\,\bfD_{\wt{U}}) \nonumber \\
& \geq \adv_U(\rec^\ddagger) - 6q \cdot \sum_{\kappa = 1}^\ell \norm{ U_\kappa(\cdot)U_\kappa^\dagger - \wt{U}_\kappa(\cdot)\wt{U}_\kappa^\dagger }_\diamond.
\end{align}
by~\Cref{fact:diamondnorm} and viewing $\sfG$ as a $2q$-query algorithm. Averaging over $\wt{U}$ in~\Cref{eq:adv_rec} and applying~\Cref{thm:tomography} yields
\begin{align}
\label{eq:adv_rec_final}
\adv_U(\rec^*)
\geq \adv_U(\rec^\ddagger) - 6q\ell \cdot (\veps + \delta)
= \adv_U(\rec^\ddagger) - 0.01,
\end{align}
where $\adv_U(\rec^*) := \Ex_{\wt{U}}\qty[ \adv_{U,\wt{U}}(\rec^*) ]$ denotes the advantage of $\rec^*$ conditioned on $U$.

\myparagraph{Attacking binding} Define the following malicious sender $\sen^*$:
\protocol{$\sen^{*,\,\cU}$}
{
on input $1^\secp$,
\begin{enumerate}
    \item In the commit phase, run $\sen^{(\cdot)}$ honestly on a random input, simulating oracle calls to $\set{U_\kappa}_{\kappa\in[\ell]}$ by querying $\cU$ until the commit phase completes, and denote the resulting transcript by $\tau$.
    \item Receive a random challenge bit $ch$ to open to.
    \item Perform the process tomography algorithm in~\Cref{thm:tomography} with $\delta = \eps = 1/(1200q\ell)$ on each $\set{U_\kappa}_{\kappa\in[\ell]}$, and obtain the classical description of the unitaries $\wt{U} := \set{ \wt{U}_\kappa }_{\kappa\in[\ell]}$.
    \item Compute the state $\rho^{ch,\wt{U},\tau}_{\reg{S}}$.
    \item In the reveal phase, run $\sen^{(\cdot)}$ honestly on input $\rho^{ch,\wt{U},\tau}_{\reg{S}}$, simulating oracle calls to $\set{U_\kappa}_{\kappa\in[\ell]}$ by querying $\cU$ until the reveal phase completes.
\end{enumerate}
}
\noindent The total number of queries made by $\sen^*$ in steps~1 and~5 is polynomial. We fix the oracle $U = $ $\set{U_\kappa}_{\kappa\in[\ell]}$ and $\wt{U} := $ $\set{ \wt{U}_\kappa }_{\kappa\in[\ell]}$ obtained in step~3. Denote the advantage of $\sen^*$ conditioned on $U$ and $\wt{U}$ by
\begin{align*}
\adv_{U, \wt{U}}(\sen^*) := \Pr[ \mu = ch \mid U,\,\wt{U} ] - \frac{1}{2},
\end{align*}
where the probability is over the binding security game played by $\sen^*$. We have the following claim:

\begin{myclaim}
\label{claim:binding_lemma}
Let $\bfT|_{b,\,U}$ (\resp $\bfT|_{b,\,\wt{U}}$) denote the conditional distribution of transcript $\tau$ over an honest execution of $\sfG^U$ conditioned on the input being $b \in \bit$ and the oracle being $U$ (\resp $\wt{U}$). Then for $b \in \bit$,
\begin{align*}
\Ex_{\bfT|_{b,\,U}} \qty[ \TD \qty( \rho^{b,\,U,\,\tau}_{\reg{S}},\, \rho^{b,\,\wt{U},\,\tau}_{\reg{S}}  ) ]
\leq 4q\ell \cdot \norm{ U(\cdot)U^\dagger - \wt{U}(\cdot)\wt{U}^\dagger }_\diamond.
\end{align*}
\end{myclaim}
\begin{proof}[Proof of~\Cref{claim:binding_lemma}]
Consider the following two classical-quantum states corresponding to the final states of an honest execution of the commit phase of $\sfG^U$ and $\sfG^{\textcolor{red}{\wt{U}}}$ on input $b$, respectively (for convenience, we color $\textcolor{red}{\wt{U}}$ in \textcolor{red}{red} throughout the proof):
\begin{align*}
  \Ex_{\bfT|_{b,\,U}}[ \rho^{b,\,U,\,\tau} \otimes \sigma^{U,\,\tau} ]  
  & = \sum_{\tau \in \cT} \Pr_{\bfT|_{b,\,U}} [\tau] \cdot \projector{\tau}_{\reg{T}} \otimes \rho^{b,\,U,\,\tau}_{\reg{S}} \otimes \sigma^{U,\,\tau}_{\reg{R}} \\
  \Ex_{\bfT|_{b,\,\textcolor{red}{\wt{U}}}}[ \rho^{b,\,{\textcolor{red}{\wt{U}}},\,\tau} \otimes \sigma^{{\textcolor{red}{\wt{U}}},\,\tau} ]  
  & = \sum_{\tau \in \cT} \Pr_{\bfT|_{b,\,\textcolor{red}{\wt{U}}}}[\tau] \cdot \projector{\tau}_{\reg{T}} \otimes \rho^{b,\,{\textcolor{red}{\wt{U}}},\,\tau}_{\reg{S}} \otimes \sigma^{{\textcolor{red}{\wt{U}}},\,\tau}_{\reg{R}}.
\end{align*}
Since they are the final states of the honestly-executed $(\sen,\rec)$, which together form a $2q$-query algorithm, the trace distance between them is at most $2q\ell \cdot \norm{ U(\cdot)U^\dagger - \wt{U}(\cdot)\wt{U}^\dagger }_\diamond$ by~\Cref{fact:diamondnorm}. Taking partial trace over $\reg{R}$ on both states, the trace distance between
\begin{align*}
  \Ex_{\bfT|_{b,\,U}}[ \rho^{b,\,U,\,\tau} ]  
  = \sum_{\tau \in \cT} \Pr_{\bfT|_{b,\,U}} [\tau] \cdot \projector{\tau}_{\reg{T}} \otimes \rho^{b,\,U,\,\tau}_{\reg{S}} \quad \text{and} \quad 
  \Ex_{\bfT|_{b,\,\textcolor{red}{\wt{U}}}}[ \rho^{b,\,{\textcolor{red}{\wt{U}}},\,\tau} ]  
  = \sum_{\tau \in \cT} \Pr_{\bfT|_{b,\,\textcolor{red}{\wt{U}}}}[\tau] \cdot \projector{\tau}_{\reg{T}} \otimes \rho^{b,\,{\textcolor{red}{\wt{U}}},\,\tau}_{\reg{S}}
\end{align*}
satisfies
\begin{align}
\label{eq:TD1}
    \TD \qty( \Ex_{\bfT|_{b,\,U}}[ \rho^{b,\,U,\,\tau} ], \Ex_{\bfT|_{b,\,\textcolor{red}{\wt{U}}}}[ \rho^{b,\,{\textcolor{red}{\wt{U}}},\,\tau} ] )
    \leq 2q\ell \cdot \norm{ U(\cdot)U^\dagger - \wt{U}(\cdot)\wt{U}^\dagger }_\diamond
\end{align}
since trace distance cannot increase under partial trace. Define the following ``hybrid'' state
\begin{align*}
\Ex_{\bfT|_{b,\,U}}[ \rho^{b,\,\textcolor{red}{\wt{U}},\,\tau} ]  
  & = \sum_{\tau \in \cT} \Pr_{\bfT|_{b,\,U}}[\tau] \cdot \projector{\tau}_{\reg{T}} \otimes \rho^{b,\,\textcolor{red}{\wt{U}},\,\tau}_{\reg{S}}.
\end{align*}
Notice that
\begin{align}
\label{eq:TD2}
    \TD \qty( \Ex_{\bfT|_{b,\,U}}[ \rho^{b,\,U,\,\tau} ], 
    \Ex_{\bfT|_{b,\,U}}[ \rho^{b,\,\textcolor{red}{\wt{U}},\,\tau} ] ) 
    = \Ex_{\bfT|_{b,\,U}} \qty[ \TD \qty( \rho^{b,\,U,\,\tau}_{\reg{S}},\, \rho^{b,\,\wt{U},\,\tau}_{\reg{S}}  ) ].
\end{align}
and
\begin{align}
\label{eq:TD3}
    \TD \qty( \Ex_{\bfT|_{b,\,U}}[ \rho^{b,\,\textcolor{red}{\wt{U}},\,\tau} ],
    \Ex_{\bfT|_{b,\,\textcolor{red}{\wt{U}}}}[ \rho^{b,\,{\textcolor{red}{\wt{U}}},\,\tau} ] ) 
    = \SD( \bfT|_{b,\,U}, \bfT|_{b,\,\textcolor{red}{\wt{U}}} )
    \leq 2q\ell \cdot \norm{ U(\cdot)U^\dagger - \wt{U}(\cdot)\wt{U}^\dagger }_\diamond.
\end{align}
Applying the triangle inequality on~\Cref{eq:TD1,eq:TD2,eq:TD3} completes the proof of~\Cref{claim:binding_lemma}.
\end{proof}

\noindent The winning probability of $\sen^*$ is 
\begin{align*}
& \Pr[ \mu = ch \mid U,\,\wt{U} ] \\
& = \frac{1}{4} \cdot \sum_{b,\,ch,\in\bit,\,\tau \in \cT} \Pr[ \tau \mid b,\,U]
\cdot \Pr \left[
    \mu = ch:\,
    \mu \gets \reveal \inner{\sen(1^\secp,\,\rho^{ch,\,\wt{U},\,\tau}_{\reg{S}})}{\rec(1^\secp,\,\sigma^{U,\,\tau}_{\reg{R}})}
\right] \\
& \geq \frac{1}{4} \cdot \sum_{\tau \in \cT} \sum_{b \in \bit} \Pr[ \tau \mid b,\,U] \\
& \cdot \sum_{ch \in \bit}
\Bigg( 
\Pr \left[
    \mu = ch:\,
    \mu \gets \reveal \inner{\sen(1^\secp,\,\rho^{ch,\,U,\,\tau}_{\reg{S}})}{\rec(1^\secp,\,\sigma^{U,\,\tau}_{\reg{R}})} 
    \right] 
- \TD \qty( \rho^{ch,\,U,\,\tau}_{\reg{S}},\, \rho^{ch,\,\wt{U},\,\tau}_{\reg{S}} )
\Bigg),
\end{align*}
where the inequality follows from the operational meaning of trace distance. Thus, it allows us to write
\begin{align*}
& \Pr[ \mu = ch \mid U ] \\
& = \sum_{\tau \in \cT} \Pr[ \tau \mid U ] \cdot \frac{1}{2} \cdot 
\sum_{ch \in \bit} \Pr \left[
    \mu = ch:\,
    \mu \gets \reveal \inner{\sen(1^\secp,\,\rho^{ch,\,U,\,\tau}_{\reg{S}})}{\rec(1^\secp,\,\sigma^{U,\,\tau}_{\reg{R}})}
\right] \\
& - \frac{1}{4} \sum_{b \in \bit} \Ex_{\bfT|_{b,\,U}} \qty[ \TD \qty( \rho^{b,\,U,\,\tau}_{\reg{S}},\, \rho^{b,\,\wt{U},\,\tau}_{\reg{S}}  ) ] 
- \frac{1}{4} \sum_{b \in \bit} \Ex_{\bfT|_{1 - b,\,U}} \qty[ \TD \qty( \rho^{b,\,U,\,\tau}_{\reg{S}},\, \rho^{b,\,\wt{U},\,\tau}_{\reg{S}}  ) ].
\end{align*}
The second term can be upper bounded by~\Cref{claim:binding_lemma}. To bound the last term, we rely on the following lemma:

\begin{lemma}
\label{lem:expectation_lemma}
Let $\bfD_0$ and $\bfD_1$ be distributions over a finite set $\cT$, and let $f_0,\,f_1:\cT \to [0,1]$ be functions. Then
\begin{align*}
\sum_{b \in \bit} \Ex_{\tau \sim \bfD_{1-b}} \qty[ f_b(\tau) ] 
\leq \sum_{b \in \bit} \Ex_{\tau \sim \bfD_b} \qty[ f_b(\tau) ] 
+ 2 \cdot \SD(\bfD_0,\,\bfD_1).
\end{align*}
\end{lemma}
\begin{proof}[Proof of~\Cref{lem:expectation_lemma}]
Subtracting the LHS from the RHS yields
\begin{align*}
    \sum_{\tau \in \cT} \bigg( \bfD_0(\tau) \cdot (f_0(\tau) - f_1(\tau)) + \bfD_1(\tau) \cdot (f_1(\tau) - f_0(\tau)) + | \bfD_0(\tau) - \bfD_1(\tau) | \bigg).
\end{align*}
For a fixed $\tau$, suppose $\bfD_0(\tau) \geq \bfD_1(\tau)$, then
\begin{align*}
    & \bfD_0(\tau) \cdot (f_0(\tau) - f_1(\tau)) + \bfD_1(\tau) \cdot (f_1(\tau) - f_0(\tau)) + | \bfD_0(\tau) - \bfD_1(\tau) | \\
    = & \bfD_0(\tau) \cdot (1 + f_0(\tau) - f_1(\tau)) - \bfD_1(\tau) \cdot (1 + f_0(\tau) - f_1(\tau)) \\
    \geq & 0
\end{align*}
since $f_0(\tau), f_1(\tau) \in [0,1]$. The other case, where $\bfD_0(\tau) < \bfD_1(\tau)$, can be proved similarly. This completes the proof of~\Cref{lem:expectation_lemma}.
\end{proof}

\noindent Therefore, we obtain
\begin{align*}
& \Pr[ \mu = ch \mid U,\,\wt{U} ] = \\
& \sum_{\tau \in \cT} \Pr[ \tau \mid U ] \cdot \frac{1}{2} \cdot 
\sum_{b \in \bit} \Pr \left[
    \mu = b:\,
    \mu \gets \reveal \inner{\sen(1^\secp,\,\rho^{b,\,U,\,\tau}_{\reg{S}})}{\rec(1^\secp,\,\sigma^{U,\,\tau}_{\reg{R}})}
\right] \\
& \quad - 4q\ell \cdot \norm{ U(\cdot)U^\dagger - \wt{U}(\cdot)\wt{U}^\dagger }_\diamond
- \frac{1}{2} \cdot \SD( \bfT|_{0,\,U},\,\bfT|_{1,\,U} ),
\end{align*}
or in terms of advantage with the notation in~\Cref{eq:adv(sen)_oracle},
\begin{align}
\label{eq:adv_sen}
\adv_{U, \wt{U}}(\sen^*) 
& \geq \adv_{U}(\sen^\ddagger) - 4q\ell \cdot \norm{ U(\cdot)U^\dagger - \wt{U}(\cdot)\wt{U}^\dagger }_\diamond - \frac{1}{2} \cdot \SD( \bfT|_{0,\,U},\,\bfT|_{1,\,U} ) \nonumber \\
& = \adv_{U}(\sen^\ddagger) - 4q\ell \cdot \norm{ U(\cdot)U^\dagger - \wt{U}(\cdot)\wt{U}^\dagger }_\diamond - \adv_{U}(\rec^\ddagger),
\end{align}
since one can verify that $\adv_{U}(\rec^\ddagger) = \SD( \bfT|_{0,\,U},\,\bfT|_{1,\,U} )/2$. Averaging over $\wt{U}$ in~\Cref{eq:adv_sen} yields
\begin{align}
\label{eq:adv_sen_final}
\adv_U(\sen^*)
& \geq \adv_{U}(\sen^\ddagger) - \adv_{U}(\rec^\ddagger) - 0.01
\end{align}
where $\adv_U(\sen^*) := \Ex_{\wt{U}}\qty[ \adv_{U,\wt{U}}(\sen^*) ]$ denotes the advantage of $\sen^*$ conditioned on $U$.

\myparagraph{Wrapping up} From~\Cref{eq:adv_rec_final,eq:adv_sen_final} and~\Cref{lem:Imp_COM_oracle}, after averaging over $U$, we obtain
\begin{align*}
\adv(\rec^*) & \geq \adv(\rec^\ddagger) - 0.01, \\
\adv(\sen^*) & \geq \adv(\sen^\ddagger) - \adv(\rec^\ddagger) - 0.01, \\
\adv(\rec^\ddagger) + \adv(\sen^\ddagger) & \geq 0.5 - \eta,
\end{align*}
where $\adv(\sen^*) := \Ex_{U}\qty[ \adv_{\wt{U}}(\sen^*) ]$. The other three terms are defined analogously. We complete the proof using a case analysis. If $\adv(\rec^\ddagger) \geq 0.16 - \eta/3$, then $\adv(\rec^*) \geq 0.15 - \eta/3$. Otherwise, we have
\begin{align*}
\adv(\sen^*) & \geq \adv(\sen^\ddagger) - \adv(\rec^\ddagger) - 0.01 \\
& \geq \qty(0.5 - \eta - \adv(\rec^\ddagger)) - \adv(\rec^\ddagger) - 0.01 \\
& \geq 0.15 - \eta/3.
\end{align*}
Hence, at least one of $\adv(\rec^*)$ and $\adv(\sen^*)$ is at least $0.15 - \eta/3$.
\end{proof}

\begin{lemma} 
\label{lem:reduce_to_shortCOM}
Suppose there exists an interactive QCCC commitment relative to $\cU$ such that (1) the completeness error is at most $\negl(\secp)$, and (2) any computationally unbounded polynomial-query malicious receiver (\resp sender) can break statistical hiding (\resp sum binding) with advantage at most $\negl(\secp)$. Then there exists a short interactive QCCC commitment relative to $\cU$ such that: (1) the sender and receiver each make a polynomial number of queries, but are not necessarily time-efficient (2) the completeness error is at most $O \qty( 1/\secp )$, and (2) any computationally unbounded polynomial-query malicious receiver (\resp sender) can break statistical hiding (\resp sum binding) with advantage at most $\negl(\secp)$.
\end{lemma}
\begin{proof}
Let $\sfG^\cU = (\sfG.\sen^\cU,\sfG.\rec^\cU)$ be an interactive QCCC commitment construction relative to $\cU$ that satisfies the given premises in which the sender and receiver each make $q = \poly(\secp)$ queries of length at most $L = \poly(\secp)$. Let $\ell(\secp) := \ceil{\log(q^4L^2 + \secp^2)} = O(\log\secp)$. Based on $\sfG$, we define the following short construction:

\myparagraph{Short construction $\sfH$ relative to $\cU$}
\protocol{$\sfH^\cU = (\sfH.\sen^\cU,\sfH.\rec^\cU)$}
{
on input $1^\secp$ and $b \in \bit$,
\begin{enumerate}
    \item $\sfH.\sen$ and $\sfH.\rec$ respectively sample a sequence of Haar random unitaries $\set{U_{\alice,\kappa}}_{\ell + 1 \leq \kappa \leq L}$ and \\
    $\set{U_{\bob,\kappa}}_{\ell + 1 \leq \kappa \leq L}$.
    \item $\sfH.\sen$ runs $\sfG.\sen^{(\cdot)}$ on input $b$ and responds to the queries as follows: if the length does not exceed $\ell$, then make a query to $\cU$ to reply; otherwise, simulate the oracle with $\set{U_{\alice,\kappa}}_{\ell + 1 \leq \kappa \leq L}$.
    \item $\sfH.\rec$ is defined analogously.
    \item $\sfH.\rec$ outputs whatever $\sfG.\rec$ outputs.
\end{enumerate}
}
In other words, $\sfH$ is identical to $\sfG$, except that any query by the sender or receiver exceeding length $\ell$ is answered with locally-simulated \emph{independent} Haar random unitaries.

\myparagraph{Completeness of $\sfH$} We use $(\sfG.\sen,\sfG.\rec)$ to construct a two-party adversary $(\HUD.\alice,\HUD.\bob)$ in~\Cref{thm:adap_LOCC} as follows. First, $\HUD.\alice$ and $\HUD.\bob$ sample and agree on a set of Haar unitaries of length at most $\ell$ through classical communication. Then $(\HUD.\alice,\HUD.\bob)$ runs $(\sfG.\sen,\sfG.\rec)$ by simulating oracles with their oracle access. Suppose $(\HUD.\alice,\HUD.\bob)$ are given identical Haar oracles, their execution is equivalent to $\sfG$; otherwise, it follows $\sfH$. By a similar calculation as in~\Cref{lem:reduce_to_shortKA}, the completeness of $\sfH$ is $1 - O\qty( 1/\secp )$.

\myparagraph{Hiding and binding of $\sfH$} Let $p = \poly(\secp)$ and $\sfH.\rec^{*,\,\cU}$ be a malicious receiver that makes $p$ queries. Since $\sfH.\sen$ never makes queries of length greater than $\ell$, we can assume that $\sfH.\rec^{*,\,\cU}$ does not either, without reducing its advantage. Consider the reduction $\sfG.\rec^{*,\,\cU}$:
\protocol{$\sfG.\rec^{*,\,\cU}$}
{
on input $1^\secp$,
  \begin{enumerate}
    \item $\sfG.\rec^*$ runs $\sfH.\rec^{*,\,(\cdot)}$, responding to queries using $\cU$.
    \item $\sfG.\rec^*$ outputs whatever $\sfH.\rec^*$ outputs.
  \end{enumerate}
}

\noindent By construction, $\sfG.\rec^*$ never makes queries of length greater than $\ell$. Hence, the distribution of the hiding security game played by $\sfG.\rec^*$ is identically distributed to that of the game played by $\sfH.\rec^*$, which implies that they have the same advantage. We can define $\sfG.\sen^*$ similarly and show that $\sfG.\sen^*$ and $\sfH.\sen^*$ have the same advantage.
\end{proof}

Now, we are ready to prove~\Cref{lem:nexist_int_com}. For convenience, we restate the lemma below:
\begin{lemma}[\Cref{lem:nexist_int_com}, restated]
Relative to $\cU$, for any construction of interactive QCCC key commitments that makes only \emph{forward} queries to $\cU$ and satisfy completeness, there exists either (1) a malicious receiver or (2) a malicious sender (or possibly both) who, while not necessarily time-efficient, makes a polynomial number of queries and breaks the security.
\end{lemma}
\begin{proof}[Proof of~\Cref{lem:nexist_int_com}]
For the sake of contradiction, suppose there was such a construction. However, by~\Cref{lem:reduce_to_shortCOM}, it implies a short interactive QCCC commitment construction that contradicts~\Cref{lem:Imp_short_query_COM}.
\end{proof}

\fi

\end{document}